\documentclass[reqno]{amsart}
\usepackage{fullpage,amsfonts,amssymb,amsthm,mathrsfs,graphicx,color,framed,esint,amsmath}
\usepackage{cancel,bm}
\usepackage{tikz}
\usetikzlibrary{decorations.markings}
\usepackage{booktabs,longtable}

\usepackage[initials]{amsrefs}
\usepackage{stackrel}

\usepackage[colorlinks=true, pdfstartview=FitV, linkcolor=blue, citecolor=blue, urlcolor=blue]{hyperref}
\definecolor{shadecolor}{rgb}{0.95, 0.95, 0.86}

\renewcommand{\d}{{\mathrm d}}

\newcommand{\im}{\mathrm{i}}
\newcommand{\e}{\mathrm{e}}

\def\tr{\mathop{\mathrm{tr}}\limits}

\def\X{{\bf X}}
\def\K{{\bf K}}
\def\Y{{\bf Y}}
\def\E{{\bf E}}
\def\L{{\bf L}}
\def\M{{\bf M}}

\def\T{{\bf T}}
\def\f{{\bf f}}

\def\F{{\bf F}}
\def\A{{\bf A}}

\numberwithin{equation}{section}

\newtheorem{theo}{Theorem}[section]

\newtheorem{lem}[theo]{Lemma}
\newtheorem{rem}[theo]{Remark}
\newtheorem{problem}[theo]{Riemann-Hilbert Problem}

\newtheorem{prop}[theo]{Proposition} 
\newtheorem{cor}[theo]{Corollary}

\begin{document}

\title[Thinned real Ginibre]{Edge distribution of thinned real eigenvalues in the real Ginibre ensemble}

\author{Jinho Baik}
\address{Department of Mathematics, University of Michigan, 2074 East Hall, 530 Church Street, Ann Arbor, MI 48109-1043, United States}
\email{baik@umich.edu}

\author{Thomas Bothner}
\address{School of Mathematics, University of Bristol, Fry Building, Woodland Road, Bristol, BS8 1UG, United Kingdom}
\email{thomas.bothner@bristol.ac.uk}



\thanks{The work of J.B. is supported in part by the NSF grants DMS-1664692 and DMS-1954790. T.B. acknowledges support by the Engineering and Physical Sciences Research Council through grant EP/T013893/1.}

\begin{abstract}
This paper is concerned with the explicit computation of the limiting distribution function of the largest real eigenvalue in the real Ginibre ensemble when each real eigenvalue has been removed independently with constant likelihood. We show that the recently discovered integrable structures in \cite{BB} generalize from the real Ginibre ensemble to its thinned equivalent. Concretely, we express the aforementioned limiting distribution function as a convex combination of two simple Fredholm determinants and connect the same function to the inverse scattering theory of the Zakharov-Shabat system. As corollaries, we provide a Zakharov-Shabat evaluation of the ensemble's real eigenvalue generating function and obtain precise control over the limiting distribution function's tails. The latter part includes the explicit computation of the usually difficult constant factors.
\end{abstract}

\date{\today}
\maketitle
\section{Introduction and statement of results}\label{sec:11} 
Let $\X\in\mathbb{R}^{n\times n},n\in\mathbb{Z}_{\geq 2}$ be a matrix whose entries are independent, identically distributed standard normal random variables with mean $0$ and variance $1$. In other words, $\X$ is a matrix drawn from the real Ginibre ensemble (GinOE) \cite{G}. It is known, cf. \cite{BS,FN,S}, that the eigenvalues $\{z_j(\X)\}_{j=1}^n$ of $\X$ form a Pfaffian point process, a fact which allows one to compute gap probabilities in the GinOE as Fredholm determinants. Of particular interest for us is the following result about the absence of real eigenvalues in $(t,\infty)\subset\mathbb{R}$.
\begin{prop}[{\cite[Proposition $2.2$]{RS}}]\label{recall} For every $n\in\mathbb{Z}_{\geq 2}$,
\begin{equation}\label{e:1}
	\mathbb{P}\left(\max_{\substack{j=1,\ldots,n\\ z_j\in\mathbb{R}}}z_j(\X)\leq t\right)=\sqrt{\det_2\big(1-\chi_t\K_n\chi_t{\upharpoonright}_{L^2(\mathbb{R})\oplus L^2(\mathbb{R})}\big)},\ \ \ \ t\in\mathbb{R},
\end{equation}
where $\chi_t$ is the operator of multiplication by the characteristic function $\chi_{[t,\infty)}$ of the interval $[t,\infty)\subset\mathbb{R}$ and $\K_n$ the following Hilbert-Schmidt integral operator on $L^2(\mathbb{R})\oplus L^2(\mathbb{R})$,
\begin{equation}\label{e:2}
	\K_n=\begin{bmatrix} \rho^{-1}S_n\rho & \rho^{-1}(DS_n^{\ast})\rho^{-1}\smallskip\\
	-\rho(IS_n)\rho+\rho\epsilon\rho & \rho S_n^{\ast}\rho^{-1}\end{bmatrix}.
\end{equation}
Here, $\rho$ multiplies by any differentiable, square-integrable weight function $\rho(x)>0$ on $\mathbb{R}$ such that $\rho^{-1}(x)\equiv 1/\rho(x)$ is polynomially bounded. Moreover $S_n$ and $\epsilon$ are the integral operators on $L^2(\mathbb{R})$ with kernels
\begin{equation*}
	S_n(x,y):=\frac{1}{\sqrt{2\pi}}\e^{-\frac{1}{2}(x^2+y^2)}\mathfrak{e}_{n-2}(xy)+\frac{x^{n-1}\e^{-\frac{1}{2}x^2}}{\sqrt{2\pi}\,(n-2)!}\int_0^yu^{n-2}\e^{-\frac{1}{2}u^2}\,\d u,\ \ \ 
	\epsilon(x,y):=\frac{1}{2}\textnormal{sgn}(y-x),
\end{equation*}
where $\mathfrak{e}_n(z):=\sum_{k=0}^n\frac{1}{k!}z^k$ is the exponential partial sum, $S_n^{\ast}$ the real adjoint of $S_n$, $D$ acts by differentiation on the independent variable and $IS_n$ has kernel 
\begin{equation*}
	\big(IS_n\big)(x,y):=\big(\epsilon S_n\big)(x,y),\ \ \ \ \ n\in\mathbb{Z}_{\geq 2}\ \textnormal{even},
\end{equation*}
and 
\begin{equation*}
	\big(IS_n\big)(x,y):=\big(\epsilon S_n\big)(x,y)+\frac{1}{2^{n/2}\Gamma(n/2)}\int_0^yu^{n-1}\e^{-\frac{1}{2}u^2}\,\d u,\ \ \ \ n\in\mathbb{Z}_{\geq 3}\ \textnormal{odd}.
\end{equation*}
\end{prop}
\begin{rem}\label{crucrig} The ordinary Fredholm determinant of $\K_n$ is ill-defined since not all its entries vanish at $\pm\infty$ and since $\epsilon$ is not trace-class on $L^2(\mathbb{R})$. This is a standard issue in random matrix theory, compare \cite[Section VIII]{TW}, \cite[page $2199$]{TW2} or \cite[page $79-84$]{DG}, and it is commonly bypassed either through the use of regularized determinants or weighted Hilbert spaces. In \eqref{e:1} we use the following regularized $2$-determinant for block operators $\L=\bigl[\begin{smallmatrix} L_{11} & L_{12}\\ L_{21} & L_{22}\end{smallmatrix}\bigr]$ with trace class diagonal $L_{11},L_{22}$ and Hilbert-Schmidt off-diagonal $L_{12},L_{21}$, cf. \cite[page $82$]{DG},
\begin{equation}\label{HC}
	\det_2(1+\L):=\det\big((1+\L)\hspace{0.02cm}\e^{-\L}\hspace{0.02cm}\big)\e^{\tr(L_{11}+L_{22})},
\end{equation}
where $\det$ is the ordinary Fredholm determinant, the block operators act on $L^2(\mathbb{R})\oplus L^2(\mathbb{R})$ and the trace in the exponent is taken in $L^2(\mathbb{R})$. Note that \eqref{HC} is slightly different from the Hilbert-Carleman determinant \cite[Chapter $9$]{Si} in that for trace class $\L$ we have $\det_2(1+\L)=\det(1+\L)$ and for any two of the above block operators
\begin{equation}\label{HC:1}
	\det_2(1+\L+\M+\L\M)=\det_2(1+\L)\det_2(1+\M).
\end{equation}
Moreover, as soon as $\L\M$ and $\M\L$ fit into the aforementioned class of block operators,
\begin{equation}\label{HC:2}
	\det_2(1+\L\M)=\det_2(1+\M\L),
\end{equation}
and $\det_2(1+\L)\neq 0$ if and only if $1+\L$ is invertible.
\end{rem}
The finite $n$ GinOE result \eqref{e:1} can be used to derive a limit theorem for the largest real eigenvalue of a real Ginibre matrix which in turn quantifies the well-known saturn effect. Indeed, in order to state the corresponding limit theorem for the largest real eigenvalue we first consider the following Riemann-Hilbert problem (RHP).
\begin{problem}[{\cite[RHP $1.5$]{BB}}]\label{master} Given $x,\gamma\in\mathbb{R}\times[0,1]$, determine $\Y(z)=\Y(z;x,\gamma)\in\mathbb{C}^{2\times 2}$ such that
\begin{enumerate}
	\item[(1)] $\Y(z)$ is analytic for $z\in\mathbb{C}\setminus\mathbb{R}$ and continuous on the closed upper and lower half-planes.
	\item[(2)] The boundary values $\Y_{\pm}(z):=\lim_{\epsilon\downarrow 0}\Y(z\pm\im\epsilon),z\in\mathbb{R}$ satisfy
	\begin{equation*}
		\Y_+(z)=\Y_-(z)\begin{bmatrix}1-|r(z)|^2 & -\bar{r}(z)\e^{-2\im xz}\smallskip\\
		r(z)\e^{2\im xz} & 1\end{bmatrix},\ \ z\in\mathbb{R};\ \ \ \ \ \ \ r(z)=r(z;\gamma):=-\im\sqrt{\gamma}\,\e^{-\frac{1}{4}z^2}.
	\end{equation*}
	\item[(3)] As $z\rightarrow\infty$,
	\begin{equation}\label{e:3}
		\Y(z)=\mathbb{I}+\Y_1(x,\gamma)z^{-1}+\mathcal{O}\big(z^{-2}\big);\ \ \ \ \ \Y_1(x,\gamma)=\big[Y_1^{jk}(x,\gamma)\big]_{j,k=1}^2.
	\end{equation}
\end{enumerate}
\end{problem}
This problem is uniquely solvable for all $(x,\gamma)\in\mathbb{R}\times[0,1]$, cf. \cite[Theorem $3.9$]{BB}, and its solution enables us to state the limit theorem for the largest real eigenvalue as follows. Eigenvalues off the real axis are much simpler to deal with, see \cite[Theorem $1.2$]{RS}.
\begin{theo}[{\cite[Theorem $1.3$]{RS}, \cite[Theorem $1.1$]{PTZ}, \cite[Theorem $1.6$]{BB}}]\label{firsttheo} Let $\X\in\mathbb{R}^{n\times n}$ be a matrix drawn from the GinOE with eigenvalues $\{z_j(\X)\}_{j=1}^n\subset\mathbb{C}$. Then for every $t\in\mathbb{R}$,
\begin{align}
	\lim_{n\rightarrow\infty}\mathbb{P}\left(\max_{\substack{j=1,\ldots,n\\ z_j\in\mathbb{R}}}z_j(\X)\leq\sqrt{n}+t\right)=&\,\sqrt{\det(1-\chi_tT\chi_t{\upharpoonright}_{L^2(\mathbb{R})})\,\Gamma_t}\label{e:4}\\
	=&\,\exp\left[-\frac{1}{8}\int_t^{\infty}(x-t)\Big|\,y\Big(\frac{x}{2};1\Big)\Big|^2\,\d x+\frac{\im}{4}\int_t^{\infty}y\Big(\frac{x}{2};1\Big)\,\d x\right],\nonumber
\end{align}
where $T:L^2(\mathbb{R})\rightarrow L^2(\mathbb{R})$ is trace class with kernel
\begin{equation}\label{e:5}
	T(x,y):=\frac{1}{\pi}\int_0^{\infty}\e^{-(x+u)^2}\e^{-(y+u)^2}\,\d u,
\end{equation}
and
\begin{equation}\label{e:6}
	\Gamma_t:=1-\int_t^{\infty}G(x)\big((1-T\chi_t{\upharpoonright}_{L^2(\mathbb{R})})^{-1}g\big)(x)\,\d x;\ \ \ \ \ g(x):=\frac{1}{\sqrt{\pi}}\e^{-x^2},\ \ \ \ G(x):=\int_{-\infty}^xg(y)\,\d y.
\end{equation}
The function $y=y(x;1):\mathbb{R}\times[0,1]\rightarrow\im\mathbb{R}$ equals $y(x;1):=2\im Y_1^{12}(x,1)$, which is expressed in terms of the matrix coefficient $\Y_1(x,1)$ that appeared in \eqref{e:3}.
\end{theo} 
\begin{rem}\label{rem:1} The first equality in \eqref{e:4} is due to Rider and Sinclair \cite[Theorem $1.2$]{RS} with a subsequent algebraic correction of the factor $\Gamma_t$ by Poplavskyi, Tribe and Zaboronski \cite[Theorem $1.1$]{PTZ}. The second equality was derived by the authors \cite[Theorem $1.6$]{BB} and should be viewed as the GinOE analogue of the famous Tracy-Widom law for the largest eigenvalue in the Gaussian Orthogonal Ensemble (GOE), compare \cite[$(53)$]{TW}. Indeed, as far as the largest real eigenvalue is concerned, the overall difference between GinOE and GOE stems from the appearance of the function $y$, i.e. the solution of a distinguished inverse scattering problem for the Zakharov-Shabat system \cite[Section $1.2$]{BB}, rather than the more familiar Painlev\'e-II Hastings-McLeod transcendent.
\end{rem} 
\begin{rem} We emphasize that the limit law \eqref{e:4} is not a feature of the GinOE alone. In fact, Cipolloni, Erd\H{o}s and Schr\"oder recently proved in \cite[Theorem $2.3$]{CES} that the edge eigenvalue statistics of a large class of real non-Hermitian random matrices with i.i.d. centered entries match those of the GinOE. Thus, in complete analogy with the Tracy-Widom law for real Wigner matrices \cite{So}, the law \eqref{e:4} is a universal limit law. The same holds true for the upcoming limit law \eqref{e:9} for thinned real non-Hermitian random matrices at their spectral edge.
\end{rem}
\subsection{Fredholm determinant formula}
In this paper we are concerned with the limiting ($n\rightarrow\infty$) distribution of the largest real eigenvalue in the following \textit{thinned} real GinOE process: consider the Pfaffian point process formed by the $m_n\leq n$ real eigenvalues of some $\X\in\textnormal{GinOE}$. Fix $\gamma\in[0,1]$ and now discard each eigenvalue $\mathbb{R}\ni z_j(\X),j=1,\ldots,m_n$ independently with likelihood $1-\gamma$. The resulting particle system 
\begin{equation*}
	\big\{z_j^{\gamma}(\X)\big\}_{j=1}^{m_{\gamma,n}}\ \ \ \textnormal{with}\ \  m_{\gamma,n}\leq m_n\leq n,
\end{equation*}
forms also a random point process, see \cite[Chapter $6.2.1$]{IPSS}, and most importantly for us, this process is Pfaffian as stated in our first result below.
\begin{lem}\label{thinPfaff} The above defined thinned real GinOE process is a Pfaffian random point process with
\begin{equation}\label{e:8}
	\mathbb{P}\left(\max_{j=1,\ldots,m_{\gamma,n}}z_j^{\gamma}(\X)\leq t\right)=\sqrt{\det_2\big(1-\gamma\chi_t\K_n\chi_t{\upharpoonright}_{L^2(\mathbb{R})\oplus L^2(\mathbb{R})}\big)},\ \ \ \ (t,\gamma)\in\mathbb{R}\times[0,1],
\end{equation}
where the operator $\K_n$ appeared in \eqref{e:2}. 
\end{lem}
Identities similar to \eqref{e:8} have been derived in \cite[Proposition $1.1$]{BoBu} for the limiting GOE and the limiting Gaussian symplectic ensemble (GSE) based on Painlev\'e representations for the underlying eigenvalue generating functions, cf. \cite[Theorem $2.1$]{D}. Our proof of Lemma \ref{thinPfaff} will rely on the observation that thinned Pfaffian point processes are Pfaffian with an appropriately $\gamma$-modified kernel, see Section \ref{sec:1} below, which is similar to the proof for determinantal point processes given in \cite[Appendix A]{L}. The fact that a thinned process built from a determinantal point process is also determinantal was first observed in \cite{BP}.\bigskip

Once \eqref{e:8} is established we will then use this finite $n$ result to derive the following limit theorem for the thinned real GinOE process, our second result. Set
\begin{equation}\label{gammabar}
	\bar{\gamma}:=\gamma(2-\gamma)
\end{equation}
and note that $\bar{\gamma}\in[0,1]$ for $\gamma\in[0,1]$. The limit is a convex combination of two simple Fredholm determinants.
\begin{theo}\label{main:1} For any $(t,\gamma)\in\mathbb{R}\times[0,1]$, the limit
\begin{equation}\label{e:13}
	P(t;\gamma):=\lim_{n\rightarrow\infty}\mathbb{P}\left(\max_{j=1,\ldots,m_{\gamma,n}}z_j^{\gamma}(\X)\leq \sqrt{n}+t\right),\ \ \gamma\in[0,1]
\end{equation}
exists and equals
\begin{equation}\label{e:16}
	P(t;\gamma)=\sqrt{\frac{1-\sqrt{\bar{\gamma}}}{2(2-\gamma)}}\det\big(1+\sqrt{\bar{\gamma}}\chi_tS\chi_t\upharpoonright_{L^2(\mathbb{R})}\big)+\sqrt{\frac{1+\sqrt{\bar{\gamma}}}{2(2-\gamma)}}\det\big(1-\sqrt{\bar{\gamma}}\chi_tS\chi_t\upharpoonright_{L^2(\mathbb{R})}\big),
\end{equation}
with $\bar{\gamma}$ defined in \eqref{gammabar}. Here, $S:L^2(\mathbb{R})\rightarrow L^2(\mathbb{R})$ is the trace class integral operator with kernel
\begin{equation*}
	S(x,y)=\frac{1}{2\sqrt{\pi}}\,\e^{-\frac{1}{4}(x+y)^2}.
\end{equation*}
\end{theo}
The special value $\gamma=1$ reduces \eqref{e:16} to
\begin{equation*}
	P(t;1)=\det\big(1-\chi_tS\chi_t{\upharpoonright}_{L^2(\mathbb{R})}\big),
\end{equation*}
which was first proven by the authors in \cite[Theorem $1.11$]{BB}. Note that the formula for $P(t;1)$ is the analogue of the Ferrari-Spohn formula \cite{FS} in the GOE,  generalized to the thinned GOE by Forrester in \cite[Corollary $1$]{F}. Comparing \eqref{e:16} to the last reference (modulo the typo correction $\xi\mapsto\bar{\xi}$ in the determinants in the first line of \cite[$(1.22)$]{F} and after completing squares), we spot a striking resemblance between the thinned GOE and the thinned GinOE: up to the kernel replacement
\begin{equation*}
	S(x,y)\mapsto \textnormal{Ai}(x+y),
\end{equation*}
with the Airy function $w=\textnormal{Ai}(z)$, see \cite[$9.2.2$]{NIST}, the formul\ae\,are exactly the same.
\subsection{Integrability of the thinned real GinOE process}
In our third result we express the limiting distribution function $P(t;\gamma)$ in \eqref{e:13} in terms of the solution of RHP \ref{master} and thus in terms of the solution of RHP \ref{master} and thus in terms of the solution to an inverse scattering problem for the Zakharov-Shabat system. Here are the details:
\begin{theo}\label{main:22} For any $(t,\gamma)\in\mathbb{R}\times[0,1]$,
\begin{equation}\label{e:9}
	P(t;\gamma)=\exp\left[-\frac{1}{8}\int_t^{\infty}(x-t)\Big|y\Big(\frac{x}{2};\bar{\gamma}\Big)\Big|^2\,\d x\right]\left(\sqrt{\frac{1-\sqrt{\bar{\gamma}}}{2(2-\gamma)}}\,\e^{\frac{1}{2}\mu(t;\bar{\gamma})}+\sqrt{\frac{1+\sqrt{\bar{\gamma}}}{2(2-\gamma)}}\,\e^{-\frac{1}{2}\mu(t;\bar{\gamma})}\right)
\end{equation}
where the function $y=y(x;\gamma):\mathbb{R}\times[0,1]\rightarrow\im\mathbb{R}$ is given by $y(x;\gamma):=2\im Y_1^{12}(x,\gamma)$ in terms of \eqref{e:3} and
\begin{equation}\label{e:10}
	\mu(t;\gamma):=-\frac{\im}{2}\int_t^{\infty}y\Big(\frac{x}{2};\gamma\Big)\,\d x.
\end{equation}
\end{theo}
\begin{rem} Note that for every $(t,\gamma)\in\mathbb{R}\times[0,1]$,
\begin{equation}\label{e:99}
	\sqrt{\frac{1-\sqrt{\bar{\gamma}}}{2(2-\gamma)}}\,\e^{\frac{1}{2}\mu(t;\bar{\gamma})}+\sqrt{\frac{1+\sqrt{\bar{\gamma}}}{2(2-\gamma)}}\,\e^{-\frac{1}{2}\mu(t;\bar{\gamma})}=\sqrt{\frac{\gamma-1-\cosh\mu(t;\bar{\gamma})+\sqrt{\bar{\gamma}}\,\sinh\mu(t;\bar{\gamma})}{\gamma-2}}.
\end{equation}
\end{rem}
We emphasize that the structure in the right hand side of \eqref{e:9}, \eqref{e:99} is completely similar to the one in the limiting distribution function for the largest eigenvalue in the thinned GOE ensemble, cf. \cite[$(1.6)$]{BoBu}. It is only the appearance of the solution to the Zakharov-Shabat inverse scattering problem which sets the thinned GinOE apart from the thinned GOE - at least as far as the largest real eigenvalue is concerned, compare Remark \ref{rem:1} for the special case $\gamma=1$. We further emphasize this point with our fourth result, a simple corollary to Theorem \ref{main:2}: let $E(m,(t,\infty))$ denote the limiting (as $n\rightarrow\infty$) probability that there are $m\in\mathbb{Z}_{\geq 0}$ edge scaled real eigenvalues $\mu_j(\X):=z_j(\X)-\sqrt{n}\in\mathbb{R}$ of a matrix $\X\in\textnormal{GinOE}$ in the interval $(t,\infty)\subset\mathbb{R}$. Now define the associated generating function
\begin{equation}\label{e:11}
	E\big((t,\infty);\lambda\big):=\sum_{m=0}^{\infty}E\big(m,(t,\infty)\big)(1-\lambda)^m,
\end{equation}
which, as a consequence of Theorem \ref{main:1} can also be evaluated in terms of the solution of RHP \ref{master}:
\begin{cor}\label{eiggen} For every $(t,\lambda)\in\mathbb{R}\times[0,1]$,
\begin{equation}\label{e:12}
	E\big((t,\infty);\lambda\big)=\exp\left[-\frac{1}{8}\int_t^{\infty}(x-t)\left|y\Big(\frac{x}{2};\bar{\lambda}\Big)\right|^2\,\d x\right]\sqrt{\frac{\lambda-1-\cosh\mu(t;\bar{\lambda})+\sqrt{\bar{\lambda}}\,\sinh\mu(t;\bar{\lambda})}{\lambda-2}},
\end{equation}
with $\bar{\lambda}:=2\lambda-\lambda^2$, the above function $y(x;\lambda)=2\im Y_1^{12}(x,\lambda)$ and the antiderivative \eqref{e:10}.
\end{cor}
Formula \eqref{e:12} is a simple consequence of the inclusion-exclusion principle, see Section \ref{supereasy} below. The generating function is of interest from the random matrix theory viewpoint as it allows one to compute the limiting distribution function $F_m(t)$ of the $m$th largest edge scaled real eigenvalue ($m=1$ is the largest) in the GinOE in recursive form,
\begin{equation*}
	F_{m+1}(t)-F_m(t)=\frac{(-1)^m}{m!}\frac{\d^m}{\d\lambda^m}E\big((t,\infty);\lambda\big)\bigg|_{\lambda=1},\ \ \ m\in\mathbb{Z}_{\geq 0}: \ \ F_0(t)\equiv 0,
\end{equation*}
see \cite[Section $6.3.2$]{BDS} for the standard probabilistic argument used in the derivation of such recursions in random matrix theory.
\begin{rem} The analogue of \eqref{e:12} for the GOE was first derived in \cite[Theorem $2.1$]{D} and then used for the computation of the limiting distribution function of the largest eigenvalue in the thinned GOE, see for example \cite[Proposition $1.1$]{BoBu}. For the GinOE, we will proceed in the reverse direction and first prove \eqref{e:9}.
\end{rem}
\subsection{Tail expansions} One major advantage of the explicit formula \eqref{e:9} - besides the fact that it places the thinned GinOE on firm integrable systems ground - originates from its usefulness in the derivation of tail expansions. Indeed, once the Riemann-Hilbert problem connection is in place, it is somewhat straightforward to obtain asymptotic information for the distribution function $P(t;\gamma)$ in \eqref{e:13} as $t\rightarrow\pm\infty$. We summarize the relevant estimates in our fifth result below.
\begin{theo}\label{main:2} Let $\gamma\in[0,1]$. We have, as $t\rightarrow+\infty$,
\begin{equation}\label{e:14}
	P(t;\gamma)=1-\frac{\gamma}{4}\textnormal{erfc}(t)+\mathcal{O}\big(t^{-1}\e^{-2t^2}\big),
\end{equation}
with the complementary error function $w=\textnormal{erfc}(z)$, see \cite[$7.2.2$]{NIST}. On the other hand, as $t\rightarrow-\infty$,
\begin{equation}\label{e:15}
	P(t;\gamma)=\e^{c_1(\gamma)t+c_0(\gamma)}\big(1+o(1)\big),
\end{equation}
with
\begin{equation}\label{e:155}
	c_1(\gamma)=\frac{1}{2\sqrt{2\pi}}\,\textnormal{Li}_{\frac{3}{2}}(\bar{\gamma}),\ \ \ \ \ \ \ c_0(\gamma)=\frac{1}{2}\ln\left(\frac{2}{2-\gamma}\right)+\frac{1}{4\pi}\int_0^{\bar{\gamma}}\left(\big(\textnormal{Li}_{\frac{1}{2}}(x)\big)^2-\frac{x\pi}{1-x}\right)\frac{\d x}{x},
\end{equation}
in terms of the polylogarithm $w=\textnormal{Li}_s(z)$, see \cite[$25.12.10$]{NIST}.
\end{theo}
Expansion \eqref{e:14} was first derived in \cite{FN} for $\gamma=1$. The leading order exponential decay of the left tail \eqref{e:15} appeared in \cite[$(1.11)$]{PTZ} for $\gamma=1$ and for $\gamma\in[0,1]$ in \cite[$(2.30)$]{FD}, albeit in somewhat implicit form. The notoriously difficult constant factor $c_0(\gamma)$ in \eqref{e:15} was recently computed in \cite[$(3)$]{FTZ} for $\gamma=1$ using probabilistic arguments. In this paper we derive \eqref{e:15} for all $\gamma\in[0,1)$ by nonlinear steepest descent techniques. The evaluation of $c_0(1)$ would require further analysis and we choose not to rederive $c_0(1)$ in this paper. Nonetheless, we note that our result \eqref{e:15}, \eqref{e:155} matches formally onto \cite[$(1.11)$]{PTZ}, \cite[$(3)$]{FTZ}, i.e. onto the $t\rightarrow-\infty$ expansion
\begin{equation}\label{known}
	P(t;1)=\exp\left[\frac{t}{2\sqrt{2\pi}}\,\zeta\left(\frac{3}{2}\right)+\frac{1}{2}\ln 2+\frac{1}{4\pi}\sum_{n=1}^{\infty}\frac{1}{n}\left(-\pi+\sum_{m=1}^{n-1}\frac{1}{\sqrt{m(n-m)}}\right)\right]\big(1+o(1)\big),
\end{equation}
since $c_1(1)=\frac{1}{2\sqrt{2\pi}}\zeta\left(\frac{3}{2}\right)$ and since $c_0(\gamma)$ in \eqref{e:155} satisfies the following property
\begin{lem}\label{toobad} The function $c_0(\gamma)$ is continuous in $\gamma\in[0,1]$ and equals
\begin{equation}\label{Jext}
	c_0(\gamma)=\frac{1}{2}\ln\left(\frac{2}{2-\gamma}\right)+\frac{1}{4\pi}\sum_{n=1}^{\infty}\frac{1}{n}\left(-\pi+\sum_{m=1}^{n-1}\frac{1}{\sqrt{m(n-m)}}\right)\bar{\gamma}^n.
\end{equation}
\end{lem}
As it is standard (for instance in invariant random matrix theory ensembles) the right tail \eqref{e:14} of the extreme value distribution $P(t;\gamma)$ follows from elementary considerations and does not need RHP \ref{master}. The left tail, however, is much more subtle since
\begin{equation*}
	\tr_{L^2(\mathbb{R})}(\chi_tT\chi_t)=\int_t^{\infty}T(x,x)\,\d x=\frac{1}{2\pi}\sqrt{\frac{\pi}{2}}\,\int_t^{\infty}\textnormal{erfc}(\sqrt{2}\,x)\,\d x=-\frac{t}{\pi}\sqrt{\frac{\pi}{2}}+\mathcal{O}(1),\ \ \ \ \ t\rightarrow-\infty,
\end{equation*}
becomes unbounded, yet the distribution function $P(t;\gamma)$ converges to zero. It is this well-known issue which requires the full use of RHP \ref{master} and associated nonlinear steepest descent techniques for its asymptotic analysis, see Section \ref{nonsteep} below.
\begin{rem} The explicit computation of constant factors such as $c_0(\gamma)$ in \eqref{e:155} is a well-known challenge in the asymptotic analysis of correlation and distribution functions in nonlinear mathematical physics. Without aiming for completeness, we mention the following contributions to the field: in the theory of exactly solvable lattice models, the works \cites{T,BT,BlBo,B1,B2}. In classical invariant random matrix theory, the works \cites{BBF,K,E,E2,DIKZ,DIK}, and most recently on $\tau$-function connection problems for Painlev\'e transcendents the works \cites{ILP,IP}.
\end{rem}

\subsection{Numerics} The Fredholm determinant formula \eqref{e:16} provides us with an efficient way to evaluate $P(t;\gamma)$ numerically, cf. \cite{B}. Indeed, in order to showcase the applicability of \eqref{e:16} we now provide the following numerical evaluations for the limiting distribution of $\max_j z_j^{\gamma}(\X)$: First, Table \ref{tab1} shows a few centralized moments for varying $\gamma$.\smallskip
\begin{table}[h]
\caption{Some moments of the thinned real GinOE process}\label{tab:1}
\begin{center}
\begin{tabular}{ r  r  r r r}
\toprule
$\gamma$ & mean & variance& skewness& kurtosis\\[2pt] \midrule
1& -1.30319& 3.97536& -1.76969&5.14560\\[2pt] \midrule
0.8& -1.94070&6.87453 &-1.86716 &5.57883\\[2pt] \midrule
0.6 & -2.99680&13.49947 & -2.02286& 8.06831\\[2pt] \midrule
0.4 & -5.12526& 36.37796& -3.02040& 22.14125\\[2pt] \bottomrule
\end{tabular}
\label{tab1}
\end{center}
\end{table}

Second, probability density and distribution function plots for varying $\gamma\in[0,1]$ are shown in Figure \ref{figure1}. 
\begin{center}
\begin{figure}[tbh]
\resizebox{0.4545\textwidth}{!}{\includegraphics{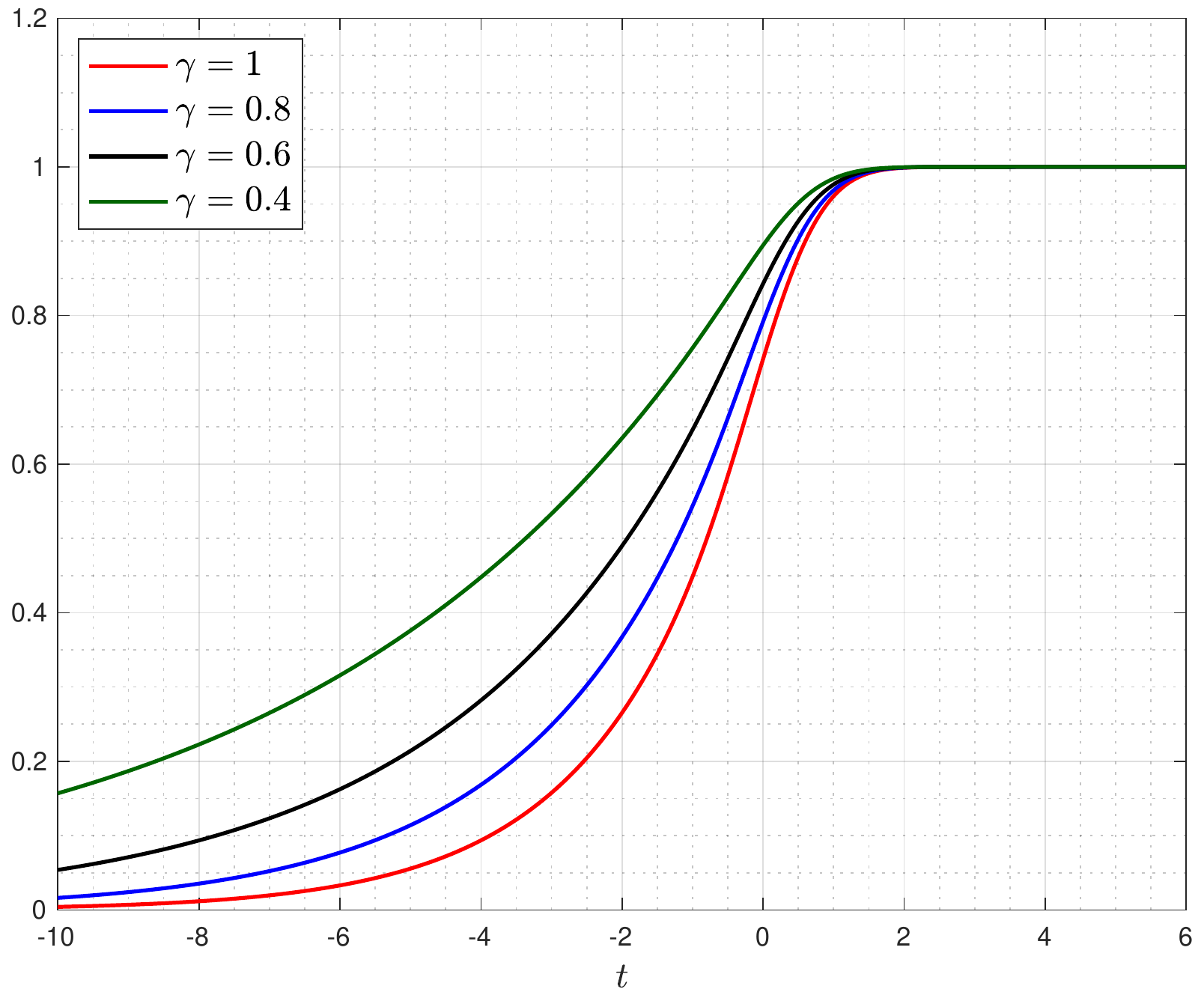}}\ \ \ \ \ \resizebox{0.455\textwidth}{!}{\includegraphics{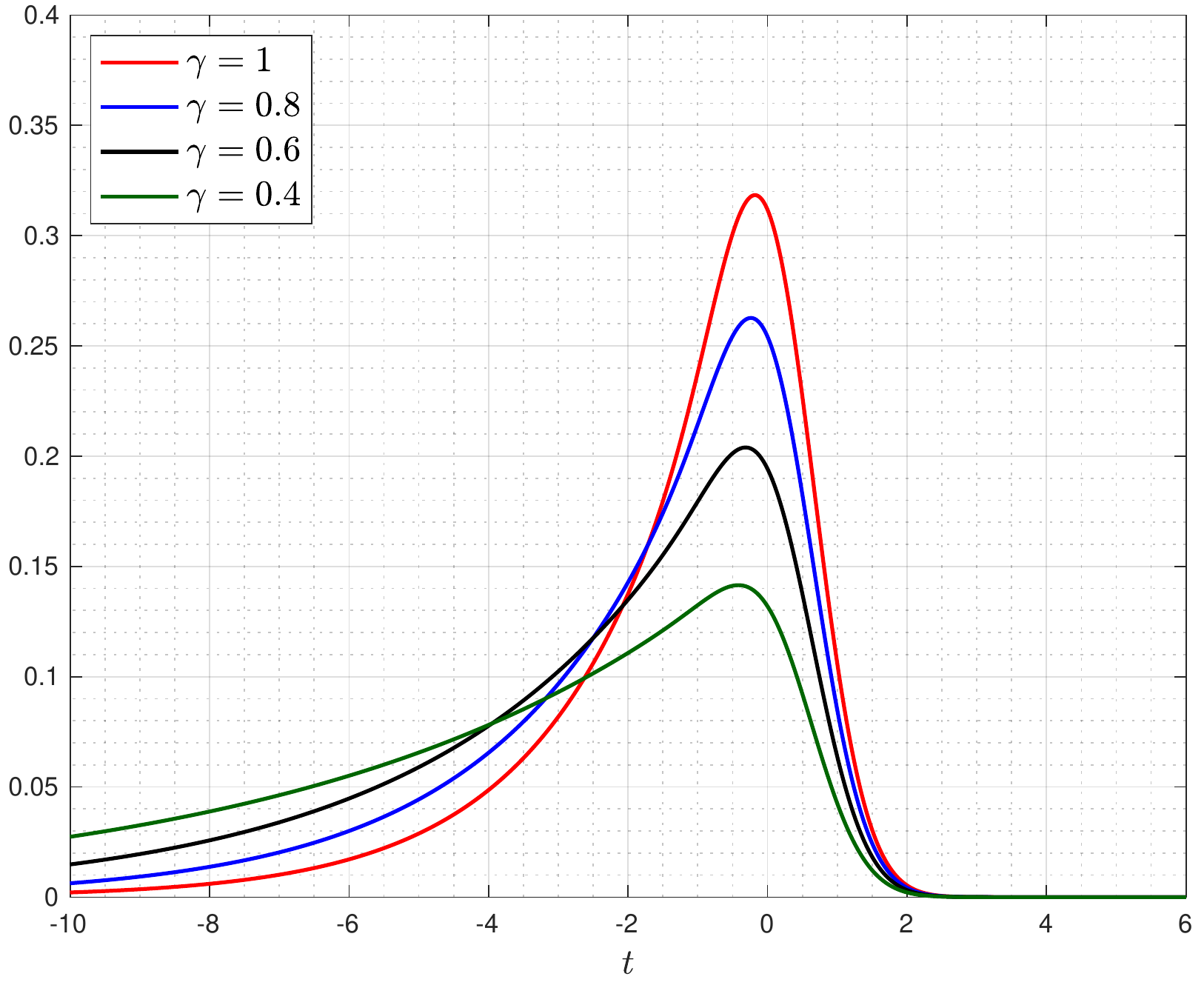}}
\caption{The distribution functions $P(t;\gamma)$ of the largest real eigenvalue in the thinned real GinOE process for varying values of $\gamma$. The plots were generated in MATLAB with $m=50$ quadrature points using the Nystr\"om method with Gauss-Legendre quadrature. On the left cdfs, on the right pdfs.}
\label{figure1}
\end{figure}
\end{center} 
Third, we compare our asymptotic expansions \eqref{e:14} and \eqref{e:15} to the numerical results obtained from \eqref{e:16} in Figure \ref{figure2} and \ref{figure3} below.
\begin{center}
\begin{figure}[tbh]
\resizebox{0.463\textwidth}{!}{\includegraphics{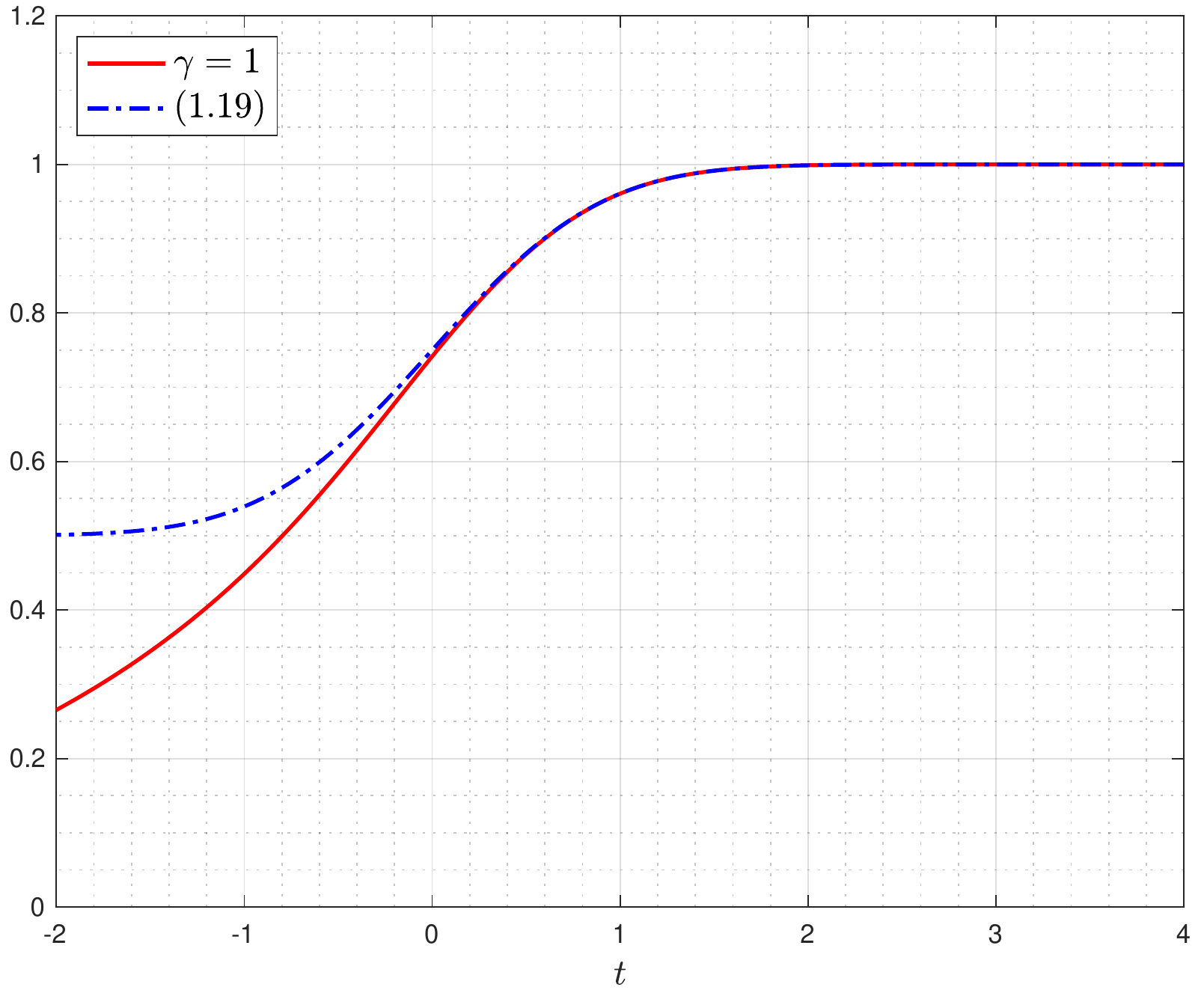}}\ \ \ \ \ \resizebox{0.462\textwidth}{!}{\includegraphics{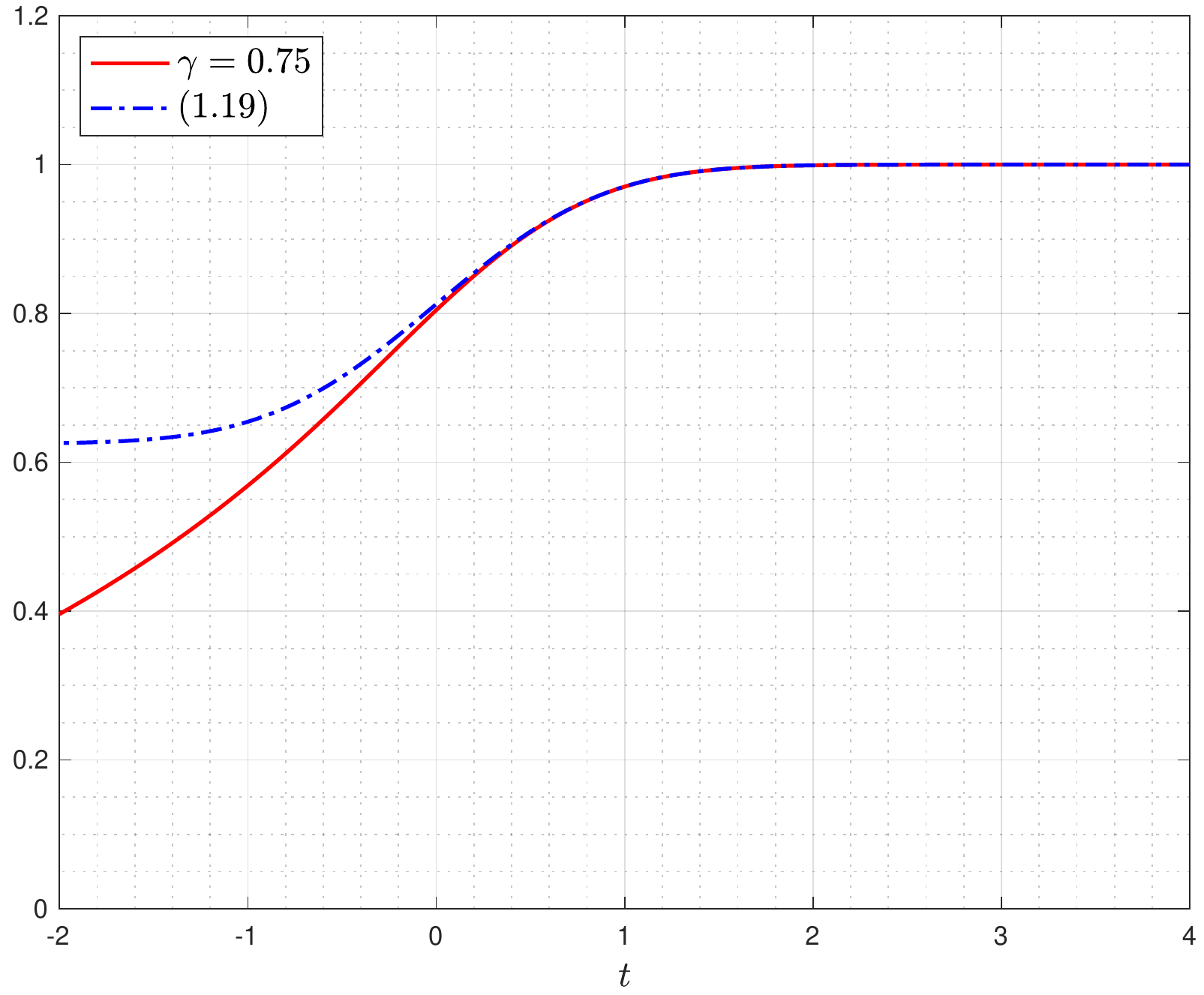}}
\caption{The distribution functions $P(t;\gamma)$ in red for $\gamma=1$ (left) and $\gamma=0.75$ (right). We compare the numerical computed values from \eqref{e:16} to the right tail expansion \eqref{e:14}. Again we used the Nystr\"om method with Gauss-Legendre quadrature and $m=50$ quadrature points.}
\label{figure2}
\end{figure}
\end{center} 
\begin{center}
\begin{figure}[tbh]
\resizebox{0.464\textwidth}{!}{\includegraphics{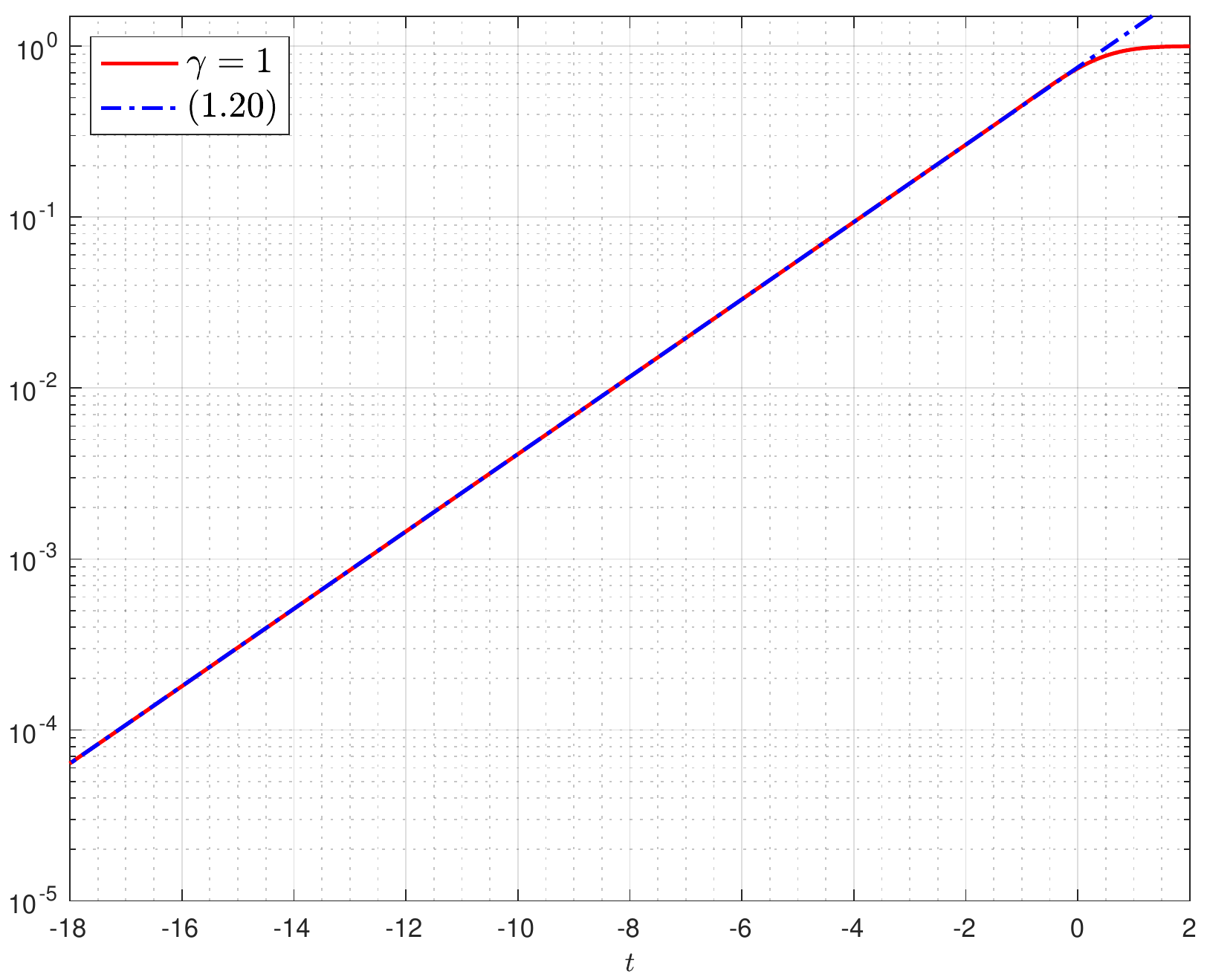}}\ \ \ \ \ \resizebox{0.462\textwidth}{!}{\includegraphics{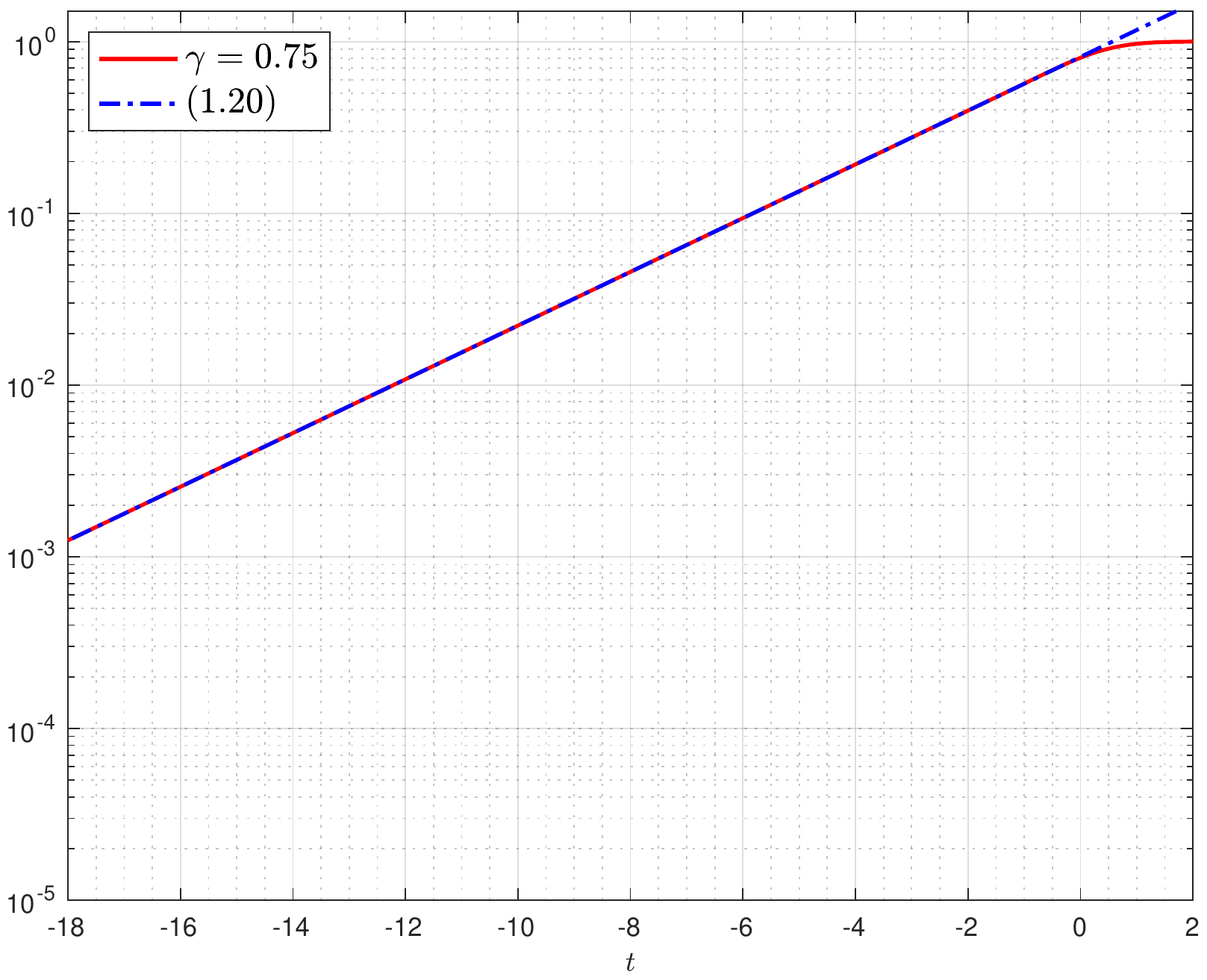}}
\caption{The distribution functions $P(t;\gamma)$ in red for $\gamma=1$ (left) and $\gamma=0.75$ (right). We compare the numerical computed values from \eqref{e:16} to the left tail expansion \eqref{e:15} in a semilogarithmic plot. Again we used the Nystr\"om method with Gauss-Legendre quadrature and $m=50$ quadrature points.}
\label{figure3}
\end{figure}
\end{center} 

\subsection{Methodology and outline of paper} The remainder of the paper is organized as follows. We prove Lemma \ref{thinPfaff} in Section \ref{sec:1} using a simple probabilistic argument. Afterwards we use \eqref{e:8} and carefully simplify the regularized Fredholm determinant in order to arrive at a finite $n$ formula which is amenable to asymptotics. Our approach is somewhat similar to the ones carried out in \cites{D,RS}, however two issues arise along the way: one, the absence of Christoffel-Darboux structures throughout forces us to rely on the Fourier tricks used in \cite[Section $2$ and $3$]{BB} in the derivation of \eqref{e:9}. Two, unlike in the invariant ensembles, our computations depend heavily on the parity of $n$. We first work out the necessary details for even $n$ in Section \ref{sec:even} and afterwards develop a comparison argument to treat all odd $n$, see Subsection \ref{sec:odd}. The content of Subsection \ref{sec:odd} seemingly marks the first time that the extreme value statistics in the GinOE for odd $n$ have been computed rigorously. Even for $\gamma=1$, typos in \cite[Section $4.2$]{RS} have been pointed out in \cite[Appendix B]{PTZ} but these had not been fixed until now. After several initial steps in Section \ref{sec:even} we complete the proof of Theorem \ref{main:1} in Section \ref{oddsec}. Once Theorem \ref{main:1} has been derived, our proof of Theorem \ref{main:22} in Section \ref{cool} is rather short, making essential use of the inverse scattering theory connection worked out in our previous paper \cite{BB}. This is followed by our short proof of \eqref{eiggen} for the eigenvalue generating function in Section \ref{supereasy}. Afterwards we prove Theorem \ref{main:2} in Section \ref{nonsteep}. In fact, the asymptotic analysis is split into two parts, one part which deals with a total integral of $y=y(x;\gamma)$ and a second part which computes the constant factor in the asymptotic expansion of the determinant 
\begin{equation*}
	\det(1-\gamma\chi_t T\chi_t\upharpoonright_{L^2(\mathbb{R})}).
\end{equation*}
Unlike for invariant matrix ensembles, compare the discussion in \cite[page $492,493$]{BIP}, we are here able to efficiently employ the $\gamma$-derivative method in the computation of the constant factor without having a differential equation in the spectral variable. Indeed, since our nonlinear steepest descent analysis in Appendix \ref{steepbetter} does not use any local model functions, the cumbersome double integration in the $\gamma$-derivative method becomes manageable. This feature is comparable with Deift's proof of the strong Szeg\H{o} limit theorem in \cite[Example $3$]{DInt} and the details of our analysis can be found in Section \ref{nonsteep}. The final two sections of the paper in Appendix \ref{sec:int} and \ref{steepbetter} prove two curious integral identities used in the proof of Theorem \ref{main:1} and present a streamlined version of the nonlinear steepest descent analysis of \cite[Section $5$]{BB} which is crucial in our proof of Theorem \ref{main:2}.


\section{Proof of Lemma \ref{thinPfaff}}\label{sec:1}
It is known from \cite{LS} that the eigenvalues $\{z_j(\X)\}_{j=1}^n\subset\mathbb{C}$ of $\X\in\mathbb{R}^{n\times n}$ drawn from the $\textnormal{GinOE}$ are distributed according to a random point process whose correlation functions are computable as Pfaffians, cf. \cite{BS,FN,S}. In particular, the real eigenvalues form a Pfaffian process whose correlations are given by
\begin{equation}\label{pro:1}
	\rho_{\ell}(w_1,\ldots,w_{\ell})=\textnormal{Pf}\big[\K_n^{\mathbb{R},\mathbb{R}}(w_j,w_k)\big]_{j,k=1}^{\ell},\ \ \ \ 1\leq \ell\leq m_n\leq n,
\end{equation}
with the skew-symmetric $2\times 2$-matrix kernel
\begin{equation*}
	\K_n^{\mathbb{R},\mathbb{R}}(x,y):=\begin{bmatrix}\epsilon(x,y)-(IS_n)(x,y) & S_n(y,x)\smallskip\\ -S_n(x,y) & -(DS_n^{\ast})(x,y)\end{bmatrix}.
\end{equation*}
Note that for any distinct points $w_j\in\mathbb{R}$,
\begin{equation*}
	\rho_{\ell}(w_1,\ldots,w_{\ell})=\lim_{\Delta w_i\rightarrow 0}\frac{\mathbb{P}(\textnormal{one real GinOE eigenvalue in each}\,(w_i,w_i+\Delta w_i))}{\Delta w_1\cdot\ldots\cdot\Delta w_{\ell}}.
\end{equation*}
Thus, if $\rho_{\ell}^{\gamma}$ denotes the $\ell$-th correlation function in the thinned real GinOE process, we find with $1\leq \ell\leq m_{\gamma,n}$,
\begin{align*}
	\rho_{\ell}^{\gamma}(w_1,\ldots,w_{\ell})=&\,\lim_{\Delta w_i\rightarrow 0}\frac{\mathbb{P}(\textnormal{one thinned real GinOE eigenvalue in each}\,(w_i,w_i+\Delta w_i))}{\Delta w_1\cdot\ldots\cdot\Delta w_{\ell}}\\
	=&\,\lim_{\Delta w_i\rightarrow 0}\frac{\mathbb{P}(\textnormal{one real GinOE eigenvalue in each}\,(w_i,w_i+\Delta w_i)\,\textnormal{and they are not discarded})}{\Delta w_1\cdot\ldots\cdot\Delta w_{\ell}}\\
	=&\,\lim_{\Delta w_i\rightarrow 0}\frac{\mathbb{P}(\textnormal{one real GinOE eigenvalue in each}\,(w_i,w_i+\Delta w_i))}{\Delta w_1\cdot\ldots\cdot\Delta w_{\ell}}\big((1-(1-\gamma)\big)^{\ell},
\end{align*}
since each eigenvalue is removed independently with likelihood $1-\gamma$. In short, $\rho_{\ell}^{\gamma}=\gamma^{\ell}\rho_{\ell}$ which shows that the thinned Pfaffian point process is also a Pfaffian process and its kernel is simply given by $\gamma \K_n^{\mathbb{R},\mathbb{R}}(x,y)$. Equipped with this insight one now repeats the computations in \cite[page $1630$]{RS} and arrives at \eqref{e:8}.

\section{Proof of Theorem \ref{main:1} - first steps}\label{sec:even}
Abbreviate
\begin{equation*}
	F_{n}\equiv F_{n}(t,\gamma):=\det_2\big(1-\gamma\chi_t\K_{n}\chi_t{\upharpoonright}_{L^2(\mathbb{R})\oplus L^2(\mathbb{R})}\big),\ \ \ n\in\mathbb{Z}_{\geq 2}.
\end{equation*}
We will first simplify $F_n$ for $n$ even and afterwards take the limit as $n\rightarrow\infty$ with $n$ even. Once done, we then compare the odd $n$ case with the even $n$ case and prove existence of the limit \eqref{e:13} all together.
\subsection{Finite even $n$ calculations}\label{evennsec} We consider $F_{2n}$. Our overall approach follows closely \cite[page $1640$]{RS}, keeping throughout track of the $\gamma$-modifications due to \eqref{e:8}. First, the kernel $\chi_t\K_{2n}\chi_t$ can be factorized as
\begin{equation}\label{r:0}
	\begin{bmatrix}\rho^{-1}\chi_tD\rho & 0\smallskip\\
	0 & \rho\chi_t\rho^{-1}\end{bmatrix}\begin{bmatrix}-\rho^{-1}\epsilon S_{2n}\chi_t\rho & \rho^{-1}S_{2n}^{\ast}\chi_t\rho^{-1}\smallskip\\
	-\rho(\epsilon S_{2n}-\epsilon)\chi_t\rho & \rho S_{2n}^{\ast}\chi_t\rho^{-1}\end{bmatrix},
\end{equation}
and by using \eqref{HC:2} we can move the factor on the left in \eqref{r:0} to the right, so $F_{2n}(t,\gamma)$ equals the regularized $2$-determinant of the operator with kernel
\begin{equation*}
	\gamma\begin{bmatrix}-\rho^{-1}\epsilon S_{2n}\chi_tD\rho & \rho^{-1}S_{2n}^{\ast}\chi_t\rho^{-1}\smallskip\\
	-\rho(\epsilon S_{2n}-\epsilon)\chi_tD\rho & \rho S_{2n}^{\ast}\chi_t\rho^{-1}\end{bmatrix}.
\end{equation*}
Next we observe that the traces of the last operator's powers of $2,3,\ldots$ match the corresponding traces of the operator with kernel
\begin{equation*}
	\gamma\begin{bmatrix}-\rho^{-1}(\epsilon S_{2n}\chi_tD-S_{2n}^{\ast}\chi_t)\rho & \rho^{-1}S_{2n}^{\ast}\chi_t\rho^{-1}\smallskip\\
	\rho\epsilon\chi_tD\rho & 0\end{bmatrix}.
\end{equation*}
Hence, by the Plemelj-Smithies formula for $\det_2$, see \cite[Theorem $9.3$]{Si}, 
\begin{equation*}
	F_{2n}(t,\gamma)=\det_2\begin{bmatrix}1+\gamma\rho^{-1}(\epsilon S_{2n}\chi_tD-S_{2n}^{\ast}\chi_t)\rho & -\gamma\rho^{-1}S_{2n}^{\ast}\chi_t\rho^{-1}\smallskip\\
	-\gamma\rho\epsilon\chi_tD\rho & 1\end{bmatrix}.
\end{equation*}
Factorizing the underlying kernel we then obtain
\begin{equation*}
	F_{2n}=\det_2\left(\begin{bmatrix}1 & -\gamma\rho^{-1}S_{2n}^{\ast}\chi_t\rho^{-1}\smallskip\\ 0 & 1\end{bmatrix}\begin{bmatrix}1+\gamma\rho^{-1}(\epsilon S_{2n}\chi_tD-S_{2n}^{\ast}\chi_t-\gamma S_{2n}^{\ast}\chi_t\epsilon\chi_tD)\rho & 0\smallskip\\
	0 & 1\end{bmatrix}\begin{bmatrix}1 & 0\smallskip\\
	-\gamma\rho\epsilon\chi_tD\rho & 1\end{bmatrix}\right),
\end{equation*}
and since both triangular factors are of the form identity plus block operator as in Remark \ref{crucrig}, we are allowed to use \eqref{HC:1}. In fact the regularized $2$-determinant of those triangular factors equals one, so we have just shown that the original determinant in \eqref{e:1} for even $n$ simplifies to
\begin{equation}\label{r:00}
	F_{2n}(t,\gamma)=\det_2\left(1+\begin{bmatrix}\gamma\rho^{-1}(\epsilon S_{2n}\chi_tD-S_{2n}^{\ast}\chi_t-\gamma S_{2n}^{\ast}\chi_t\epsilon\chi_tD)\rho & 0\smallskip\\
	0 & 0\end{bmatrix}{\upharpoonright}_{L^2(\mathbb{R})\oplus L^2(\mathbb{R})}\right).
\end{equation}
Clearly, the determinant in \eqref{r:00} on $L^2(\mathbb{R})\oplus L^2(\mathbb{R})$ is really a determinant on $L^2(\mathbb{R})$ alone,
\begin{equation}\label{r:000}
	F_{2n}(t,\gamma)=\det_2\big(1+\gamma\rho^{-1}(\epsilon S_{2n}\chi_tD-S_{2n}^{\ast}\chi_t-\gamma S_{2n}^{\ast}\chi_t\epsilon\chi_tD)\rho{\upharpoonright}_{L^2(\mathbb{R})}\big),
\end{equation}
and as our upcoming computations will show (see in particular \eqref{p:3} below) the operator $\epsilon S_{2n}\chi_tD-S_{2n}^{\ast}\chi_t-\gamma S_{2n}^{\ast}\chi_t\epsilon\chi_tD$ is of finite rank, i.e. the regularized $2$-determinant in \eqref{r:000} is an ordinary Fredholm determinant by Remark \ref{crucrig} and the conjugation with $\rho$ now redundant. We have thus arrived at the following replacement of the equation right above \cite[$(4.6)$]{RS},
\begin{equation}\label{r:1}
	F_{2n}(t,\gamma)=\det\big(1-\gamma S_{2n}^{\ast}\chi_t+\gamma\epsilon S_{2n}\chi_tD-\gamma^2S_{2n}^{\ast}\chi_t\epsilon\chi_tD{\upharpoonright}_{L^2(\mathbb{R})}\big).
\end{equation}
In order to simplify \eqref{r:1} further we now record
\begin{lem}[{\cite[page $1640$]{RS}}] For any $n\in\mathbb{Z}_{\geq 1}$, 
\begin{equation}\label{r:2}
	\epsilon S_{2n}=S_{2n}^{\ast}\epsilon.
\end{equation}
\end{lem}
\begin{proof} The stated identity follows easily by induction on $n\in\mathbb{Z}_{\geq 1}$ using only that
\begin{equation*}
	S_n(x,y)=\frac{1}{\sqrt{2\pi}}\e^{-\frac{1}{2}(x^2+y^2)}\mathfrak{e}_{n-2}(xy)+\frac{x^{n-1}\e^{-\frac{1}{2}x^2}}{\sqrt{2\pi}\,(n-2)!}\int_0^yu^{n-2}\e^{-\frac{1}{2}u^2}\,\d u,\ \ \ n\in\mathbb{Z}_{\geq 2}.
\end{equation*}
\end{proof}
Inserting \eqref{r:2} into \eqref{r:1} we find
\begin{equation}\label{p:1}
	F_{2n}(t,\gamma)=\det\big(1-\gamma S_{2n}^{\ast}\chi_t+\gamma S_{2n}^{\ast}(1-\gamma\chi_t)\epsilon\chi_tD\upharpoonright_{L^2(\mathbb{R})}\big).
\end{equation}
We write $\alpha\otimes \beta$ for a general rank one integral operator on $L^2(\mathbb{R})$ with kernel $(\alpha\otimes\beta)(x,y)=\alpha(x)\beta(y)$. Noting $\epsilon D\chi_t=-\chi_t$ and applying the commutator identity, cf. \cite[$(16)$]{TW},
\begin{equation*}
	\epsilon\big[\chi_t,D\big]=-\epsilon_t\otimes\delta_t+\epsilon_{\infty}\otimes\delta_{\infty}\ \ \ \textnormal{with} \ \ \ \int_{-\infty}^{\infty}f(x)\delta_a(x)\,\d x:=f(a),\ \ \ \ \begin{cases}\,\,\,\epsilon_t(x):=\frac{1}{2}\textnormal{sgn}(t-x),\ \ t\in\mathbb{R}\smallskip&\\ \epsilon_{\infty}(x):=\frac{1}{2}&\end{cases},
\end{equation*}
one part in \eqref{p:1} simplifies to
\begin{equation*}
	(1-\gamma\chi_t)\epsilon\chi_t D=(1-\gamma\chi_t)(-\epsilon_t\otimes\delta_t+\epsilon_{\infty}\otimes\delta_{\infty})-(1-\gamma)\chi_t,
\end{equation*}
which (since $\epsilon_t=\frac{1}{2}-\chi_{[t,\infty)},\epsilon_{\infty}=\frac{1}{2}$) yields
\begin{equation}\label{p:2}
	(1-\gamma\chi_t)\epsilon\chi_t D=-\big((1-\gamma\chi_t)\epsilon_{\infty}\big)\otimes(\delta_t-\delta_{\infty})+(1-\gamma)\big(\chi_{[t,\infty)}\otimes\delta_t\big)-(1-\gamma)\chi_t.
\end{equation}
Substituting \eqref{p:2} back into \eqref{p:1} we have thus (recall $\bar{\gamma}=2\gamma-\gamma^2$)
\begin{equation}\label{p:3}
	F_{2n}(t,\gamma)=\det\left(1-\bar{\gamma}S_{2n}^{\ast}\chi_t-\gamma S_{2n}^{\ast}\Big(\big((1-\gamma\chi_t)\epsilon_{\infty}\big)\otimes(\delta_t-\delta_{\infty})\Big)+\gamma(1-\gamma)S_{2n}^{\ast}(\chi_{[t,\infty)}\otimes\delta_t)\upharpoonright_{L^2(\mathbb{R})}\right).
\end{equation}
Next, from the definition of $S_n$ in Proposition \ref{recall}, we may write, see \cite[$(4.7)$]{RS},
\begin{equation*}
	S_n^{\ast}(x,y)=\frac{1}{\sqrt{2\pi}}\e^{-\frac{1}{2}(x^2+y^2)}\mathfrak{e}_{n-2}(xy)+\frac{y^{n-1}\e^{-\frac{1}{2}y^2}}{\sqrt{2\pi}\,(n-2)!}\int_0^xu^{n-2}\e^{-\frac{1}{2}u^2}\,\d u=:T_n(x,y)+\big(\phi_n\otimes\psi_n\big)(x,y),
\end{equation*}
where $T_n(x,y)$ is a symmetric kernel and
\begin{equation}\label{p:4}
	\phi_n(x):=\sqrt{\frac{\sqrt{n}}{\sqrt{2\pi}\,(n-2)!}}\int_0^xu^{n-2}\e^{-\frac{1}{2}u^2}\,\d u,\ \ \ \ \ \ \ \psi_n(y):=\sqrt{\frac{1}{\sqrt{2\pi n}\,(n-2)!}}\,\,y^{n-1}\e^{-\frac{1}{2}y^2}.
\end{equation}
\begin{lem}\label{lem:2} Given $t\geq 0$ and $n\in\mathbb{Z}_{\geq 2}$, the trace class operator $T_n:L^2(t,\infty)\rightarrow L^2(t,\infty)$ with kernel $T_n(x,y)$ satisfies $0\leq T_n\leq 1$ and $1-\gamma T_n$ is invertible on $L^2(t,\infty)$ for all $\gamma\in[0,1]$.
\end{lem}
\begin{proof} For every $f\in L^2(t,\infty)$,
\begin{equation*}
	\langle f,T_n f\rangle_{L^2(t,\infty)}=\frac{1}{\sqrt{2\pi}}\sum_{k=0}^{n-2}\frac{1}{k!}\left|\int_t^{\infty}f(x)\,\e^{-\frac{1}{2}x^2}x^k\,\d x\right|^2,
\end{equation*}
which implies non-negativity of $T_n$. For the upper bound we apply Schur's test, \begin{equation*}
	\|T_n\|\leq\sup_{y>t}\int_t^{\infty}\big|T_n(x,y)\big|\,\d x\leq\sup_{y>t}\frac{1}{\sqrt{2\pi}}\int_t^{\infty}\e^{-\frac{1}{2}(|x|-|y|)^2}\,\d x\leq\frac{1}{\sqrt{\pi}}\int_{-\infty}^{\infty}\e^{-v^2}\,\d v=1\ \ \forall\,t\geq 0,
\end{equation*}
and conclude by self-adjointness of $T_n$ that
\begin{equation*}
	\sup_{\|f\|_{L^2(t,\infty)}=1}\big|\langle f,T_nf\rangle_{L^2(t,\infty)}\big|=\|T_n\|\leq 1,
\end{equation*}
i.e. $T_n\leq 1$ for any $n\in\mathbb{Z}_{\geq 2}$. Next, using that $\|T_n\|\leq 1$, the invertibility of $1-\gamma T_n$ on $L^2(t,\infty)$ follows readily from the underlying Neumann series provided $\gamma\in[0,1)$. The case $\gamma=1$ has been addressed in \cite[Lemma $4.2$]{RS}. This concludes our proof.
\end{proof}
In the following we will use the result of Lemma \ref{lem:2} for the operator $\bar{\gamma}\chi_tT_{2n}\chi_t$ which acts on $L^2(\mathbb{R})$. Inserting the operator decomposition $S_n^{\ast}=T_n+\phi_n\otimes\psi_n$ into \eqref{p:3} and using the general identities $(\alpha\otimes\beta)(\gamma\otimes\delta)=\langle\beta,\gamma\rangle (\alpha\otimes\delta)$ and $A(\beta\otimes\gamma)D=(A\beta)\otimes (D^{\ast}\gamma)$ (for arbitrary operators $A,D$),
\begin{align}
	F_{2n}(t,\gamma)=\det\Big(1-\bar{\gamma}T_{2n}\chi_t&\,-\bar{\gamma}\phi_{2n}\otimes(\chi_t\psi_{2n})-\gamma\big(\langle\psi_{2n},(1-\gamma\chi_t)\epsilon_{\infty}\rangle\phi_{2n}+T_{2n}(1-\gamma\chi_t)\epsilon_{\infty}\big)\otimes(\delta_t-\delta_{\infty})\nonumber\\
	&\,+\gamma(1-\gamma)\big(\langle\psi_{2n},\chi_{[t,\infty)}\rangle\phi_{2n}+T_{2n}\chi_{[t,\infty)}\big)\otimes\delta_t\upharpoonright_{L^2(\mathbb{R})}\Big).\label{p:44}
\end{align}
Here, $\langle\cdot,\cdot\rangle$ is the standard $L^2(\mathbb{R})$ inner product. Since $\chi_t^2=\chi_t,\chi_t^{\ast}=\chi_t$ and 
\begin{equation*}
	\alpha\otimes(\delta_t-\delta_{\infty})=\alpha\otimes\big(\chi_t(\delta_t-\delta_{\infty})\big),\ \ \ \ \alpha\otimes\delta_t=\alpha\otimes(\chi_t\delta_t),
\end{equation*}
we can then rewrite \eqref{p:44} by Sylvester's identity \cite[Chapter IV, $(5.9)$]{GGK} as
\begin{align}
	F_{2n}(t,\gamma)&\,=\det\Big(1-\bar{\gamma}\chi_tT_{2n}\chi_t-\bar{\gamma}(\chi_t\phi_{2n})\otimes(\chi_t\psi_{2n})-\gamma\big(\langle\psi_{2n},(1-\gamma\chi_t)\epsilon_{\infty}\rangle\chi_t\phi_{2n}\label{twist}\\
	&\,+\chi_tT_{2n}(1-\gamma\chi_t)\epsilon_{\infty}\big)\otimes(\delta_t-\delta_{\infty})+\gamma(1-\gamma)\big(\langle\psi_{2n},\chi_{[t,\infty)}\rangle\chi_t\phi_{2n}+\chi_tT_{2n}\chi_{[t,\infty)}\big)\otimes\delta_t\upharpoonright_{L^2(\mathbb{R})}\Big).\nonumber
\end{align}
Next, using that $1-\bar{\gamma}\chi_tT_n\chi_t$ is invertible on $L^2(\mathbb{R})$ for $t\geq 0$ by Lemma \ref{lem:2}, we factorize $F_{2n}(t,\gamma)$ as follows,
\begin{equation}\label{p:5}
	F_{2n}(t,\gamma)=\det\big(1-\bar{\gamma}\chi_tT_{2n}\chi_t\upharpoonright_{L^2(\mathbb{R})}\big)\det\left(1-\sum_{k=1}^3\alpha_k\otimes\beta_k\upharpoonright_{L^2(\mathbb{R})}\right).
\end{equation}
Here, $\alpha_j,\beta_k$ denote the six functions
\begin{align*}
	\alpha_1:=&\,\,\bar{\gamma}(1-\bar{\gamma}\chi_tT_{2n}\chi_t)^{-1}\chi_t\phi_{2n},\ \ \ \ \ \ \ \ \ \ \ \ \ \ \ \ \ \ \ \ \ \ \ \beta_1:=\chi_t\psi_{2n},\ \ \ \ \beta_2:=\delta_t-\delta_{\infty},\ \ \ \ \beta_3:=\delta_t,\smallskip\\
	\alpha_2:=&\,\,\gamma\langle\psi_{2n},(1-\gamma\chi_t)\epsilon_{\infty}\rangle(1-\bar{\gamma}\chi_tT_{2n}\chi_t)^{-1}\chi_t\phi_{2n}+\gamma(1-\bar{\gamma}\chi_tT_{2n}\chi_t)^{-1}\chi_tT_{2n}(1-\gamma\chi_t)\epsilon_{\infty},\smallskip\\
	\alpha_3:=&\,\,-\gamma\big\langle\psi_{2n},(1-\gamma\chi_t)\chi_{[t,\infty)}\big\rangle(1-\bar{\gamma}\chi_tT_{2n}\chi_t)^{-1}\chi_t\phi_{2n}-\gamma(1-\bar{\gamma}\chi_tT_{2n}\chi_t)^{-1}\chi_tT_{2n}(1-\gamma\chi_t)\chi_{[t,\infty)}.
\end{align*}
By general theory, cf. \cite[Chapter I, $(3.3)$]{GGK}
\begin{equation}\label{p:6}
	\det\left(1-\sum_{k=1}^3\alpha_k\otimes\beta_k\upharpoonright_{L^2(\mathbb{R})}\right)=\det\big[\delta_{jk}-\langle\alpha_j,\beta_k\rangle\big]_{j,k=1}^3,\ \ \ \ \delta_{jk}:=\begin{cases}1,&j=k\\ 0,&j\neq k\end{cases}
\end{equation}
with the $L^2(\mathbb{R})$ inner product $\langle\cdot,\cdot\rangle$. We conclude our finite $n$ calculation for even $n$ with the following further algebraic simplifications.
\begin{lem}\label{inp1} We have
\begin{equation*}
	\langle\alpha_1,\beta_1\rangle=\big\langle\chi_t\phi_{2n},\bar{\gamma}(1-\bar{\gamma}\chi_tT_{2n}\chi_t)^{-1}\chi_t\psi_{2n}\big\rangle,\ \ \ \ \ \ \ \ \langle\alpha_1,\beta_2\rangle=\bar{\gamma}\big((1-\bar{\gamma}\chi_tT_{2n}\chi_t)^{-1}\chi_t\phi_{2n}\big)(t)-\bar{\gamma}\phi_{2n}(\infty),
\end{equation*}
followed by
\begin{equation*}
	\langle\alpha_1,\beta_3\rangle=\bar{\gamma}\big((1-\bar{\gamma}\chi_tT_{2n}\chi_t)^{-1}\chi_t\phi_{2n}\big)(t).
\end{equation*}
Next, with $R_n:=\bar{\gamma}T_n\chi_t(1-\bar{\gamma}\chi_tT_n\chi_t)^{-1},n\in\mathbb{Z}_{\geq 1}$ which is well-defined as operator on $L^2(\mathbb{R})$ by Lemma \ref{lem:2} for any $t\geq 0$,
\begin{align*}
	\langle\alpha_2,\beta_1\rangle=&\,\frac{\gamma}{2\bar{\gamma}}\left[c_{2n}\langle\alpha_1,\beta_1\rangle-c_{2n}+\int_{-\infty}^{\infty}\big(1-\gamma\chi_{[t,\infty)}(x)\big)\big((1+R_{2n})\psi_{2n}\big)(x)\,\d x\right],\\
	 \langle\alpha_2,\beta_2\rangle=&\,\frac{\gamma}{2\bar{\gamma}}\left[c_{2n}\langle\alpha_1,\beta_2\rangle+\int_{-\infty}^{\infty}\big(1-\gamma\chi_{[t,\infty)}(x)\big)R_{2n}(x,t)\,\d x\right],\\
	\langle\alpha_2,\beta_3\rangle=&\,\frac{\gamma}{2\bar{\gamma}}\left[c_{2n}\langle\alpha_1,\beta_3\rangle+\int_{-\infty}^{\infty}\big(1-\gamma\chi_{[t,\infty)}(x)\big)R_{2n}(x,t)\,\d x\right],
\end{align*}
where $c_n:=\langle\psi_n,1-\gamma\chi_{[t,\infty)}\rangle$. Moreover
\begin{align*}
	\langle\alpha_3,\beta_1\rangle=&\,-\frac{\gamma}{\bar{\gamma}}\left[d_{2n}\langle\alpha_1,\beta_1\rangle-d_{2n}+\int_{-\infty}^{\infty}\big(1-\gamma\chi_{[t,\infty)}(x)\big)\chi_{[t,\infty)}(x)\big((1+R_{2n})\psi_{2n}\big)(x)\,\d x\right],\\
	\langle\alpha_3,\beta_2\rangle=&\,-\frac{\gamma}{\bar{\gamma}}\left[d_{2n}\langle\alpha_1,\beta_2\rangle+\int_{-\infty}^{\infty}\big(1-\gamma\chi_{[t,\infty)}(x)\big)\chi_{[t,\infty)}(x)R_{2n}(x,t)\,\d x\right],\\
	\langle\alpha_3,\beta_3\rangle=&\ -\frac{\gamma}{\bar{\gamma}}\left[d_{2n}\langle\alpha_1,\beta_3\rangle+\int_{-\infty}^{\infty}\big(1-\gamma\chi_{[t,\infty)}(x)\big)\chi_{[t,\infty)}(x)R_{2n}(x,t)\,\d x\right],
\end{align*}
with $d_n:=\langle\psi_n,(1-\gamma\chi_t)\chi_{[t,\infty)}\rangle$.
\end{lem}

\begin{proof} We use self-adjointness of the operator $T_n$ and write $R_n(x,t)$ for (cf. \cite[page $732$]{TW})
\begin{equation*}
	\lim_{\substack{y\rightarrow t \\ y>t}}R_n(x,y).
\end{equation*}
\end{proof}
With Lemma \ref{inp1} in place we finally evaluate the Fredholm determinant in \eqref{p:6}. Noting that the terms $c_n,d_n$ cancel out due to multilinearity of the finite-dimensional determinant, we obtain
\begin{equation}\label{p:7}
	\det\left(1-\sum_{k=1}^3\alpha_k\otimes\beta_k\upharpoonright_{L^2(\mathbb{R})}\right)=\det\begin{bmatrix}1-\langle\alpha_1,\beta_1\rangle & -\langle\alpha_1,\beta_2\rangle & -\langle\alpha_1,\beta_3\rangle\smallskip\\
	-\frac{\gamma}{2\bar{\gamma}}I_3 & 1-\frac{\gamma}{2\bar{\gamma}}I_1 & -\frac{\gamma}{2\bar{\gamma}}I_1\smallskip\\
	\frac{\gamma}{\bar{\gamma}}I_4 & \frac{\gamma}{\bar{\gamma}}I_2 & 1+\frac{\gamma}{\bar{\gamma}}I_2\end{bmatrix},
\end{equation}
in terms of the three inner products $\langle\alpha_1,\beta_k\rangle,k=1,2,3$ and the four integrals $I_j=I_j(t,\gamma,2n)$ with
\begin{align}
	I_1:=&\,\int_{-\infty}^{\infty}\big(1-\gamma\chi_{[t,\infty)}(x)\big)R_{2n}(x,t)\,\d x,\ \ \ \ \ \ \ \ \ \ \ I_2:=\int_{-\infty}^{\infty}\big(1-\gamma\chi_{[t,\infty)}(x)\big)\chi_{[t,\infty)}(x)R_{2n}(x,t)\,\d x,\nonumber\\
	I_3:=&\,\int_{-\infty}^{\infty}\big(1-\gamma\chi_{[t,\infty)}(x)\big)\big((1+R_{2n})\psi_{2n}\big)(x)\,\d x,\nonumber\\
	I_4:=&\,\int_{-\infty}^{\infty}\big(1-\gamma\chi_{[t,\infty)}(x)\big)\chi_{[t,\infty)}(x)\big((1+R_{2n})\psi_{2n}\big)(x)\,\d x.\label{p:8}
\end{align}
Identities \eqref{p:5} and \eqref{p:7} conclude our calculations for finite $n$, provided $n$ is even.
\begin{rem} The $3\times 3$ determinant \eqref{p:7} is the analogue of the GOE computation \cite[$(3.63)$]{D}.
\end{rem}
\subsection{The limit $n\rightarrow\infty$, $n$ even} In order to pass to the large $n$ limit we first shift the independent variable $t$ according to $t\mapsto t+\sqrt{2n}$, compare the left hand side of \eqref{e:4}. Under this scaling we have 
\begin{equation*}
	\det(1-\bar{\gamma}\chi_{t+\sqrt{2n}}\,T_{2n}\chi_{t+\sqrt{2n}}\upharpoonright_{L^2(\mathbb{R})})=\det\big(1-\bar{\gamma}\chi_t\widetilde{T}_{2n}\chi_t\upharpoonright_{L^2(\mathbb{R})}\big),
\end{equation*}
where $\widetilde{T}_n:L^2(\mathbb{R})\rightarrow L^2(\mathbb{R})$ has kernel
\begin{equation}\label{p:9}
	\widetilde{T}_n(x,y):=T_n(x+\sqrt{n},y+\sqrt{n}\,).
\end{equation}
Moreover, the entries in the $3\times 3$ determinant \eqref{p:7} transform in a similar fashion, for instance
\begin{equation*}
	\langle\alpha_1,\beta_1\rangle\mapsto\langle\widetilde{\alpha}_1,\widetilde{\beta}_1\rangle=\langle\chi_t\widetilde{\phi}_{2n},\bar{\gamma}(1-\bar{\gamma}\chi_t\widetilde{T}_{2n}\chi_t)^{-1}\chi_t\widetilde{\psi}_{2n}\rangle,
\end{equation*}
and likewise
\begin{align*}
	\langle\alpha_1,\beta_2\rangle\mapsto\langle\widetilde{\alpha}_1,\widetilde{\beta}_2\rangle=&\,\bar{\gamma}\big((1-\bar{\gamma}\chi_t\widetilde{T}_{2n}\chi_t)^{-1}\chi_t\widetilde{\phi}_{2n}\big)(t)-\bar{\gamma}\widetilde{\phi}_{2n}(\infty),\\
	\langle\alpha_1,\beta_3\rangle\mapsto\langle\widetilde{\alpha}_1,\widetilde{\beta}_3\rangle=&\,\bar{\gamma}\big((1-\bar{\gamma}\chi_t\widetilde{T}_{2n}\chi_t)^{-1}\chi_t\widetilde{\phi}_{2n}\big)(t),
\end{align*}
which involve $\widetilde{\phi}_n(x):=\phi_n(x+\sqrt{n}\,)$ and $\widetilde{\psi}_n(x):=\psi_n(x+\sqrt{n}\,)$. The remaining four integrals $I_k$, see \eqref{p:8}, are treated the same way and every occurrence of $R_n$ in them gets replaced by $\widetilde{R}_n$ with
\begin{equation*}
	\widetilde{R}_n=\bar{\gamma}\widetilde{T}_n\chi_t(1-\bar{\gamma}\chi_t\widetilde{T}_n\chi_t)^{-1}
\end{equation*}
defined in terms of $\widetilde{T}_n:L^2(\mathbb{R})\rightarrow L^2(\mathbb{R})$ with kernel \eqref{p:9}. At this point we collect a sequence of technical limits.
\begin{lem}\label{1:limit} Uniformly in $x\in\mathbb{R}$ chosen from compact subsets,
\begin{equation}\label{ex:1}
	\lim_{n\rightarrow\infty}\widetilde{\phi}_n(x)=\frac{1}{\sqrt{2\pi}}\int_{-\infty}^x\e^{-y^2}\,\d y\stackrel{\eqref{e:6}}{=}\frac{1}{\sqrt{2}}\,G(x),\ \ \ \ \ \ \ \ \lim_{n\rightarrow\infty}\widetilde{\psi}_n(x)=\frac{1}{\sqrt{2\pi}}\,\e^{-x^2}\stackrel{\eqref{e:6}}{=}\frac{1}{\sqrt{2}}\,g(x);
\end{equation}
and for any fixed $s\in\mathbb{R}$ with $p\in\{1,2\}$,
\begin{equation}\label{ex:2}
	\lim_{n\rightarrow\infty}\left\|\widetilde{\psi}_n-\frac{g}{\sqrt{2}}\right\|_{L^p(s,\infty)}=0,\ \ \ \ \ \ \ \  \lim_{n\rightarrow\infty}\left\|\widetilde{\phi}_n-\widetilde{\phi}_n(\infty)+\frac{1}{2\sqrt{2}}\,\textnormal{erfc}\right\|_{L^p(s,\infty)}=0.
\end{equation}
Here, $w=\textnormal{erfc}(z)$ denotes the complementary error function, cf. \cite[$7.2.2$]{NIST}.
\end{lem}
\begin{proof} The limits \eqref{ex:1}, \eqref{ex:2} are mentioned en route in \cite[page $1640$]{RS} and we thus only give a few details: as $n\rightarrow\infty$, uniformly in $x\in\mathbb{R}$,
\begin{equation*}
	\widetilde{\psi}_n(x)=\frac{1}{\sqrt{2\pi}}\left(1+\frac{x}{\sqrt{n}}\right)^{n-1}\e^{-x\sqrt{n}-\frac{1}{2}x^2}\Big(1+\mathcal{O}\big(n^{-1}\big)\Big).
\end{equation*}
But on compact subsets of $\mathbb{R}\ni x$,
\begin{equation*}
	\left(1+\frac{x}{\sqrt{n}}\right)^{n-1}\e^{-x\sqrt{n}}\,\,\,\stackrel{n\rightarrow\infty}{\longrightarrow}\,\,\,\e^{-\frac{1}{2}x^2},
\end{equation*}
which yields the second limit in \eqref{ex:1}. Since also for any $x>0$ and $n\in\mathbb{Z}_{\geq 2}$,
\begin{equation*}
	\left(1+\frac{x}{\sqrt{n}}\right)^{n-1}\e^{-x\sqrt{n}}\leq\e^{-x/\sqrt{n}}\leq 1,
\end{equation*}
the dominated convergence theorem yields the first $L^p(s,\infty)$ convergence in \eqref{ex:2}. For the limits involving $\widetilde{\phi}_n$, we note that as $n\rightarrow\infty$, uniformly in $x\in\mathbb{R}$,
\begin{equation*}
	\widetilde{\phi}_n(x)=\frac{1}{\sqrt{2}}\,P\left(\frac{1}{2}(n-1),\frac{1}{2}\big(x+\sqrt{n}\,\big)^2\right)\Big(1+\mathcal{O}\big(n^{-1}\big)\Big),
\end{equation*}
with the normalized incomplete gamma function $w=P(z)$, cf. \cite[$8.2.4$]{NIST}. But on compact subsets of $\mathbb{R}\ni x$, see \cite[$8.11.10$]{NIST},
\begin{equation*}
	P\left(\frac{1}{2}(n-1),\frac{1}{2}\big(x+\sqrt{n}\,\big)^2\right)\,\,\,\stackrel{n\rightarrow\infty}{\longrightarrow}\,\,\,\frac{1}{2}\,\textnormal{erfc}(-x)=\frac{1}{\sqrt{\pi}}\int_{-\infty}^x\e^{-y^2}\,\d y
\end{equation*}
which yields the first limit in \eqref{ex:1}. For the outstanding limit in \eqref{ex:2} we use that as $n\rightarrow\infty$, uniformly in $x\in\mathbb{R}$,
\begin{equation}\label{ex:3}
	\widetilde{\phi}_n(x)=\widetilde{\phi}_n(\infty)-\frac{1}{\sqrt{2}}\frac{\Gamma(\frac{1}{2}(n-1),\frac{1}{2}(x+\sqrt{n})^2)}{\Gamma(\frac{1}{2}(n-1))}\Big(1+\mathcal{O}\big(n^{-1}\big)\Big)
\end{equation}
with $w=\Gamma(a,z)$ the incomplete Gamma function, see \cite[$8.2.2$]{NIST}. But since for $x>0$ and $a\geq 1$ such that $x+1-a>0$,
\begin{equation*}
	\Gamma(a,x)=x^a\e^{-x}\int_0^{\infty}(1+y)^{a-1}\e^{-yx}\,\d y\leq x^a\e^{-x}\int_0^{\infty}\e^{y(a-1)}\e^{-yx}\,\d y=\frac{x^a\e^{-x}}{x+1-a},
\end{equation*}
we find for any $x>0$ and $n\in\mathbb{Z}_{\geq 2}$,
\begin{equation}\label{ex:4}
	\frac{\Gamma(\frac{1}{2}(n-1),\frac{1}{2}(x+\sqrt{n})^2)}{\Gamma(\frac{1}{2}(n-1))}\leq c\,\e^{-\frac{1}{2}x^2},\ \ \ \ c>0.
\end{equation}
Using also that on compact subsets of $\mathbb{R}\ni x$,
\begin{equation}\label{ex:5}
	\frac{\Gamma(\frac{1}{2}(n-1),\frac{1}{2}(x+\sqrt{n})^2)}{\Gamma(\frac{1}{2}(n-1))}=1-P\left(\frac{1}{2}(n-1),\frac{1}{2}\big(x+\sqrt{n}\,\big)^2\right)\,\,\,\stackrel{n\rightarrow\infty}{\longrightarrow}\,\,\,1-\frac{1}{2}\,\textnormal{erfc}(-x)=\frac{1}{2}\,\textnormal{erfc}(x),
\end{equation}
the second limit in \eqref{ex:2} follows from \eqref{ex:3}, \eqref{ex:4}, \eqref{ex:5} and the dominated convergence theorem. This completes our proof.
\end{proof}
The next limits concern the large $n$-behavior of the kernel function $\widetilde{T}_n(x,y)$ and its total integrals. Recall the kernel $T(x,y)$ defined in \eqref{e:5}.
\begin{lem}[{\cite[page $1642-1644$]{RS}}]\label{22:limit} Uniformly in $x,y\in\mathbb{R}$ chosen from compact subsets,
\begin{equation}\label{ex:6}
	\lim_{n\rightarrow\infty}\widetilde{T}_n(x,y)=T(x,y)\stackrel{\eqref{e:5}}{=}\frac{1}{\pi}\int_0^{\infty}\e^{-(x+u)^2}\e^{-(y+u)^2}\,\d u,\ \ \ \lim_{n\rightarrow\infty}\int_{-\infty}^{\infty}\widetilde{T}_n(x,y)\,\d y=\int_{-\infty}^{\infty}T(x,y)\,\d y.
\end{equation}
Moreover, for any fixed $s,t\in\mathbb{R}$ and with $p\in\{1,2\}$,
\begin{equation}\label{ex:7}
	\lim_{n\rightarrow\infty}\Big\|\widetilde{T}_n(\cdot,t)-T(\cdot,t)\Big\|_{L^p(s,\infty)}=0,\ \ \ \ \ \ \ \  \lim_{n\rightarrow\infty}\Big\|\int_{-\infty}^{\infty}\widetilde{T}_n(\cdot,y)\,\d y-\int_{-\infty}^{\infty}T(\cdot,y)\,\d y\Big\|_{L^p(s,\infty)}=0.
\end{equation}
\end{lem}
\begin{proof} The limits \eqref{ex:6} and \eqref{ex:7} follow from the detailed discussion on page $1642$ and $1643$ in \cite{RS}, see also \cite[$(4.16)$]{RS}. We omit details.
\end{proof}	
Finally we state the central convergence result for the operator $\chi_t\widetilde{T}_n\chi_t$ on $L^2(\mathbb{R})$.
\begin{lem}[{\cite[Lemma $4.2$]{RS}}]\label{3:limit} Given $t\in\mathbb{R}$, the operator $\chi_t\widetilde{T}_n\chi_t$ converges in trace norm on $L^2(\mathbb{R})$ and in $L^p(\mathbb{R})$ operator norm with $p\in\{1,2,\infty\}$ to the operator $\chi_tT\chi_t$. Additionally, for any $\gamma\in[0,1]$, 
\begin{equation}\label{ex:8}
	(1-\bar{\gamma}\chi_t\widetilde{T}_n\chi_t)^{-1}\,\,\,\stackrel{n\rightarrow\infty}{\longrightarrow}\,\,\,(1-\bar{\gamma}\chi_tT\chi_t)^{-1}
\end{equation}
in $L^p(\mathbb{R})$ operator norm with $p\in\{1,2,\infty\}$.
\end{lem}
\begin{proof} The convergences have been proven for $\gamma=1$ in \cite[Lemma $4.2$]{RS}. The extension to $\gamma\in[0,1)$ follows from the Neumann series expansion of the resolvents in \eqref{ex:8}, compare Lemma \ref{lem:2} and \cite[Lemma $2.1$]{BB}.
\end{proof}
We now apply Lemmas \ref{1:limit}, \ref{22:limit} and Lemma \ref{3:limit} in the large $n$ analysis of the Fredholm determinants back in \eqref{p:5}, after the rescaling $t\mapsto t+\sqrt{2n}$. First the leading factor:
\begin{lem}\label{0:limit} For any $\gamma\in[0,1]$ and $t\in\mathbb{R}$, 
\begin{equation*}
	\lim_{n\rightarrow\infty}\det\big(1-\bar{\gamma}\chi_t\widetilde{T}_{2n}\chi_t{\upharpoonright}_{L^2(\mathbb{R})}\big)=\det\big(1-\bar{\gamma}\chi_tT\chi_t{\upharpoonright}_{L^2(\mathbb{R})}\big).
\end{equation*}
\end{lem}
\begin{proof} We know from Lemma \ref{lem:2} that $\widetilde{T}_n$ is trace class on $L^2(t,\infty)$ and the same applies to $T$ (since it is a product of Hilbert-Schmidt operators). Thus with \cite[Chapter IV, $(5.14)$]{GGK},
\begin{align*}
	\Big|\det\big(1-\bar{\gamma}\chi_t\widetilde{T}_{2n}\chi_t{\upharpoonright}_{L^2(\mathbb{R})}\big)-\det\big(1-\bar{\gamma}\chi_tT\chi_t{\upharpoonright}_{L^2(\mathbb{R})}\big)\Big|\leq&\,\bar{\gamma}\,\|\chi_t\widetilde{T}_{2n}\chi_t-\chi_tT\chi_t\|_{1}\\
	&\,\times\exp\big(1+\bar{\gamma}\|\chi_t\widetilde{T}_{2n}\chi_t\|_1+\bar{\gamma}\|\chi_tT\chi_t\|_1\big).
\end{align*}
But the operator difference in trace norm converges to zero by Lemma \ref{3:limit} and $\|\chi_t\widetilde{T}_{2n}\chi_t\|_1$ remains bounded by the same result. This completes our proof.
\end{proof}
Next we move on to the $L^2(\mathbb{R})$ inner products which appear in \eqref{p:7}.
\begin{lem}\label{4:limit} For any $\gamma\in[0,1]$ and $t\in\mathbb{R}$, 
\begin{eqnarray*}
	\lim_{n\rightarrow\infty}\langle\widetilde{\alpha}_1,\widetilde{\beta}_1\rangle&=&\frac{\bar{\gamma}}{2}\int_t^{\infty}G(x)\big((1-\bar{\gamma}T\chi_t{\upharpoonright}_{L^2(\mathbb{R})})^{-1}g\big)(x)\,\d x,\\
	\lim_{n\rightarrow\infty}\langle\widetilde{\alpha}_1,\widetilde{\beta}_2\rangle&=&\frac{\bar{\gamma}}{\sqrt{2}}\big((1-\bar{\gamma}T\chi_t{\upharpoonright}_{L^2(\mathbb{R})})^{-1}G\big)(t)-\frac{\bar{\gamma}}{\sqrt{2}},\\ 
	\lim_{n\rightarrow\infty}\langle\widetilde{\alpha}_1,\widetilde{\beta}_3\rangle&=&\frac{\bar{\gamma}}{\sqrt{2}}\big((1-\bar{\gamma}T\chi_t{\upharpoonright}_{L^2(\mathbb{R})})^{-1}G\big)(t).
\end{eqnarray*}
\end{lem}
\begin{proof} In the first inner product we write
\begin{align}
	\langle\widetilde{\alpha}_1,\widetilde{\beta}_1\rangle=&\,\left\langle\chi_t\big(\widetilde{\phi}_{2n}(\infty)-\frac{1}{2\sqrt{2}}\,\textnormal{erfc}\big),\bar{\gamma}(1-\bar{\gamma}\chi_t\widetilde{T}_{2n}\chi_t)^{-1}\chi_t\widetilde{\psi}_{2n}\right\rangle\nonumber\\
	&\ \ \ \ \ \ +\left\langle\chi_t\big(\widetilde{\phi}_{2n}-\widetilde{\phi}_{2n}(\infty)+\frac{1}{2\sqrt{2}}\,\textnormal{erfc}\big),\bar{\gamma}(1-\bar{\gamma}\chi_t\widetilde{T}_{2n}\chi_t)^{-1}\chi_t\widetilde{\psi}_{2n}\right\rangle\label{l:5}
\end{align}
and now use that, uniformly in $x\in\mathbb{R}$,
\begin{equation*}
	\lim_{n\rightarrow\infty}\left(\widetilde{\phi}_{2n}(\infty)-\frac{1}{2\sqrt{2}}\,\textnormal{erfc}(x)\right)=\frac{1}{\sqrt{2}}-\frac{1}{2\sqrt{2}}\,\textnormal{erfc}(x)\stackrel{\eqref{e:6}}{=}\frac{1}{\sqrt{2}}\,G(x).
\end{equation*}
But from Lemma \ref{1:limit} and Lemma \ref{3:limit} we also know that for $p\in\{1,2\}$,
\begin{equation*}
	\lim_{n\rightarrow\infty}\left\|(1-\bar{\gamma}\chi_t\widetilde{T}_{2n}\chi_t)^{-1}\chi_t\widetilde{\psi}_{2n}-\frac{1}{\sqrt{2}}(1-\bar{\gamma}\chi_tT\chi_t)^{-1}\chi_t\,g\right\|_{L^p(\mathbb{R})}=0,
\end{equation*}
so with \eqref{ex:2} and H\"older's inequality therefore back in \eqref{l:5}
\begin{equation*}
	\lim_{n\rightarrow\infty}\langle\widetilde{\alpha}_1,\widetilde{\beta}_1\rangle=\left\langle\chi_t\frac{1}{\sqrt{2}}G,\bar{\gamma}(1-\bar{\gamma}\chi_tT\chi_t)^{-1}\frac{1}{\sqrt{2}}\,g\right\rangle=\frac{\bar{\gamma}}{2}\int_t^{\infty}G(x)\big((1-\bar{\gamma}T\chi_t\upharpoonright_{L^2(\mathbb{R})})^{-1}g\big)(x)\,\d x,
\end{equation*}
as claimed. For the second inner product we write instead (with $\chi_{[t,\infty)}(t)=1$)
\begin{align*}
	\langle\widetilde{\alpha}_1,\widetilde{\beta}_2\rangle=&\,\bar{\gamma}\big(\widetilde{\phi}_{2n}(t)-\widetilde{\phi}_{2n}(\infty)\big)+\bar{\gamma}^2\big(\widetilde{T}_{2n}\chi_t(1-\bar{\gamma}\chi_t\widetilde{T}_{2n}\chi_t)^{-1}\chi_t\,\widetilde{\phi}_{2n}\big)(t)\\
	=&\,\,\bar{\gamma}\big(\widetilde{\phi}_{2n}(t)-\widetilde{\phi}_{2n}(\infty)\big)+\bar{\gamma}^2\int_t^{\infty}\widetilde{T}_{2n}(t,x)\big((1-\bar{\gamma}\chi_t\widetilde{T}_{2n}\chi_t)^{-1}\chi_t\,\widetilde{\phi}_{2n}\big)(x)\,\d x
\end{align*}
and recall the previous decomposition of $\widetilde{\phi}_n$ used in \eqref{l:5}. Hence, with Lemma \ref{3:limit} and \eqref{ex:1}, \eqref{ex:7} we find from H\"older's inequality,
\begin{align*}
	\lim_{n\rightarrow\infty}\langle\widetilde{\alpha}_1,&\,\widetilde{\beta}_2\rangle=\frac{\bar{\gamma}}{\sqrt{2}}G(t)-\frac{\bar{\gamma}}{\sqrt{2}}+\frac{\bar{\gamma}^2}{\sqrt{2}}\int_t^{\infty}T(t,x)\big((1-\bar{\gamma}\chi_tT\chi_t)^{-1}\chi_t\,G\big)(x)\,\d x\\
	=&\,\frac{\bar{\gamma}}{\sqrt{2}}G(t)-\frac{\bar{\gamma}}{\sqrt{2}}+\frac{\bar{\gamma}}{\sqrt{2}}\big(\bar{\gamma}\chi_tT\chi_t(1-\bar{\gamma}\chi_tT\chi_t)^{-1}\chi_t\,G\big)(t)=\frac{\bar{\gamma}}{\sqrt{2}}\big((1-\bar{\gamma}T\chi_t{\upharpoonright}_{L^2(\mathbb{R})})^{-1}G\big)(t)-\frac{\bar{\gamma}}{\sqrt{2}}.
\end{align*}
The derivation of the third inner product is completely analogous.
\end{proof}
At this point we are left with the computation of the large $n$ limits of the rescaled integrals $I_k$. Let $R(x,y)$ denote the kernel of the resolvent $R=\bar{\gamma}T\chi_t(1-\bar{\gamma}\chi_tT\chi_t)^{-1}$ on $L^2(\mathbb{R})$.
\begin{lem}\label{2:limit} For every $\gamma\in[0,1]$ and $t\in\mathbb{R}$,
\begin{eqnarray*}
	\lim_{n\rightarrow\infty}I_1(t+\sqrt{2n},\gamma,2n)&=&\int_{-\infty}^{\infty}\big(1-\gamma\chi_{[t,\infty)}(x)\big)R(x,t)\,\d x,\\
	\lim_{n\rightarrow\infty}I_2(t+\sqrt{2n},\gamma,2n)&=&\int_{-\infty}^{\infty}\big(1-\gamma\chi_{[t,\infty)}(x)\big)\chi_{[t,\infty)}(x)R(x,t)\,\d x,\\
	\lim_{n\rightarrow\infty}I_3(t+\sqrt{2n},\gamma,2n)&=&\frac{1}{\sqrt{2}}\int_{-\infty}^{\infty}\big(1-\gamma\chi_{[t,\infty)}(x)\big)\big((1-\bar{\gamma}T\chi_t{\upharpoonright}_{L^2(\mathbb{R})})^{-1}g\big)(x)\,\d x-\frac{1}{\sqrt{2}},\\
	\lim_{n\rightarrow\infty}I_4(t+\sqrt{2n},\gamma,2n)&=&\frac{1}{\sqrt{2}}\int_{-\infty}^{\infty}\big(1-\gamma\chi_{[t,\infty)}(x)\big)\chi_{[t,\infty)}(x)\big((1-\bar{\gamma}T\chi_t{\upharpoonright}_{L^2(\mathbb{R})})^{-1}g\big)(x)\,\d x.
\end{eqnarray*}
\end{lem}
\begin{proof} We begin with the kernel function identity (cf. \cite[page $748$]{TW}),
\begin{align*}
	\widetilde{R}_n(x,t)=&\,\bar{\gamma}\big(\widetilde{T}_n\chi_t(1-\bar{\gamma}\chi_t\widetilde{T}_n\chi_t)^{-1}\big)(x,t)=\big((1-\bar{\gamma}\widetilde{T}_n\chi_t)^{-1}\bar{\gamma}\widetilde{T}_n\chi_t\big)(x,t)\\
	=&\,\bar{\gamma}\widetilde{T}_n(x,t)+\bar{\gamma}\big(\widetilde{T}_n\chi_t(1-\bar{\gamma}\chi_t\widetilde{T}_n\chi_t)^{-1}\bar{\gamma}\widetilde{T}_n\big)(x,t),\ \ \ \ x\in\mathbb{R},
\end{align*}
which, upon insertion into the integrand of $I_1(t+\sqrt{2n},\gamma,2n)$, leads to four integrals,
\begin{equation}\label{ex:9}
	\int_t^{\infty}\widetilde{T}_{2n}(x,t)\,\d x,\ \ \ \ \ \int_{-\infty}^t\widetilde{T}_{2n}(x,t)\,\d x,\ \ \ \ \
	\int_t^{\infty}\big(\widetilde{T}_{2n}\chi_t(1-\bar{\gamma}\chi_t\widetilde{T}_{2n}\chi_t)^{-1}\widetilde{T}_{2n}\big)(x,t)\,\d x,
\end{equation}
and
\begin{equation*}
	\int_{-\infty}^t\big(\widetilde{T}_{2n}\chi_t(1-\bar{\gamma}\chi_t\widetilde{T}_{2n}\chi_t)^{-1}\widetilde{T}_{2n}\big)(x,t)\,\d x.
\end{equation*}
Apply Lemma \ref{2:limit} and conclude for the first two integrals
\begin{align*}
	\lim_{n\rightarrow\infty}\int_t^{\infty}\widetilde{T}_{2n}(x,t)\,&\d x\stackrel{\eqref{ex:7}}{=}\int_t^{\infty}T(x,t)\,\d x,\\
	 \lim_{n\rightarrow\infty}\int_{-\infty}^t\widetilde{T}_{2n}(x,t)\,&\d x=\lim_{n\rightarrow\infty}\left[\int_{-\infty}^{\infty}\widetilde{T}_{2n}(x,t)\,\d x-\int_t^{\infty}\widetilde{T}_{2n}(x,t)\,\d x\right]
	\underset{\eqref{ex:7}}{\overset{\eqref{ex:6}}{=}}\int_{-\infty}^tT(x,t)\,\d x.
\end{align*}
For the third integral in \eqref{ex:9} we write
\begin{equation*}
	\int_t^{\infty}\big(\widetilde{T}_{2n}\chi_t(1-\bar{\gamma}\chi_t\widetilde{T}_{2n}\chi_t)^{-1}\widetilde{T}_{2n}\big)(x,t)\,\d x=\left\langle\chi_{[t,\infty)}\int_t^{\infty}\widetilde{T}_{2n}(x,\cdot)\,\d x,\big((1-\bar{\gamma}\chi_t\widetilde{T}_{2n}\chi_t)^{-1}\chi_t\widetilde{T}_{2n}\big)(\cdot,t)\right\rangle,
\end{equation*}
and note that each entry of the last $L^2(\mathbb{R})$ inner product converges to its formal limits in $L^2(\mathbb{R})$ sense, cf. Lemma \ref{3:limit}, equation \eqref{ex:7} and the workings in \cite[page $1643$]{RS}. The outstanding fourth integral is treated similarly, the difference being that the first entry in the corresponding $L^2(\mathbb{R})$ inner product equals
\begin{equation*}
	\chi_{[t,\infty)}\int_{-\infty}^t\widetilde{T}_{2n}(x,\cdot)\,\d x=\chi_{[t,\infty)}\int_{-\infty}^{\infty}\widetilde{T}_{2n}(x,\cdot)\,\d x-\chi_{[t,\infty)}\int_t^{\infty}\widetilde{T}_{2n}(x,\cdot)\,\d x.
\end{equation*}
Since both terms converge to their formal limits in $L^2(\mathbb{R})$ sense, compare our reasoning above and \eqref{ex:7}, we find all together,
\begin{equation*}
	\lim_{n\rightarrow\infty}I_1(t+\sqrt{2n},\gamma,2n)=\int_{-\infty}^{\infty}\big(1-\gamma\chi_{[t,\infty)}(x)\big)\big(\bar{\gamma}T\chi_t(1-\bar{\gamma}\chi_tT\chi_t)^{-1}\big)(x,t)\,\d x,
\end{equation*}
which is the desired formula for $I_1$, given that $R=\bar{\gamma}T\chi_t(1-\bar{\gamma}\chi_tT\chi_t)^{-1}$. The derivation of the limit for $I_2$ is completely analogous and in fact simpler since no integrals over $(-\infty,t)$ occur. Moving ahead, the limit evaluation of $I_3(t+\sqrt{2n},\gamma,2n)$ also requires four integrals,
\begin{equation}\label{ex:10}
	\int_t^{\infty}\widetilde{\psi}_{2n}(x)\,\d x,\ \ \ \ \int_{-\infty}^t\widetilde{\psi}_{2n}(x)\,\d x,\ \ \ \ \int_t^{\infty}\big(\widetilde{T}_{2n}\chi_t(1-\bar{\gamma}\chi_t\widetilde{T}_{2n}\chi_t)^{-1}\widetilde{\psi}_{2n}\big)(x)\,\d x,
\end{equation}
\begin{equation*}
	\int_{-\infty}^t\big(\widetilde{T}_{2n}\chi_t(1-\bar{\gamma}\chi_t\widetilde{T}_{2n}\chi_t)^{-1}\widetilde{\psi}_{2n}\big)(x)\,\d x.
\end{equation*}
Note that
\begin{align*}
	\lim_{n\rightarrow\infty}\int_t^{\infty}\widetilde{\psi}_{2n}(x)\,&\d x\stackrel{\eqref{ex:2}}{=}\frac{1}{\sqrt{2}}\int_t^{\infty}g(x)\,\d x,\\
	\lim_{n\rightarrow\infty}\int_{-\infty}^t\widetilde{\psi}_{2n}(x)\,&\d x=\lim_{n\rightarrow\infty}\left[\int_{-\infty}^{\infty}\widetilde{\psi}_{2n}(x)\,\d x-\int_t^{\infty}\widetilde{\psi}_{2n}(x)\,\d x\right]=-\lim_{n\rightarrow\infty}\int_t^{\infty}\widetilde{\psi}_{2n}(x)\,\d x,
\end{align*}
since $\psi_{2n}$ is an odd function. Also
\begin{equation*}
	\int_t^{\infty}\big(\widetilde{T}_{2n}\chi_t(1-\bar{\gamma}\chi_t\widetilde{T}_{2n}\chi_t)^{-1}\widetilde{\psi}_{2n}\big)(x)\,\d x=\left\langle\chi_{[t,\infty)}\int_t^{\infty}\widetilde{T}_{2n}(x,\cdot)\,\d x,\big((1-\bar{\gamma}\chi_t\widetilde{T}_{2n}\chi_t)^{-1}\chi_t\widetilde{\psi}_{2n}\big)(\cdot)\right\rangle,
\end{equation*}
which converges to its formal limit as $n\rightarrow\infty$, compare our reasoning for $I_1$ and \eqref{ex:2}. The same is true for the remaining fourth integral and we obtain all together, as $n\rightarrow\infty$,
\begin{align*}
	I_3(t+\sqrt{2n},\gamma,2n)\rightarrow(1-\gamma)\frac{1}{\sqrt{2}}\int_t^{\infty}\big((1+R)\,g\big)(x)\,\d x+\frac{1}{\sqrt{2}}\int_{-\infty}^t\big((1+R)\,g\big)(x)\,\d x-\frac{1}{\sqrt{2}}\int_{-\infty}^{\infty}g(x)\,\d x,
\end{align*}
which is the claimed identity. The derivation for $I_4$ is again similar and does not use any integrals along $(-\infty,t)$. This completes our proof.
\end{proof}
With Lemma \ref{0:limit}, \ref{4:limit} and \ref{2:limit} in place, we now obtain the following result. 
\begin{prop}\label{Jprop:1} As $n\rightarrow\infty$, uniformly for $t\in\mathbb{R}$ chosen from compact subsets and any $\gamma\in[0,1]$,
\begin{equation}\label{ex:11}
	F_{2n}(t+\sqrt{2n},\gamma)\rightarrow\det\big(1-\bar{\gamma}\chi_tT\chi_t{\upharpoonright}_{L^2(\mathbb{R})}\big)\det\begin{bmatrix}1-u & -v+\frac{\bar{\gamma}}{\sqrt{2}} & -v\smallskip\\
	-\frac{\gamma}{2\bar{\gamma}}(p+q-\frac{1}{\sqrt{2}}) & 1-\frac{\gamma}{2\bar{\gamma}}(r+w) & -\frac{\gamma}{2\bar{\gamma}}(r+w)\smallskip\\
	\frac{\gamma}{\bar{\gamma}}\,p & \frac{\gamma}{\bar{\gamma}}\,r & 1+\frac{\gamma}{\bar{\gamma}}\,r\end{bmatrix},
\end{equation}
where $u,v,p,q,r,w$ are the following six functions of $(t,\gamma)\in\mathbb{R}\times[0,1]$,
\begin{align}
	u:=&\,\,\frac{\bar{\gamma}}{2}\int_t^{\infty}G(x)\big((1-\bar{\gamma}T\chi_t{\upharpoonright}_{L^2(\mathbb{R})})^{-1}g\big)(x)\,\d x,\ \ \ \ \ \ \ v:=\frac{\bar{\gamma}}{\sqrt{2}}\big((1-\bar{\gamma}T\chi_t{\upharpoonright}_{L^2(\mathbb{R})})^{-1}G\big)(t),\label{extra:1}\\
	p:=&\,\,\frac{1-\gamma}{\sqrt{2}}\int_t^{\infty}\big((1-\bar{\gamma}T\chi_t{\upharpoonright}_{L^2(\mathbb{R})})^{-1}g\big)(x)\,\d x,\ \ \ \ \ \ \ \ q:=\frac{1}{\sqrt{2}}\int_{-\infty}^t\big((1-\bar{\gamma}T\chi_t{\upharpoonright}_{L^2(\mathbb{R})})^{-1}g\big)(x)\,\d x,\nonumber\\
	r:=&\,\,(1-\gamma)\int_t^{\infty}R(x,t)\,\d x,\ \ \ \ \ \ \ \ \ \ \ \ \ \ \ \ \ \ \ \ \ \ \ \ \ \ \ \ \ \ \ w:=\int_{-\infty}^tR(x,t)\,\d x.\nonumber
\end{align}
\end{prop}

\subsection{The limit \eqref{e:13} for odd $n$}\label{sec:odd}
In this subsection we will compute the limit $F_{2n+1}(t+\sqrt{2n+1},\gamma)$ using a comparison argument. Precisely, we show how the computations in Subsection \ref{evennsec} have to be modified in order to account for odd $n\in\mathbb{Z}_{\geq 3}$. These additional manipulations are necessary given the different structure of the operator $IS_n$ in Proposition \ref{recall} for odd $n$. The details are as follows. We first relate $\K_{2n+1}$ to $\K_{2n}$:
\begin{prop} For any $n\in\mathbb{Z}_{\geq 3}$,
\begin{equation}\label{f:2}
	S_n=S_{n-1}+\psi_n\otimes\phi_n-f_n\psi_{n-1}\otimes\phi_{n+1},\ \ \ \ f_n^2:=\frac{n-1}{n-2}\sqrt{\frac{n-1}{n+1}}\stackrel{n\rightarrow\infty}{\longrightarrow} 1,\ \ \ f_n>0,
\end{equation}
and thus in turn,
\begin{equation}\label{f:3}
	IS_{2n+1}=\epsilon S_{2n}+\epsilon\psi_{2n+1}\otimes\phi_{2n+1}-\phi_{2n+1}\otimes\epsilon\psi_{2n+1},\ \ n\in\mathbb{Z}_{\geq 1}.
\end{equation}
\end{prop}
\begin{proof} Identity \eqref{f:2} follows from the equality
\begin{equation}\label{f:4}
	S_n(x,y)=S_{n-1}(x,y)+\frac{1}{\sqrt{2\pi}\,(n-2)!}\int_0^y(x-u)(xu)^{n-2}\e^{-\frac{1}{2}(x^2+u^2)}\,\d u,\ \ n\in\mathbb{Z}_{\geq 3},
\end{equation}
which appears in \cite[page $1628$]{RS} and which can be proven by induction on $n\in\mathbb{Z}_{\geq 3}$ using the original definition of $S_n(x,y)$ given in Proposition \ref{recall}. Once \eqref{f:4} is known we find immediately \eqref{f:2} by comparison with \eqref{p:4}. On the other hand,
\begin{align}
	(\epsilon S_{2n+1})&\,(x,y)=(\epsilon S_{2n})(x,y)-\frac{1}{\sqrt{2\pi}\,(2n-1)!}\int_0^xz^{2n}\e^{-\frac{1}{2}z^2}\,\d z\,\int_0^yu^{2n-1}\e^{-\frac{1}{2}u^2}\,\d u\nonumber\\
	&\,-\frac{1}{\sqrt{2\pi}\,(2n-1)!}\int_x^{\infty}z^{2n-1}\e^{-\frac{1}{2}z^2}\,\d z\int_0^yu^{2n}\e^{-\frac{1}{2}u^2}\,\d u
	=(\epsilon S_{2n})(x,y)+(\epsilon\psi_{2n+1})(x)\phi_{2n+1}(y)\nonumber\\
	&\,-\phi_{2n+1}(x)(\epsilon\psi_{2n+1})(y)-\frac{2^{n-1}}{\sqrt{2\pi}}\frac{\Gamma(n)}{\Gamma(2n)}\int_0^yu^{2n}\e^{-\frac{1}{2}u^2}\,\d u,\ \ \ \ n\in\mathbb{Z}_{\geq 1},\label{f:5}
\end{align}
which used \eqref{p:4} and
\begin{equation*}
	(\epsilon\psi_{2n+1})(x)=-\int_0^x\psi_{2n+1}(y)\,\d y,
\end{equation*}
in the last equality. However, by the Legendre duplication formula \cite[$5.5.5$]{NIST}, for any $n\in\mathbb{Z}_{\geq 1}$,
\begin{equation*}
	\frac{2^{n-1}}{\sqrt{2\pi}}\frac{\Gamma(n)}{\Gamma(2n)}=\frac{1}{2^{n+\frac{1}{2}}\Gamma(n+\frac{1}{2})},
\end{equation*}
so \eqref{f:5} yields
\begin{equation*}
	(\epsilon S_{2n+1})(x,y)+\frac{1}{2^{n+\frac{1}{2}}\Gamma(n+\frac{1}{2})}\int_0^yu^{2n}\e^{-\frac{1}{2}u^2}\,\d u=\big(\epsilon S_{2n}+\epsilon\psi_{2n+1}\otimes\phi_{2n+1}-\phi_{2n+1}\otimes\epsilon\psi_{2n+1}\big)(x,y),
\end{equation*}
and this is \eqref{f:3} after comparison with the kernel of $IS_{2n+1}$ written in Proposition \ref{recall}.
\end{proof}
Inserting \eqref{f:2} and \eqref{f:3} into formula \eqref{e:2} for $\K_{2n+1}$, we find that $\K_{2n+1}=\K_{2n}+\E_{2n}$ where the operator $\E_n$ has kernel
\begin{equation*}
	\E_n:=\begin{bmatrix}\rho^{-1}(\psi_{n+1}\otimes\phi_{n+1}-f_{n+1}\psi_n\otimes\phi_{n+2})\rho & \rho^{-1}(D\phi_{n+1}\otimes\psi_{n+1}-f_{n+1}D\phi_{n+2}\otimes\psi_n)\rho^{-1}\smallskip\\
	\rho(\epsilon\psi_{n+1}\otimes\phi_{n+1}-\phi_{n+1}\otimes\epsilon\psi_{n+1})\rho & \rho(\phi_{n+1}\otimes\psi_{n+1}-f_{n+1}\phi_{n+2}\psi_n)\rho^{-1}\end{bmatrix}.
\end{equation*}
Note that $\chi_t\E_n\chi_t$ is finite rank on $L^2(\mathbb{R})\oplus L^2(\mathbb{R})$, so in particular trace class. Also, since for any $x\in\mathbb{R}$,
\begin{equation*}
	(\epsilon\psi_{n+1})(x)=-\underbrace{\sqrt{\frac{n}{\sqrt{(n+1)(n+2)}}}}_{\rightarrow 1,\ n\rightarrow\infty}\,\phi_{n+2}(x),\ \ \ \ \ \ (D\phi_{n+1})(x)=\underbrace{\sqrt{\frac{\sqrt{n(n+1)}}{n-1}}}_{\rightarrow 1\ n\rightarrow\infty}\,\psi_n(x),\ \ \ \ \ \ n\in\mathbb{Z}_{\geq 2},
\end{equation*}
Lemma \ref{1:limit} and triangle inequality yield that, in trace norm,
\begin{equation}\label{f:6}
	\|\chi_{t+\sqrt{n}}\E_n\chi_{t+\sqrt{n}}\|_1\rightarrow 0\ \ \ \textnormal{as}\ \ n\rightarrow\infty.
\end{equation}
But $1-\gamma\chi_{t+\sqrt{2n}}\K_{2n}\chi_{t+\sqrt{2n}}$ is invertible for sufficiently large $n$ and any $(t,\gamma)\in\mathbb{R}\times[0,1]$ by the working of Section \ref{sec:even} and Remark \ref{crucrig}. Hence we use \eqref{HC:1} and obtain for $n\geq n_0$,
\begin{align*}
	F_{2n+1}&(t+\sqrt{2n},\gamma)=\det_2(1-\gamma\chi_{t+\sqrt{2n}}\K_{2n}\chi_{t+\sqrt{2n}}-\gamma\chi_{t+\sqrt{2n}}\E_{2n}\chi_{t+\sqrt{2n}}{\upharpoonright}_{L^2(\mathbb{R})\oplus L^2(\mathbb{R})})\\
	=&\,F_{2n}(t+\sqrt{2n},\gamma)\det\big(1-\gamma(1-\gamma\chi_{t+\sqrt{2n}}\K_{2n}\chi_{t+\sqrt{2n}})^{-1}\chi_{t+\sqrt{2n}}\E_{2n}\chi_{t+\sqrt{2n}}{\upharpoonright}_{L^2(\mathbb{R})\oplus L^2(\mathbb{R})}\big),
\end{align*}
where the second (finite rank) determinant converges to one as $n\rightarrow\infty$ because of \eqref{f:6}. This shows that
\begin{align}
	\lim_{n\rightarrow\infty}F_{2n+1}(t+\sqrt{2n},\gamma)=&\lim_{n\rightarrow\infty}F_{2n}(t+\sqrt{2n},\gamma)\label{f:7}
\end{align}
for any $(t,\gamma)\in\mathbb{R}\times[0,1]$. In fact, the above convergence is uniform in $(t,\gamma)\in\mathbb{R}\times[0,1]$ chosen from compact subsets and since $F_{2n+1}(t+\sqrt{2n},\gamma)$ is at least differentiable in $t\in\mathbb{R}$ (this can be seen directly from \eqref{e:1} by scaling $t$ into the kernel and then using the logic behind \cite[Lemma $2.20$]{ACQ}), we find
\begin{equation}\label{f:8}
	F_{2n+1}(t+\sqrt{2n},\gamma)=F_{2n+1}(t+\sqrt{2n+1},\gamma)+o(1),\ \ \ n\rightarrow\infty
\end{equation}
on compact subsets of $(t,\gamma)\in\mathbb{R}\times[0,1]$. Hence, combining \eqref{f:7} with \eqref{f:8} we arrive at the analogue of \eqref{ex:11} for odd $n$, i.e.
\begin{prop}\label{Jprop:2} Proposition \ref{Jprop:1} holds with $F_{2n}(t+\sqrt{2n},\gamma)$ in the left hand side of \eqref{ex:11} replaced by $F_{2n+1}(t+\sqrt{2n+1},\gamma)$.
\end{prop}
Finally, merging Propositions \ref{Jprop:1} and \ref{Jprop:2} we have now established the existence of the limit \eqref{e:13}. This completes the current section.

\section{Proof of Theorem \ref{main:1} - final steps}\label{oddsec}
In order to prove the outstanding representation \eqref{e:16} we now find a new representation for the $3\times 3$ determinant in \eqref{ex:11}. To begin with, we list four algebraic relations between the functions $u,v,p,q,r$ and $w$ in Corollary \ref{JBcor} below. These follow from the next Lemma. Recall $R=(1-\bar{\gamma}T\chi_t\upharpoonright_{L^2(\mathbb{R})})^{-1}-1$ and the definitions of $g$ and $G$ in \eqref{e:6}.
\begin{lem}\label{cur} For every $(t,\gamma)\in\mathbb{R}\times[0,1]$, 
\begin{align*}
	\big((1-\bar{\gamma}T\chi_t{\upharpoonright}_{L^2(\mathbb{R})})^{-1}G\big)(t)=&\,\,\int_{-\infty}^t\big((1-\bar{\gamma}T\chi_t{\upharpoonright}_{L^2(\mathbb{R})})^{-1}g\big)(x)\,\d x,\\
	\int_t^{\infty}R(x,t)\,\d x=&\,\,\bar{\gamma}\int_t^{\infty}\big(1-G(x)\big)\big((1-\bar{\gamma}T\chi_t{\upharpoonright}_{L^2(\mathbb{R})})^{-1}g\big)(x)\,\d x,\\
	\int_{-\infty}^tR(x,t)\,\d x=&\,\,\bar{\gamma}\int_t^{\infty}G(x)\big((1-\bar{\gamma}T\chi_t{\upharpoonright}_{L^2(\mathbb{R})})^{-1}g\big)(x)\,\d x,\\
	1+\int_t^{\infty}R(x,t)\,\d x=&\int_{-\infty}^{\infty}\big((1-\bar{\gamma}T\chi_t\upharpoonright_{L^2(\mathbb{R})})^{-1}g\big)(x)\,\d x.
\end{align*}
\end{lem}
\begin{proof} The first equality follows from \cite[$(4.9)$]{BB} with the formal replacements $\gamma\mapsto\bar{\gamma},G^{\gamma}\mapsto G$ and $g^{\gamma}\mapsto g$, see \cite[$(4.3)$]{BB}. The second and third are a consequence of \eqref{a:2} below. Indeed, we have
\begin{align*}
	\int_t^{\infty}R(x,t)\,\d x=&\,\int_0^{\infty}R(x+t,t)\,\d x=\int_0^{\infty}\big(\bar{\gamma}T\chi_t(1-\bar{\gamma}\chi_tT\chi_t{\upharpoonright}_{L^2(\mathbb{R})})^{-1}\big)(x+t,t)\,\d x\\
	=&\,\sum_{k=1}^{\infty}\bar{\gamma}^k\int_0^{\infty}(T\chi_t{\upharpoonright}_{L^2(\mathbb{R})})^k(x+t,t)\,\d x,
\end{align*}
and, similarly,
\begin{align*}
	\int_{-\infty}^tR(x,t)\,\d x=&\,\int_{-\infty}^0R(x+t,t)\,\d x=\int_{-\infty}^0\big(\bar{\gamma}T\chi_t(1-\bar{\gamma}\chi_tT\chi_t{\upharpoonright}_{L^2(\mathbb{R})})^{-1}\big)(x+t,t)\,\d x\\
	=&\,\sum_{k=1}^{\infty}\bar{\gamma}^k\int_{-\infty}^0(T\chi_t{\upharpoonright}_{L^2(\mathbb{R})})^k(x+t,t)\,\d x.
\end{align*}
Now choose $K=T$ (which is self-adjoint since $\phi(x)=\psi(x)=g(x)$ in \eqref{a:0}) and $I=(0,\infty)$ in \eqref{a:2}, so that
\begin{align*}
	\int_t^{\infty}R(x,t)\,\d x=&\,\sum_{k=1}^{\infty}\bar{\gamma}^k\int_t^{\infty}\left[\int_x^{\infty}g(v)\,\d v\right]\big((T\chi_t{\upharpoonright}_{L^2(\mathbb{R})})^{k-1}g\big)(x)\,\d x\\
	=&\,\bar{\gamma}\int_t^{\infty}\big(1-G(x)\big)\big((1-\bar{\gamma}T\chi_t{\upharpoonright}_{L^2(\mathbb{R})})^{-1}g\big)(x)\,\d x,
\end{align*}
which is the second integral identity. For the third, we simply choose $I=(-\infty,0)$ in \eqref{a:2} and for the fourth we use \eqref{c:2}, self-adjointness of $T$ and $\int_{-\infty}^{\infty}g(x)\,\d x=1$ to find that
\begin{align*}
	\int_t^{\infty}R(x,t)\,\d x=&\,\int_t^{\infty}R(t,x)\,\d x=\int_0^{\infty}R(t,x+t)\,\d x=\sum_{k=1}^{\infty}\bar{\gamma}^k\int_0^{\infty}(T\chi_t\upharpoonright_{L^2(\mathbb{R})})^k(t,x+t)\,\d x\\
	\stackrel{\eqref{c:2}}{=}&\,\sum_{k=1}^{\infty}\bar{\gamma}^k\int_{-\infty}^{\infty}\big((T\chi_t\upharpoonright_{L^2(\mathbb{R})})^kg\big)(u)\,\d u=-1+\int_{-\infty}^{\infty}\big((1-\bar{\gamma}T\chi_t\upharpoonright_{L^2(\mathbb{R})})^{-1}g\big)(x)\,\d x.
\end{align*}
This completes our proof.
\end{proof}
\begin{rem} The first and third integral identity in Lemma \ref{cur} are the $\bar{\gamma}$-generalizations of the equalities \cite[$(2.6),(2.8),(2.10)$]{PTZ} and \cite[$(2.3),(2.9)$]{PTZ}. The second and fourth identity are seemingly new.
\end{rem}

\begin{cor}\label{JBcor} For any $(t,\gamma)\in\mathbb{R}\times[0,1]$,
\begin{equation}\label{ex:12}
	v=\bar{\gamma}\,q,\ \ \ \ \ \ \ \ \ r=-2(1-\gamma)u+\bar{\gamma}\sqrt{2}\,p,\ \ \ \ \ \ \ \ \ w=2u,\ \ \ \ \ \ \ \ \ 2u+\sqrt{2}\,q+\sqrt{2}\left(\frac{1-\bar{\gamma}}{1-\gamma}\right)p=1.
\end{equation}
\end{cor}
\begin{proof} These follow from inserting the integral identities of Lemma \ref{cur} into the definitions of $u,v,p,q,r$ and $w$.
\end{proof}
Once we substitute \eqref{ex:12} into the $3\times 3$ determinant \eqref{ex:11} we are left with two unknown, $p$ and $q$, say, and the determinant simplifies to
\begin{align}
	\det&\begin{bmatrix}1-u & -v+\frac{\bar{\gamma}}{\sqrt{2}} & -v\smallskip\\
	-\frac{\gamma}{2\bar{\gamma}}(p+q-\frac{1}{\sqrt{2}}) & 1-\frac{\gamma}{2\bar{\gamma}}(r+w) & -\frac{\gamma}{2\bar{\gamma}}(r+w)\smallskip\\
	\frac{\gamma}{\bar{\gamma}}\,p & \frac{\gamma}{\bar{\gamma}}\,r & 1+\frac{\gamma}{\bar{\gamma}}\,r\end{bmatrix}\nonumber\\
	&\,\hspace{1cm}=\frac{1}{2(2-\gamma)}\left[(2-\gamma)\left(p+q+\frac{1}{\sqrt{2}}\right)^2-\gamma\left(p-q+\frac{1}{\sqrt{2}}\right)^2\right].\label{JB:1}
\end{align}
Next, we define the two functions
\begin{equation*}
	\tau_k=\tau_k(t,\bar{\gamma}):=\int_{-\infty}^{\infty}\big((1-\bar{\gamma}T\chi_t\upharpoonright_{L^2(\mathbb{R})})^{-1}g\big)(x)\,\d x+(-1)^{k-1}\sqrt{\bar{\gamma}}\int_t^{\infty}\big((1-\bar{\gamma}T\chi_t\upharpoonright_{L^2(\mathbb{R})})^{-1}g\big)(x)\,\d x
\end{equation*}
for $(t,\gamma)\in\mathbb{R}\times[0,1],k=1,2$ and note that by \eqref{extra:1}
\begin{equation}\label{JB:2}
	p=\frac{1-\gamma}{2\sqrt{2\bar{\gamma}}}\big(\tau_1-\tau_2\big),\ \ \ \ \ q=\frac{1}{2\sqrt{2\bar{\gamma}}}\Big[\sqrt{\bar{\gamma}}\big(\tau_1+\tau_2\big)-\big(\tau_1-\tau_2\big)\Big].
\end{equation}
Inserting \eqref{JB:2} in \eqref{JB:1} we find in turn
\begin{equation}\label{JB:3}
	\textnormal{RHS in}\ \eqref{JB:1}=\frac{(1-\gamma)(1+\tau_1\tau_2)+\tau_1+\tau_2-\sqrt{\bar{\gamma}}(\tau_1-\tau_2)}{2(2-\gamma)},
\end{equation}
and now set out to simplify $\tau_k$. First, by the second and fourth identity in Lemma \ref{cur},
\begin{equation}\label{JB:4}
	\tau_k=1+(-1)^{k-1}\sqrt{\bar{\gamma}}\int_t^{\infty}\left[1+(-1)^{k-1}\sqrt{\bar{\gamma}}\int_x^{\infty}g(y)\,\d y\right]\big((1-\bar{\gamma}T\chi_t\upharpoonright_{L^2(\mathbb{R})})^{-1}g\big)(x)\,\d x.
\end{equation}
Second, making essential use of the regularization scheme for Fredholm determinant and inner product manipulations in \cite[Section VIII]{TW}, we have the following two analogues of \cite[$(4.18),(4.21)$]{FD0} which we will use with $a=\pm\sqrt{\bar{\gamma}}$.
\begin{lem}\label{Forrlift}For any $(t,a)\in\mathbb{R}\times[-1,1]$,
\begin{equation*}
	1-a\int_t^{\infty}\left[1-a\int_x^{\infty}g(y)\,\d y\right]\big((1-a^2 T\chi_t\upharpoonright_{L^2(\mathbb{R})})^{-1}g\big)(x)\,\d x=\langle\chi_0,(1+aS_t\upharpoonright_{L^2(0,\infty)})^{-1}\delta_0\rangle_{L^2(0,\infty)},
\end{equation*}
where, for any test function $f$,
\begin{equation*}
	 \int_0^{\infty}f(x)\delta_0(x)\,\d x:=f(0),\hspace{1.5cm}\chi_0(x):=\begin{cases}1,&x\geq 0\\ 0,&x<0\end{cases},
\end{equation*}
and $S_t:L^2(0,\infty)\rightarrow L^2(0,\infty)$ denotes the trace class integral operator on $L^2(0,\infty)$ with kernel
\begin{equation*}
	S_t(x,y):=\frac{1}{\sqrt{\pi}}\,\e^{-(x+y+t)^2},\ \ \ x,y\geq 0.
\end{equation*}
\end{lem}
\begin{proof} Note that
\begin{equation*}
	1-a\int_{x+t}^{\infty}g(y)\,\d y=1-a\int_0^{\infty}S_t(x,y)\,\d y=\big((1-aS_t\upharpoonright_{L^2(0,\infty)})\chi_0\big)(x).
\end{equation*}
On the other hand, if $T_t:L^2(0,\infty)\rightarrow L^2(0,\infty)$ has kernel $T_t(x,y):=T(x+t,y+t)$, then
\begin{equation*}
	\big((1-a^2T\chi_t\upharpoonright_{L^2(\mathbb{R})})^{-1}g\big)(x+t)=\big((1-a^2 T_t\upharpoonright_{L^2(0,\infty)})^{-1}S_t\delta_0\big)(x)
\end{equation*}
and we have $T_t=S_tS_t$. Thus, without explicitly writing the underlying Hilbert spaces,
\begin{align*}
	\int_t^{\infty}&\,\bigg[1-a\int_x^{\infty}g(y)\,\d y\bigg]\big((1-a^2T\chi_t\upharpoonright_{L^2(\mathbb{R})})^{-1}g\big)(x)\,\d x\\
	&=\big\langle(1-aS_t)\chi_0,(1-a^2T_t)^{-1}S_t\delta_0,\big\rangle_{L^2(0,\infty)}=\big\langle\chi_0,(1+aS_t)^{-1}S_t\delta_0\big\rangle_{L^2(0,\infty)}
\end{align*}
by self-adjointness of $S_t$ and \cite[Lemma $6.1$]{BB}. Now using that $(1+aS_t)^{-1}=1-a(1+aS_t)^{-1}S_t$, we obtain at once the claimed identity from $\langle\chi_0,\delta_0\rangle_{L^2(0,\infty)}=1$.
\end{proof}
\begin{lem} For every $(t,a)\in\mathbb{R}\times[-1,1]$,
\begin{equation}\label{F:2}
	\det(1-aS_t\upharpoonright_{L^2(0,\infty)})=\det(1+aS_t\upharpoonright_{L^2(0,\infty)})\big\langle\chi_0,(1+aS_t\upharpoonright_{L^2(0,\infty)})^{-1}\delta_0\big\rangle_{L^2(0,\infty)}.
\end{equation}
\end{lem}
\begin{proof} As outlined in \cite[$(6.9)$]{BB}, identity \eqref{F:2} is equivalent to
\begin{equation*}
	\tr_{L^2(0,\infty)}\Big((1-a^2 S_t^2)^{-1}a\frac{\d S_t}{\d t}\Big)=-\frac{1}{2}\frac{\d}{\d t}\ln\big\langle\delta_0,(1+aS_t)^{-1}\chi_0\big\rangle_{L^2(0,\infty)},
\end{equation*}
and thus to
\begin{equation}\label{F:3}
	\big\langle\delta_0,(1-a^2S_t^2)^{-1}aS_t\delta_0\rangle_{L^2(0,\infty)}=\frac{\d}{\d t}\ln\big\langle\delta_0,(1+aS_t)^{-1}\chi_0\big\rangle_{L^2(0,\infty)},
\end{equation}
where we do not indicate the underlying Hilbert spaces for compact notation. In proving \eqref{F:3} we use the following straightforward $a$-generalization of \cite[$(6.11)$]{BB},
\begin{equation*}
	\frac{\d}{\d t}(1+aS_t)^{-1}=(1-a^2S_t^2)^{-1}aS_tD+(1-a^2S_t^2)^{-1}aS_t\Delta_0(1+aS_t)^{-1},
\end{equation*}
where $\Delta_0$ denotes multiplication by $\delta_0(x)$ and $D(=\frac{\d}{\d x})$ differentiation. We have thus
\begin{align*}
	\frac{\d}{\d t}\ln\big\langle\delta_0,&(1+aS_t)^{-1}\chi_0\big\rangle_{L^2(0,\infty)}=\frac{\langle\delta_0,\frac{\d}{\d t}(1+aS_t)^{-1}\chi_0\rangle_{L^2(0,\infty)}}{\langle\delta_0,(1+aS_t)^{-1}\chi_0\rangle_{L^2(0,\infty)}}\\
	&\stackrel{D\chi_0=0}{=}\frac{\langle\delta_0,(1-a^2S_t^2)^{-1}aS_t\Delta_0(1+aS_t)^{-1}\chi_0\rangle}{\langle\delta_0,(1+aS_t)^{-1}\chi_0\rangle}=\frac{\langle\delta_0,(1-a^2 S_t^2)^{-1}aS_t\delta_0\rangle\langle\delta_0,(1+aS_t)^{-1}\chi_0\rangle}{\langle\delta_0,(1+aS_t)^{-1}\chi_0\rangle}\\
	&=\big\langle\delta_0,(1-a^2S_t^2)^{-1}aS_t\delta_0\big\rangle_{L^2(0,\infty)},
\end{align*}
i.e. identity \eqref{F:3}. This completes our proof.
\end{proof}
Hence, given that $\det(1\mp\sqrt{\bar{\gamma}} S_t\upharpoonright_{L^2(0,\infty)})=\det(1\mp\sqrt{\bar{\gamma}}\chi_tS\chi_t\upharpoonright_{L^2(\mathbb{R})})>0$ with $S:L^2(\mathbb{R})\rightarrow L^2(\mathbb{R})$ as in the formulation of Theorem \ref{main:1}, we obtain the following result from Lemma \ref{Forrlift} and \eqref{F:2} with $a=\pm\sqrt{\bar{\gamma}}$.
\begin{prop} For any $(t,\gamma)\in\mathbb{R}\times[0,1]$,
\begin{equation}\label{JB:5}
	\tau_1(t,\bar{\gamma})=\frac{\det(1+\sqrt{\bar{\gamma}}\chi_tS\chi_t\upharpoonright_{L^2(\mathbb{R})})}{\det(1-\sqrt{\bar{\gamma}}\chi_tS\chi_t\upharpoonright_{L^2(\mathbb{R})})},\ \ \ \ \ \ \tau_2(t,\bar{\gamma})=\frac{1}{\tau_1(t,\bar{\gamma})}
\end{equation}
\end{prop}
We now return to \eqref{JB:3} and first use that $\tau_1\tau_2=1$, so after simplification
\begin{equation}\label{JB:6}
	\sqrt{\textnormal{RHS in}\ \eqref{JB:1}}=\sqrt{\frac{1-\sqrt{\bar{\gamma}}}{2(2-\gamma)}}\sqrt{\tau_1}+\sqrt{\frac{1+\sqrt{\bar{\gamma}}}{2(2-\gamma)}}\sqrt{\tau_2}.
\end{equation}
But since
\begin{equation}\label{JB:7}
	\det\big(1-\bar{\gamma}\chi_tT\chi_t\upharpoonright_{L^2(\mathbb{R})}\big)=\det\big(1-\sqrt{\bar{\gamma}}\chi_tS\chi_t\upharpoonright_{L^2(\mathbb{R})}\big)\det\big(1+\sqrt{\bar{\gamma}}\chi_tS\chi_t\upharpoonright_{L^2(\mathbb{R})}\big),
\end{equation}
we then find \eqref{e:16} from \eqref{e:8}, Propositions \ref{Jprop:1}, \ref{Jprop:2} and equations \eqref{JB:5},\eqref{JB:6},\eqref{JB:7}. This completes our proof of Theorem \ref{main:1}.

\section{Proof of Theorem \ref{main:22}}\label{cool}
Our proof begins with the following analogue of \cite[$(4.12)$]{FD0}.
\begin{lem} For any $(t,\gamma)\in\mathbb{R}\times[0,1]$ we have with $\mu(t;\gamma)$ as in \eqref{e:10},
\begin{equation}\label{JB:8}
	\e^{-\mu(t;\bar{\gamma})}=\tau_2(t,\bar{\gamma}),\ \ \ \ \ \ \ \ \ \ \e^{\mu(t;\bar{\gamma})}=\tau_1(t,\bar{\gamma}).
\end{equation}
\end{lem}
\begin{proof} By definition of $p$ and $u$ in Proposition \ref{Jprop:1},
\begin{align*}
	\tau_2(t,\bar{\gamma})=&\,1-\sqrt{\bar{\gamma}}\big(1-\sqrt{\bar{\gamma}}\big)\int_t^{\infty}\big((1-\bar{\gamma}T\chi_t\upharpoonright_{L^2(\mathbb{R})})^{-1}g\big)(x)\,\d x-\bar{\gamma}\int_t^{\infty}G(x)\big((1-\bar{\gamma}T\chi_t\upharpoonright_{L^2(\mathbb{R})})^{-1}g\big)(x)\,\d x\\
	=&\,1-\sqrt{2\bar{\gamma}}\left(\frac{1-\sqrt{\bar{\gamma}}}{1-\gamma}\right)p-2u.
\end{align*}
But with the formal replacement $\gamma\mapsto\bar{\gamma}$ in \cite[Section $4$]{BB}, we have
\begin{equation}\label{ex:13}
	u=\frac{1}{2}\big(1+\sqrt{\bar{\gamma}}\,\sinh\mu(t;\bar{\gamma})-\cosh\mu(t;\bar{\gamma})\big),\ \ \ \ (t,\gamma)\in\mathbb{R}\times[0,1],
\end{equation}
see \cite[$(4.19)$]{BB}, where (compare \eqref{e:10} and \cite[Proposition $3.10$]{BB}\footnote{\cite[RHP $3.8$]{BB} is a rescaled version of our RHP \ref{master} , hence the independent variable $\frac{x}{2}$ occurs in the integrand of $\mu$.})
\begin{equation*}
	\mu(t;\gamma)=\int_t^{\infty}Y_1^{12}\left(\frac{x}{2},\gamma\right)\,\d x=-\frac{\im}{2}\int_t^{\infty}y\left(\frac{x}{2};\gamma\right)\,\d x.
\end{equation*}
On the other hand, from \cite[$(4.11), (4.18)$]{BB} after dividing out $\sqrt{2\gamma}$ and replacing $\gamma\mapsto\bar{\gamma}$,
\begin{equation}\label{ex:14}
	q=\frac{1}{\sqrt{2}}\left(\cosh\mu(t;\bar{\gamma})-\frac{1}{\sqrt{\bar{\gamma}}}\sinh\mu(t;\bar{\gamma})\right),\ \ \ \ (t,\gamma)\in\mathbb{R}\times[0,1],
\end{equation}
so that with \eqref{ex:13} and \eqref{ex:14} back in the fourth equation in \eqref{ex:12},
\begin{equation}\label{ex:15}
	p=\frac{1-\gamma}{\sqrt{2\bar{\gamma}}}\sinh\mu(t;\bar{\gamma}),\ \ \ \ (t,\gamma)\in\mathbb{R}\times[0,1].
\end{equation}
Thus, all together, 
\begin{equation*}
	\tau_2(t,\bar{\gamma})=1-\sqrt{2\bar{\gamma}}\left(\frac{1-\sqrt{\bar{\gamma}}}{1-\gamma}\right)p-2u\stackrel[\eqref{ex:15}]{\eqref{ex:13}}{=}\e^{-\mu(t;\bar{\gamma})},
\end{equation*}
which is the analogue of \cite[$(4.12)$]{FD0}). The outstanding formula for $\tau_1(t,\bar{\gamma})$ follows from \eqref{JB:5}.
\end{proof}
In order to arrive at \eqref{e:9} we now apply \eqref{ex:11}, Proposition \ref{Jprop:2} and \eqref{JB:6},
\begin{equation*}
	P(t;\gamma)=\sqrt{\det\big(1-\bar{\gamma}\chi_tT\chi_t\upharpoonright_{L^2(\mathbb{R})}\big)}\left(\sqrt{\frac{1-\sqrt{\bar{\gamma}}}{2(2-\gamma)}}\sqrt{\tau_1}+\sqrt{\frac{1+\sqrt{\bar{\gamma}}}{2(2-\gamma)}}\sqrt{\tau_2}\right),\ \ \ (t,\gamma)\in\mathbb{R}\times[0,1].
\end{equation*}
But for any $a\in[0,1]$, see \cite[$(3.33)$]{BB},
\begin{equation}\label{blast}
	\det(1-a\chi_tT\chi_t\upharpoonright_{L^2(\mathbb{R})})=\exp\left[-\frac{1}{4}\int_t^{\infty}(x-t)\Big|y\Big(\frac{x}{2};a\Big)\Big|^2\,\d x\right],
\end{equation}
so with \eqref{JB:8},
\begin{equation*}
	P(t;\gamma)=\exp\left[-\frac{1}{8}\int_t^{\infty}(x-t)\Big|y\left(\frac{x}{2};\bar{\gamma}\right)\Big|^2\,\d x\right]\left(\sqrt{\frac{1-\sqrt{\bar{\gamma}}}{2(2-\gamma)}}\e^{\frac{1}{2}\mu(t;\bar{\gamma})}+\sqrt{\frac{1+\sqrt{\bar{\gamma}}}{2(2-\gamma)}}\e^{-\frac{1}{2}\mu(t;\bar{\gamma})}\right).
\end{equation*}
This is exactly \eqref{e:9}.

\section{Proof of Corollary \ref{eiggen}}\label{supereasy}
Note that with the abbreviation \eqref{e:13}, by the inclusion-exclusion principle, for any $\gamma\in[0,1]$,
\begin{equation}\label{exx:1}
	P(t;\gamma)=\lim_{n\rightarrow\infty}\mathbb{P}\big(\textnormal{no edge scaled eigenvalues}\ \mu_j^{\gamma}(\X)\ \textnormal{in}\ (t,\infty)\big)=\sum_{m=0}^{\infty}E(m,(t,\infty))(1-\gamma)^m,
\end{equation}
since each eigenvalue is removed independently with likelihood $1-\gamma$. But comparing the latter with \eqref{e:11} we find immediately \eqref{e:12}. Note also that since $\mu(x;\gamma)$ is in fact real analytic in $x\in\mathbb{R}$ for any fixed $\gamma\in[0,1]$ (see \cite[Corollary $3.6$]{BB} for continuity in $x$, real analyticity follows by a similar argument using the analytic Fredholm alternative in \cite{Z}) we obtain from Taylor's theorem, \eqref{exx:1} and \eqref{e:11},
\begin{equation*}
	E\big(m,(t,\infty)\big)=\frac{(-1)^m}{m!}\frac{\partial^m}{\partial\xi^m}E\big((t,\infty);\xi\big)\bigg|_{\xi=1},\ \ \ m\in\mathbb{Z}_{\geq 0},\ \ t\in\mathbb{R}.
\end{equation*}
This is the standard relation between the generating function and eigenvalue occupation probability known for any continuous one-dimensional statistical mechanical system, cf. \cite[$(8.1)$]{FP}.


\section{Proof of Theorem \ref{main:2} and Lemma \ref{toobad}}\label{nonsteep}
We prove \eqref{e:14}, \eqref{e:15} and Lemma \ref{toobad} in the upcoming three subsections.

\subsection{Right tail asymptotics - proof of \eqref{e:14}} From \eqref{blast}, i.e. \cite[$(3.33)$]{BB},
\begin{equation}\label{z:1}
	\exp\left[-\frac{1}{8}\int_t^{\infty}(x-t)\Big|y\Big(\frac{x}{2};\bar{\gamma}\Big)\Big|^2\,\d x\right]=\sqrt{\det(1-\bar{\gamma}\chi_tT\chi_t{\upharpoonright}_{L^2(\mathbb{R})})},
\end{equation}
and thus from \cite[Lemma $3.11$]{BB}, as $t\rightarrow+\infty$,
\begin{equation}\label{b:1}
	\exp\left[-\frac{1}{8}\int_t^{\infty}(x-t)\Big|y\Big(\frac{x}{2};\bar{\gamma}\Big)\Big|^2\,\d x\right]=1-\frac{\bar{\gamma}}{\sqrt{2\pi}}\frac{\textnormal{erfc}(\sqrt{2}\,t)}{16t}\big(1+\mathcal{O}(t^{-2})\big),
\end{equation}
uniformly in $\gamma\in[0,1]$. Moreover, from \cite[$(3.31)$, Proposition $3.10$]{BB}, as $x\rightarrow+\infty$,
\begin{equation*}
	y\left(\frac{x}{2};a\right)=2\im\sqrt{\frac{a}{\pi}}\,\e^{-x^2}\Big(1+\mathcal{O}\big(\e^{-x^2}\big)\Big),\ \ \ a\in[0,1],
\end{equation*}
so that in \eqref{e:10}, as $t\rightarrow+\infty$,
\begin{equation}\label{b:2}
	\mu(t;\bar{\gamma})=\frac{\sqrt{\bar{\gamma}}}{2}\textnormal{erfc}(t)+\mathcal{O}\big(t^{-1}\e^{-2t^2}\big),
\end{equation}
uniformly in $\gamma\in[0,1]$. Inserting \eqref{b:2} into \eqref{e:99} and combining the so obtained result with \eqref{b:1} yields immediately \eqref{e:14}.
\subsection{Left tail asymptotics - proof of \eqref{e:15}} 
From \cite[Proposition $5.7$]{BB}, as $t\rightarrow-\infty$ for any fixed $a\in[0,1)$,
\begin{equation}\label{b:3}
	-\frac{1}{8}\int_t^{\infty}(x-t)\left|y\left(\frac{x}{2};a\right)\right|^2\,\d x=\frac{t}{2\sqrt{2\pi}}\,\textnormal{Li}_{\frac{3}{2}}(a)+D_1(a)+o(1),
\end{equation}
with the polylogarithm $\textnormal{Li}_s(z)$, \cite[$25.12.10$]{NIST} and an unknown, $t$-independent, term $D_1(a)$. Moreover, from \cite[page $492$]{BB}, as $t\rightarrow-\infty$ and fixed $a\in[0,1)$,
\begin{equation}\label{b:4}
	\mu(t;a)=D_2(a)+o(1),
\end{equation}
with another unknown, $t$-independent, term $D_2(a)$. Thus combining \eqref{b:3} and \eqref{b:4} in the right hand side of \eqref{e:9} we find that for $\gamma\in[0,1)$, as $t\rightarrow-\infty$,
\begin{equation}\label{strat2}
	\ln P(t;\gamma)=\frac{t}{2\sqrt{2\pi}}\,\textnormal{Li}_{\frac{3}{2}}(\bar{\gamma})+\eta(\gamma)+o(1),\ \ \ \gamma\in[0,1),
\end{equation}
where $\eta(\gamma)$ is, as of now, unknown. Since $\eta(\gamma)$ comes from $D_1(\bar{\gamma})$ and $D_2(\bar{\gamma})$, we split its computation into two parts.
\subsubsection{Total integral computation} We first address the computation of $D_2(\bar{\gamma})$. Since
\begin{equation*}
	\mu(t;\gamma)=\int_t^{\infty}Y_1^{12}\left(\frac{x}{2},\gamma\right)\d x\stackrel{\eqref{b:4}}{=}\int_{-\infty}^{\infty}Y_1^{12}\left(\frac{x}{2},\gamma\right)\d x+o(1),\ \ \ \ \ \ \ \ t\rightarrow-\infty,\ \gamma\in[0,1),
\end{equation*}
we need to evaluate a total integral. In order to achieve this, we follow the approach developed in \cite{BBFI}, our net result being an analogue of \cite[$(28)$]{BBFI}. Recall that ${\bf Y}(z)={\bf Y}(z;x,\gamma)$ solves RHP \ref{master}.
\begin{lem} 
The well-defined and invertible limit
\begin{equation*}
	{\bf V}(x,\gamma):=\lim_{\substack{z\rightarrow 0\\ \Im z<0}}{\bf Y}(z;x,\gamma),\ \ \ \ \ \ \ (x,\gamma)\in\mathbb{R}\times[0,1],
\end{equation*}
satisfies 
\begin{equation}\label{b:5}
	{\bf V}(x,\gamma)=\begin{bmatrix}\cosh\nu & \im\sinh\nu\smallskip\\
	-\im\sinh\nu & \cosh\nu\end{bmatrix}{\bf V}(x_0,\gamma),\ \ \ \ \ \ \ \nu=\nu(x,x_0,\gamma):=2\int_{x_0}^xY_1^{12}(u,\gamma)\,\d u,
\end{equation}
for arbitrary $x,x_0\in\mathbb{R}$ and $\gamma\in[0,1]$.
\end{lem}
\begin{proof} Define ${\bf W}(z;x,\gamma):=\Y(z;x,\gamma)\e^{-\im zx\sigma_3}$ for $z\in\mathbb{C}\setminus\mathbb{R}$ with $(x,\gamma)\in\mathbb{R}\times[0,1]$ and where $\sigma_3:=\bigl[\begin{smallmatrix}1 & 0\\ 0 & -1\end{smallmatrix}\bigr]$. It is well-known, cf. \cite[page $479$]{BB}, that ${\bf W}={\bf W}(z;x,\gamma)$ solves the Zakharov-Shabat system
\begin{equation*}
	\frac{\partial{\bf W}}{\partial x}=\left\{-\im z\sigma_3+2\im\begin{bmatrix}0 & Y_1^{12}\smallskip\\
	-Y_1^{12} & 0\end{bmatrix}\right\}{\bf W},\ \ \ \ Y_1^{12}=Y_1^{12}(x,\gamma).
\end{equation*}
Taking the limit $z\rightarrow 0$ with $\Im z<0$, we find that
\begin{equation*}
	\frac{\partial{\bf V}}{\partial x}=2\im\begin{bmatrix}0 & Y_1^{12}\\
	-Y_1^{12} & 0\end{bmatrix}{\bf V}
\end{equation*}
with general solution \eqref{b:5}. This completes our proof.
\end{proof}
We now compute the limits of $\cosh\nu$ and $\sinh\nu$ as $x\rightarrow+\infty$ and $x_0\rightarrow-\infty$. By \eqref{b:5} these limits follow from the $x$-asymptotic behavior of ${\bf V}(x,\gamma)$ and thus from ${\bf Y}(z;x,\gamma)$. Some aspects of the asymptotic analysis of ${\bf Y}(z;x,\gamma)$ were carried out in \cite[Section $3.4$]{BB}, others can be found in Appendix \ref{steepbetter} below.
\begin{prop}\label{totint} Let $Y_1^{12}(x,\gamma)$ denote the $(12)$-entry of the matrix coefficient $\Y_1(x,\gamma)$ in RHP \ref{master}, condition $(3)$. Then for any fixed $\gamma\in[0,1)$,
\begin{equation*}
	\int_{-\infty}^{\infty}Y_1^{12}(u,\gamma)\,\d u=\frac{1}{4}\ln\left(\frac{1+\sqrt{\gamma}}{1-\sqrt{\gamma}}\right).
\end{equation*}
\end{prop}
\begin{proof} Since $Y_1^{12}(\cdot,\gamma)\in L^1(\mathbb{R})$ for any $\gamma\in[0,1]$, see \cite[Corollary $3.6$]{BB} and \cite[page $481,492$]{BB}, we will take $x\rightarrow+\infty$ and $x_0\rightarrow-\infty$ in \eqref{b:5} in order to compute the desired total integral. First, consider the limit
\begin{equation*}
	\lim_{x\rightarrow+\infty}{\bf V}(x,\gamma),\ \ \ \gamma\in[0,1].
\end{equation*}
By \cite[$(3.28)$]{BB}, for any $x>0$ and $\gamma\in[0,1]$,
\begin{equation*}
	{\bf V}(x,\gamma)=\T(0;2x,\gamma)\begin{bmatrix}1 & \im\sqrt{\gamma}\,\smallskip\\ 0 & 1\end{bmatrix},
\end{equation*}
in terms of the solution $\T(z;x,\gamma)$ of \cite[RHP $3.12$]{BB} evaluated at $z=0$. But \cite[$(3.29),(3.30)$]{BB} imply that ${\bf T}(0;2x,\gamma)\rightarrow\mathbb{I}$ as $x\rightarrow+\infty$, hence
\begin{equation}\label{b:7}
	\lim_{x\rightarrow+\infty}{\bf V}(x,\gamma)=\begin{bmatrix}1 & \im\sqrt{\gamma}\,\smallskip\\
	0 & 1\end{bmatrix},\ \ \gamma\in[0,1].
\end{equation}
Second, we compute 
\begin{equation*}
	\lim_{x\rightarrow-\infty}{\bf V}(x,\gamma),\ \ \ \gamma\in[0,1)
\end{equation*}
using the results of the nonlinear steepest descent analysis in Appendix \ref{steepbetter} below. From \eqref{a:5}, \eqref{a:7}, for any $x<0$ and $\gamma\in[0,1)$,
\begin{equation}\label{b:8}
	{\bf V}(x,\gamma)=\M(0;2x,\gamma)\exp\left[-\frac{\sigma_3}{2\pi\im}\,\textnormal{pv}\int_{-\infty}^{\infty}h(s;2x,\gamma)\frac{\d s}{s}\right]\begin{bmatrix}1 & 0\smallskip\\\im\sqrt{\gamma} & 1\end{bmatrix}(1-\gamma)^{-\frac{1}{2}\sigma_3},
\end{equation}
in terms of the solution $\M(z;x,\gamma)$ of RHP \ref{split} evaluated at $z=0$, where
\begin{equation*}
	h(s;t,\gamma)=-\ln\big(1-\gamma\e^{-\frac{1}{2}t^2s^2}\big),\ \ \ s,t\in\mathbb{R}
\end{equation*}
But since the integrand in \eqref{b:8} is an odd function of $s$, the principal value integral in \eqref{b:8} equals zero. Furthermore, \eqref{a:10} implies that ${\bf M}(0;2x,\gamma)\rightarrow\mathbb{I}$ as $x\rightarrow-\infty$. Thus,
\begin{equation}\label{b:9}
	\lim_{x\rightarrow-\infty}{\bf V}(x,\gamma)=\begin{bmatrix}1 & 0\smallskip\\
	\im\sqrt{\gamma} & 1\end{bmatrix}(1-\gamma)^{-\frac{1}{2}\sigma_3},\ \ \gamma\in[0,1).
\end{equation}
Combining \eqref{b:7} and \eqref{b:9},
\begin{equation*}
	{\bf V}(+\infty,\gamma)\big({\bf V}(-\infty,\gamma)\big)^{-1}=\frac{1}{\sqrt{1-\gamma}}\begin{bmatrix} 1 & \im\sqrt{\gamma}\,\smallskip\\
	-\im\sqrt{\gamma} & 1\end{bmatrix},\ \ \gamma\in[0,1),
\end{equation*}
which inserted into \eqref{b:5} yields
\begin{equation*}
	\cosh\big(\nu(+\infty,-\infty,\gamma)\big)=\frac{1}{\sqrt{1-\gamma}},\ \ \ \ \ \ \ \ \ \ \ \sinh\big(\nu(+\infty,-\infty,\gamma)\big)=\sqrt{\frac{\gamma}{1-\gamma}},
\end{equation*}
and thus after simplification (with $Y_1^{12}\in\mathbb{R}$) the claimed integral identity. This completes our proof.
\end{proof}
With Proposition \ref{totint} at hand we obtain in turn 
\begin{cor} For every fixed $\gamma\in[0,1)$, as $t\rightarrow-\infty$,
\begin{equation*}
	\mu(t;\gamma)=\frac{1}{2}\ln\left(\frac{1+\sqrt{\gamma}}{1-\sqrt{\gamma}}\right)+o(1),\end{equation*}
and thus
\begin{equation}\label{b:10}
	\sqrt{\frac{1-\sqrt{\bar{\gamma}}}{2(2-\gamma)}}\,\e^{\frac{1}{2}\mu(t;\bar{\gamma})}+\sqrt{\frac{1+\sqrt{\bar{\gamma}}}{2(2-\gamma)}}\,\e^{-\frac{1}{2}\mu(t;\bar{\gamma})}=\sqrt{2}\sqrt{\frac{1-\gamma}{2-\gamma}}+o(1).
\end{equation}
\end{cor}
The last corollary concludes our computation of $D_2(\gamma)$ in \eqref{b:4}.
\subsubsection{Resolvent integration}
We now compute $D_1(\bar{\gamma})$ in \eqref{b:3} using a different set of techniques. To be precise we first recall from \cite[Proposition $3.3$]{BB}, 
\begin{equation}\label{b:11}
	\det(1-\gamma\chi_tT\chi_t{\upharpoonright}_{L^2(\mathbb{R})})=\det(1-\gamma T\chi_t{\upharpoonright}_{L^2(\mathbb{R})})=\det(1-G{\upharpoonright}_{L^2(\Omega)}),
\end{equation}
with the oriented contour (see Figure \ref{figure4} in Appendix \ref{steepbetter})
\begin{equation*}
	\Omega=\mathbb{R}\sqcup(\mathbb{R}+\im\omega),
\end{equation*}
where $\omega>0$ will be determined in Lemma \ref{err1} below and $G:L^2(\Omega,|\d\lambda|)\rightarrow L^2(\Omega,|\d\lambda|)$ has kernel 
\begin{equation}\label{b:12}
	G(\lambda,\mu)=\frac{\f^{\intercal}(\lambda){\bf g}(\mu)}{\lambda-\mu},\ \ \ \ \ \ \ \ \f(\lambda)=\sqrt{\frac{\gamma}{2\pi}}\e^{-\frac{1}{8}\lambda^2}\begin{bmatrix}\chi_{\mathbb{R}}(\lambda)\smallskip\\ \e^{\im t\lambda}\chi_{\mathbb{R}+\im\omega}(\lambda)\end{bmatrix},\ \ \ {\bf g}(\mu)=\frac{1}{\sqrt{2\pi}}\e^{-\frac{1}{8}\mu^2}\begin{bmatrix}\chi_{\mathbb{R}+\im\omega}(\mu)\smallskip\\ \e^{-\im t\mu}\chi_{\mathbb{R}}(\mu)\end{bmatrix}.
\end{equation}
The algebraic form \eqref{b:12} of its kernel identifies the operator $G$ as an integrable operator, cf. \cite{IIKS}, whose resolvent $R=1+(1-G)^{-1}$ on $L^2(\Omega)$, if existent, has the form \eqref{a:3}. Choosing right sided limits for definiteness we have from \eqref{a:3} and \eqref{a:4}, for $\lambda,\mu\in\mathbb{R}$,
\begin{equation}\label{b:13}
	{\bf F}(\lambda)={\bf S}_-(\lambda){\bf f}(\lambda),\ \ \ \ {\bf G}(\mu)=\big({\bf S}_-^{\intercal}(\mu)\big)^{-1}{\bf g}(\mu),
\end{equation}
and for $\lambda,\mu\in\mathbb{R}+\im\omega$,
\begin{equation}\label{b:14}
	{\bf F}(\lambda)={\bf S}(\lambda)\begin{bmatrix}1 & 0\\
	\im\sqrt{\gamma}\,\e^{-\frac{1}{4}\lambda^2+\im t\lambda} & 1\end{bmatrix}{\bf f}(\lambda),\ \ \ 
	{\bf G}(\mu)=\big({\bf S}^{\intercal}(\mu)\big)^{-1}\begin{bmatrix}1 & -\im\sqrt{\gamma}\,\e^{-\frac{1}{4}\mu^2+\im t\mu}\\ 0 & 1\end{bmatrix}{\bf g}(\mu),
\end{equation}
where ${\bf S}(z)$ connects to RHP \ref{master} via ${\bf S}(z;t,\gamma)={\bf Y}(z;\frac{t}{2},\gamma),z\in\mathbb{C}\setminus\mathbb{R}$, compare \cite[$(3.20)$]{BB}. Next we record the following standard differential identity.
\begin{prop}\label{gderiv} For any $(t,\gamma)\in\mathbb{R}\times[0,1]$,
\begin{equation}\label{b:15}
	\frac{\partial}{\partial\gamma}\ln\det(1-\gamma\chi_tT\chi_t{\upharpoonright}_{L^2(\mathbb{R})})=-\frac{1}{2\gamma}\int_{\Omega}R(\lambda,\lambda)\,\d\lambda,
\end{equation}
with the kernel $R(\lambda,\mu)$ of the resolvent $R=1+(1-G)^{-1}$ on $L^2(\Omega)$.
\end{prop}
\begin{proof} We know from \cite[page $475$]{BB} that the resolvent operator exists for any $(t,\gamma)\in\mathbb{R}\times[0,1]$, thus by straightforward differentiation of \eqref{b:11} and \eqref{b:12},
\begin{eqnarray*}
	\frac{\partial}{\partial\gamma}\ln\det(1-\gamma\chi_tT\chi_t{\upharpoonright}_{L^2(\mathbb{R})})&=&\frac{\partial}{\partial\gamma}\ln\det(1-G{\upharpoonright}_{L^2(\Omega)})=-\tr_{L^2(\Omega)}\left((1-G)^{-1}\frac{\partial G}{\partial\gamma}\right)\\
	&\stackrel{\eqref{b:12}}{=}&-\frac{1}{2\gamma}\tr_{L^2(\Omega)}\big((1-G)^{-1}G\big)=-\frac{1}{2\gamma}\int_{\Omega}R(\lambda,\lambda)\,\d\lambda.
\end{eqnarray*}
This concludes our proof.
\end{proof}
In order to apply \eqref{b:15} we use the explicit formula \eqref{a:3} for the kernel of $R(\lambda,\mu)$ (see \cite{IIKS} for regularity properties of $R(\lambda,\mu)$),
\begin{equation*}
	R(\lambda,\lambda)=\big(\F^{\intercal}(\lambda)\big)'{\bf G}(\lambda),\ \ \ (')=\frac{\partial}{\partial\lambda},\ \ \ \ \ \lambda\in\Omega,
\end{equation*}
and combine it with the asymptotic results of Appendix \ref{steepbetter}, afterwards we integrate in \eqref{b:15}. In more detail, once the $t\rightarrow-\infty$ asymptotic expansion of the kernel $R(\lambda,\lambda)$ is known uniformly with respect to fixed $\gamma\in[0,1)$ and any $\lambda\in\Omega$ we simply integrate
\begin{equation}\label{b:16}
	\ln\det(1-\gamma\chi_tT\chi_t{\upharpoonright}_{L^2(\mathbb{R})})=\int_0^{\gamma}\frac{\partial}{\partial\gamma'}\ln\det(1-\gamma'\chi_tT\chi_t{\upharpoonright}_{L^2(\mathbb{R})})\,\d\gamma'=-\frac{1}{2}\int_0^{\gamma}\left[\int_{\Omega}R(\lambda,\lambda)\,\d\lambda\right]\frac{\d\gamma'}{\gamma'}
\end{equation}
and arrive at \eqref{e:15} and \eqref{e:155}. The detailed steps of this approach are as follows: From \eqref{a:5}, \eqref{a:7} for any $(t,\gamma)\in(-\infty,0)\times[0,1)$, provided we choose $\omega>0$ so that $\Sigma_{\bf M}\cap(\mathbb{R}+\im\frac{\omega}{|t|})=\emptyset$, see Lemma \ref{err1} below,
\begin{equation}\label{b:177}
	\begin{cases}
	{\bf F}(\lambda)={\bf M}\big(\frac{\lambda}{|t|};t,\gamma\big)\Big[{\bf A}_1(\lambda)\chi_{\mathbb{R}}(\lambda)+{\bf A}_2(\lambda)\chi_{\mathbb{R}+\im\omega}(\lambda)\Big]{\bf f}(\lambda)&\bigskip\\
	{\bf G}(\mu)=\Big({\bf M}^{\intercal}\big(\frac{\mu}{|t|};t,\gamma\big)\Big)^{-1}\Big[\big({\bf A}_1^{\intercal}(\mu)\big)^{-1}\chi_{\mathbb{R}}(\mu)+\big({\bf A}_2^{\intercal}(\mu)\big)^{-1}\chi_{\mathbb{R}+\im\omega}(\mu)\Big]{\bf g}(\mu)&
	\end{cases},
\end{equation}
for $(\lambda,\mu)\in\Omega\times\Omega$ with the unimodular factors
\begin{equation*}
	{\bf A}_1(\lambda):=\exp\left[-\frac{\sigma_3}{2\pi\im}\,\textnormal{pv}\int_{-\infty}^{\infty}\frac{h(s;1,\gamma)}{s-\lambda}\,\d s\right]\begin{bmatrix}1 & 0\smallskip\\
	\im\sqrt{\gamma}\,\e^{-\frac{1}{4}\lambda^2+\im t\lambda} & 1\end{bmatrix}\big(1-\gamma\e^{-\frac{1}{2}\lambda^2}\big)^{-\frac{1}{2}\sigma_3},\ \ \ \lambda\in\mathbb{R},
\end{equation*}
and
\begin{equation*}
	{\bf A}_2(\lambda):=\exp\left[-\frac{\sigma_3}{2\pi\im}\,\int_{-\infty}^{\infty}\frac{h(s;1,\gamma)}{s-\lambda}\,\d s\right]\begin{bmatrix}(1-\gamma\e^{-\frac{1}{2}\lambda^2})^{-1}&-\im\sqrt{\gamma}(1-\gamma\e^{-\frac{1}{2}\lambda^2})^{-1}\e^{-\frac{1}{4}\lambda^2-\im t\lambda}\smallskip\\
	\im\sqrt{\gamma}\,\e^{-\frac{1}{4}\lambda^2+\im t\lambda} & 1\end{bmatrix},
	 \ \lambda\in\mathbb{R}+\im\omega.
\end{equation*}
Inserting \eqref{b:177} into the right hand side of \eqref{b:15} we obtain after a short computation
\begin{align}
	\int_{\Omega}R(\lambda,\lambda)\,\d\lambda=\int_{\Omega}\big(\f^{\intercal}(\lambda)\big)'&\,{\bf g}(\lambda)\,\d\lambda+\int_{\mathbb{R}}\f^{\intercal}(\lambda)\big({\bf A}_1^{\intercal}(\lambda)\big)'\big({\bf A}_1^{\intercal}(\lambda)\big)^{-1}{\bf g}(\lambda)\,\d\lambda\label{compt}\\
	&+\int_{\mathbb{R}+\im\omega}\f^{\intercal}(\lambda)\big(\A_2^{\intercal}(\lambda)\big)'\big(\A_2^{\intercal}(\lambda)\big)^{-1}{\bf g}(\lambda)\,\d\lambda+\int_{\Omega}{\bf f}^{\intercal}(\lambda){\bf E}(\lambda){\bf g}(\lambda)\,\d\lambda,\nonumber
\end{align}
with
\begin{align*}
	{\bf E}(\lambda):=\Big[\A_1^{\intercal}(\lambda)\chi_{\mathbb{R}}(\lambda)+&\,\A_2^{\intercal}(\lambda)\chi_{\mathbb{R}+\im\omega}(\lambda)\Big]\frac{1}{|t|}\big({\bf M}^{\intercal}\big)'\left(\lambda/|t|\right)\big({\bf M}^{\intercal}(\lambda/|t|)\big)^{-1}\\
	&\times\Big[\big(\A_1^{\intercal}(\lambda)\big)^{-1}\chi_{\mathbb{R}}(\lambda)+\big(\A_2^{\intercal}(\lambda)\big)^{-1}\chi_{\mathbb{R}+\im\omega}(\lambda)\Big],\ \ \lambda\in\Omega.
\end{align*}
Given the particular shape of $\f(\lambda)$ and ${\bf g}(\lambda)$ in \eqref{b:12}, the first integral in \eqref{compt} evaluates to zero. For the fourth integral we record the following estimate.
\begin{lem}\label{err1} There exists $c>0$ such that for every fixed $\gamma\in[0,1)$ we can find $t_0=t_0(\gamma)>0$ so that
\begin{equation}\label{b:19}
	\left|-\frac{1}{2\gamma}\int_{\Omega}{\bf f}^{\intercal}(\lambda){\bf E}(\lambda){\bf g}(\lambda)\,\d\lambda\right|\leq c\left(\frac{s(\gamma)}{\sqrt{\gamma}\,|\ln\gamma|}\right)\e^{\frac{t}{2}\sqrt{-\ln\gamma}},\ \ \ s(\gamma):=\frac{\gamma^{\frac{1}{4}}}{1-\sqrt{\gamma}},
\end{equation}
for all $(-t)\geq t_0$.
\end{lem}
\begin{proof} If $\gamma=0$, then ${\bf A}_k(\lambda)\equiv\mathbb{I}$ for $k=1,2$ and likewise ${\bf M}(z)\equiv\mathbb{I}$, compare RHP \ref{split}. Thus the integral in question is identically zero and the claim trivially true. If $\gamma\in(0,1)$ is fixed, pick $\omega:=\frac{1}{4}\sqrt{-\ln\gamma}>0$ so that $0<\frac{\omega}{|t|}<\delta_{t\gamma}$ for $(-t)\geq t_0$, and first note from \eqref{a:10},
\begin{equation*}
	\left|\int_{\mathbb{R}}{\bf f}^{\intercal}(\lambda){\bf E}(\lambda){\bf g}(\lambda)\,\d\lambda\right|\leq c\sqrt{\gamma}\,s(\gamma)\frac{t^{-2}\e^{t\sqrt{-\ln\gamma}}}{\textnormal{dist}^2(\Sigma_{\bf M},\mathbb{R})}\ \ \ \forall\,(-t)\geq t_0,
\end{equation*}
thus, since $\textnormal{dist}(\Sigma_{\bf M},\mathbb{R})\geq\delta_{t\gamma}>0$, we indeed obtain the right hand side in \eqref{b:19} as upper bound. On the other hand, by explicit computation using again \eqref{a:10},
\begin{equation*}
	\left|\int_{\mathbb{R}+\im\omega}{\bf f}^{\intercal}(\lambda){\bf E}(\lambda){\bf g}(\lambda)\,\d\lambda\right|\leq c\sqrt{\gamma}\,s(\gamma)\frac{t^{-2}\e^{t\sqrt{-\ln\gamma}}}{\textnormal{dist}^2(\Sigma_{\bf M},\mathbb{R}+\im\frac{\omega}{|t|})}\e^{-2t\omega},
\end{equation*}
where $\textnormal{dist}(\Sigma_{\bf M},\mathbb{R}+\im\frac{\omega}{|t|})\geq\frac{3}{4}\delta_{t\gamma}>0$ by choice of $\omega$. We thus also obtain the right hand side of \eqref{b:19} as upper bound and have therefore completed our proof.
\end{proof}
The remaining two integrals in \eqref{compt} yield non-trivial contributions. We first state a lemma which is used in their evaluation.
\begin{lem} For any $a,b,\omega>0$,
\begin{equation}\label{n:1}
	\int_{\mathbb{R}+\im\omega}\int_{\mathbb{R}}\frac{\e^{-\frac{a}{2}\lambda^2-\frac{b}{2}s^2}}{(s-\lambda)^2}\,\d s\,\d\lambda=-\frac{2\pi\sqrt{ab}}{a+b}.
\end{equation}
\end{lem}
\begin{proof} Integration by parts in the variable $s$, as well as in the variable $\lambda$, yields
\begin{equation*}
	\textnormal{LHS in}\ \eqref{n:1}=-b\int_{\mathbb{R}+\im\omega}\int_{\mathbb{R}}\frac{s}{s-\lambda}\e^{-\frac{a}{2}\lambda^2-\frac{b}{2}s^2}\,\d s\,\d\lambda=a\int_{\mathbb{R}+\im\omega}\int_{\mathbb{R}}\frac{\lambda}{s-\lambda}\e^{-\frac{a}{2}\lambda^2-\frac{b}{2}s^2}\,\d s\,\d\lambda,
\end{equation*}
and therefore
\begin{equation*}
	-\left(\frac{1}{a}+\frac{1}{b}\right)\times\textnormal{LHS in}\ \eqref{n:1}=\int_{\mathbb{R}+\im\omega}\int_{\mathbb{R}}\e^{-\frac{a}{2}\lambda^2-\frac{b}{2}s^2}\,\d s\,\d\lambda=\frac{2\pi}{\sqrt{ab}},
\end{equation*}
since both remaining integrals are standard Gaussians. This proves \eqref{n:1}.
\end{proof}
We now compute the two outstanding integrals in \eqref{compt}
\begin{lem}\label{Jappr} For every $\gamma\in[0,1)$,
\begin{equation}\label{b:22}
	-\frac{1}{2\gamma}\int_{\mathbb{R}+\im\omega}{\bf f}^{\intercal}(\lambda)\big({\bf A}_2^{\intercal}(\lambda)\big)'\big({\bf A}_2^{\intercal}(\lambda)\big)^{-1}{\bf g}(\lambda)\,\d\lambda=\frac{\partial}{\partial\gamma}\left[\frac{t}{2\sqrt{2\pi}}\,\textnormal{Li}_{\frac{3}{2}}(\gamma)\right]+\frac{1}{4\pi\gamma}\big(\textnormal{Li}_{\frac{1}{2}}(\gamma)\big)^2.
\end{equation}
\end{lem}
\begin{proof} Inserting the formul\ae\,for ${\bf f}(\lambda),{\bf g}(\lambda)$ and ${\bf A}_2(\lambda)$ we find
\begin{align*}
	\textnormal{LHS in}&\ \eqref{b:22}=-\frac{1}{4\pi\sqrt{\gamma}}\int_{\mathbb{R}+\im\omega}\e^{-\frac{1}{4}\lambda^2+\im t\lambda}\begin{bmatrix}0\\ 1\end{bmatrix}^{\intercal}\big({\bf A}_2^{\intercal}(\lambda)\big)'\big({\bf A}_2^{\intercal}(\lambda)\big)^{-1}\begin{bmatrix}1\\ 0\end{bmatrix}\,\d\lambda\\
	=&\,\frac{\im}{4\pi}\int_{\mathbb{R}+\im\omega}\e^{-\frac{1}{4}\lambda^2+\im t\lambda}\frac{\d}{\d\lambda}\left[\frac{\e^{-\frac{1}{4}\lambda^2-\im t\lambda}}{1-\gamma\e^{-\frac{1}{2}\lambda^2}}\right]\,\d\lambda-\frac{1}{4\pi^2}\int_{\mathbb{R}+\im\omega}\frac{\e^{-\frac{1}{2}\lambda^2}}{1-\gamma\e^{-\frac{1}{2}\lambda^2}}\,\frac{\d}{\d\lambda}\left[\int_{-\infty}^{\infty}\frac{h(s;1,\gamma)}{s-\lambda}\,\d s\right]\d\lambda.
\end{align*}
Integrating by parts, collapsing $\mathbb{R}+\im\omega$ to $\mathbb{R}$ and using the oddness of a part of the integrand, we see that the first remaining integral yields
\begin{align*}
	\frac{\im}{4\pi}\int_{\mathbb{R}+\im\omega}\e^{-\frac{1}{4}\lambda^2+\im t\lambda}&\,\frac{\d}{\d\lambda}\left[\frac{\e^{-\frac{1}{4}\lambda^2-\im t\lambda}}{1-\gamma\e^{-\frac{1}{2}\lambda^2}}\right]\,\d\lambda=\frac{t}{4\pi}\int_{\mathbb{R}}\frac{\e^{-\frac{1}{2}\lambda^2}}{1-\gamma\e^{-\frac{1}{2}\lambda^2}}\,\d\lambda=\frac{t}{4\pi\gamma}\int_{-\infty}^{\infty}\left(\sum_{n=1}^{\infty}\gamma^n\e^{-\frac{n}{2}\lambda^2}\right)\d\lambda\\
	=&\,\frac{t}{2\sqrt{2\pi}}\sum_{n=1}^{\infty}\frac{\gamma^{n-1}}{\sqrt{n}}=\frac{\partial}{\partial\gamma}\left[\frac{t}{2\sqrt{2\pi}}\,\textnormal{Li}_{\frac{3}{2}}(\gamma)\right].
\end{align*}
In the second (double) integral, we use geometric progression and the power series expansion $\ln(1-z)=-\sum_{n=1}^{\infty}\frac{1}{n}z^n,|z|<1$ for $h(s;1,\gamma)$,
\begin{align*}
	-\frac{1}{4\pi^2}\int_{\mathbb{R}+\im\omega}\frac{\e^{-\frac{1}{2}\lambda^2}}{1-\gamma\e^{-\frac{1}{2}\lambda^2}}&\,\frac{\d}{\d\lambda}\left[\int_{-\infty}^{\infty}\frac{h(s;1,\gamma)}{s-\lambda}\,\d s\right]\d\lambda=-\frac{1}{4\pi^2\gamma}\sum_{n,m=1}^{\infty}\frac{\gamma^{n+m}}{m}\int_{\mathbb{R}+\im\omega}\int_{\mathbb{R}}\frac{\e^{-\frac{n}{2}\lambda^2-\frac{m}{2}s^2}}{(s-\lambda)^2}\,\d s\,\d\lambda\\
	\stackrel{\eqref{n:1}}{=}&\,\frac{1}{2\pi\gamma}\sum_{n,m=1}^{\infty}\frac{\gamma^{n+m}}{n+m}\sqrt{\frac{n}{m}}=\frac{1}{2\pi\gamma}\int_0^{\gamma}\textnormal{Li}_{-\frac{1}{2}}(x)\textnormal{Li}_{\frac{1}{2}}(x)\frac{\d x}{x}=\frac{1}{4\pi\gamma}\big(\textnormal{Li}_{\frac{1}{2}}(\gamma)\big)^2,
\end{align*}
since $\frac{\d}{\d x}\textnormal{Li}_{\frac{1}{2}}(x)=\frac{1}{x}\textnormal{Li}_{-\frac{1}{2}}(x)$ and $\textnormal{Li}_{-\frac{1}{2}}(0)=\textnormal{Li}_{\frac{1}{2}}(0)=0$. This completes our proof.
\end{proof}

\begin{lem}\label{swet} For every $\gamma\in[0,1)$,
\begin{equation}\label{b:20}
	-\frac{1}{2\gamma}\int_{\mathbb{R}}{\bf f}^{\intercal}(\lambda)\big({\bf A}_1^{\intercal}(\lambda)\big)'\big({\bf A}_1^{\intercal}(\lambda)\big)^{-1}{\bf g}(\lambda)\,\d\lambda=\frac{\partial}{\partial\gamma}\left[\frac{t}{2\sqrt{2\pi}}\,\textnormal{Li}_{\frac{3}{2}}(\gamma)\right]+\frac{1}{4\pi\gamma}\big(\textnormal{Li}_{\frac{1}{2}}(\gamma)\big)^2.
\end{equation}
\end{lem}
\begin{proof} Using the above formula for ${\bf A}_1(\lambda)$ and \eqref{b:12} we find at once
\begin{align*}
	-\frac{1}{2\gamma}\int_{\mathbb{R}}{\bf f}^{\intercal}(\lambda)\big({\bf A}_1^{\intercal}(\lambda)\big)'\big({\bf A}_1^{\intercal}(\lambda)\big)^{-1}{\bf g}(\lambda)\,\d\lambda=&\,\frac{t}{4\pi}\int_{-\infty}^{\infty}\frac{\e^{-\frac{1}{2}\lambda^2}\,\d\lambda}{1-\gamma\e^{-\frac{1}{2}\lambda^2}}\\
	&\,-\frac{1}{4\pi^2}\int_{-\infty}^{\infty}\frac{\e^{-\frac{1}{2}\lambda^2}}{1-\gamma\e^{-\frac{1}{2}\lambda^2}}\,\frac{\d}{\d\lambda}\left[\textnormal{pv}\int_{-\infty}^{\infty}\frac{h(s;1,\gamma)}{s-\lambda}\,\d s\right]\d\lambda.
\end{align*}
Here, the first remaining integral was already computed in the proof of Lemma \ref{Jappr},
\begin{equation*}
	\frac{t}{4\pi}\int_{-\infty}^{\infty}\frac{\e^{-\frac{1}{2}\lambda^2}\,\d\lambda}{1-\gamma\e^{-\frac{1}{2}\lambda^2}}=\frac{\partial}{\partial\gamma}\left[\frac{t}{2\sqrt{2\pi}}\,\textnormal{Li}_{\frac{3}{2}}(\gamma)\right],\ \ \ \gamma\in[0,1).
\end{equation*}
For the second one, we use the Plemelj-Sokhotski formula,
\begin{equation}\label{n:2}
	\lim_{\substack{z\rightarrow\lambda\in\mathbb{R}\\ \Im z>0}}\int_{-\infty}^{\infty}\frac{h(s;1,\gamma)}{s-z}\,\d s=\im\pi h(\lambda;1,\gamma)+\textnormal{pv}\int_{-\infty}^{\infty}\frac{h(s;1,\gamma)}{s-\lambda}\,\d s,
\end{equation}
and note that by oddness of the integrand,
\begin{equation}\label{n:3}
	\int_{-\infty}^{\infty}h(\lambda;1,\gamma)\,\frac{\d}{\d\lambda}\left[\frac{\e^{-\frac{1}{2}\lambda^2}}{1-\gamma\e^{-\frac{1}{2}\lambda^2}}\right]\,\d\lambda=0.
\end{equation}
Thus, integrating by parts and adding \eqref{n:3}, we find
\begin{align*}
	-\frac{1}{4\pi^2}\int_{-\infty}^{\infty}\frac{\e^{-\frac{1}{2}\lambda^2}}{1-\gamma\e^{-\frac{1}{2}\lambda^2}}&\,\frac{\d}{\d\lambda}\left[\textnormal{pv}\int_{-\infty}^{\infty}\frac{h(s;1,\gamma)}{s-\lambda}\,\d s\right]\d\lambda\\
	=&\,\frac{1}{4\pi^2}\int_{-\infty}^{\infty}\left(\frac{\d}{\d\lambda}\left[\frac{\e^{-\frac{1}{2}\lambda^2}}{1-\gamma\e^{-\frac{1}{2}\lambda^2}}\right]\right)\left(\textnormal{pv}\int_{-\infty}^{\infty}\frac{h(s;1,\gamma)}{s-\lambda}\,\d s+\im\pi h(\lambda;1,\gamma)\right)\d\lambda.
\end{align*}
Now change the contour $\mathbb{R}\ni\lambda$ to $\mathbb{R}+\im\omega$ by Cauchy's theorem while using the analytic continuation \eqref{n:2} for the second round bracket. The result equals
\begin{align*}
	-\frac{1}{4\pi^2}\int_{-\infty}^{\infty}\frac{\e^{-\frac{1}{2}\lambda^2}}{1-\gamma\e^{-\frac{1}{2}\lambda^2}}&\,\frac{\d}{\d\lambda}\left[\textnormal{pv}\int_{-\infty}^{\infty}\frac{h(s;1,\gamma)}{s-\lambda}\,\d s\right]\d\lambda=
	\frac{1}{4\pi^2}\int_{\mathbb{R}+\im\omega}\left(\frac{\d}{\d\lambda}\left[\frac{\e^{-\frac{1}{2}\lambda^2}}{1-\gamma\e^{-\frac{1}{2}\lambda^2}}\right]\right)\int_{-\infty}^{\infty}\frac{h(s;1,\gamma)}{s-\lambda}\,\d s\,\d\lambda\\
	=&\,-\frac{1}{4\pi^2}\int_{\mathbb{R}+\im\omega}\frac{\e^{-\frac{1}{2}\lambda^2}}{1-\gamma\e^{-\frac{1}{2}\lambda^2}}\,\frac{\d}{\d\lambda}\left[\int_{-\infty}^{\infty}\frac{h(s;1,\gamma)}{s-\lambda}\,\d s\right]\d\lambda
\end{align*}
after another integration by parts in the last equality. The obtained result is identical to the second (double) integral in the proof of Lemma \ref{Jappr} and we therefore find \eqref{b:20} all together.
\end{proof}

We now combine \eqref{b:19}, \eqref{b:20}, \eqref{b:22} and \eqref{b:16} to obtain the following result.
\begin{prop} There exists $c>0$ such that for every fixed $\gamma\in[0,1)$ we can find $t_0=t_0(\gamma)>0$ so that
\begin{equation}\label{b:24}
	\ln\det(1-\gamma\chi_tT\chi_t\upharpoonright_{L^2(\mathbb{R})})=\frac{t}{\sqrt{2\pi}}\,\textnormal{Li}_{\frac{3}{2}}(\gamma)+\frac{1}{2\pi}\int_0^{\gamma}\big(\textnormal{Li}_{\frac{1}{2}}(x)\big)^2\frac{\d x}{x} + r(t,\gamma)
\end{equation}
for $(-t)\geq t_0$ where the error term $r(t,\gamma)$ is differentiable with respect to $\gamma$ and satisfies
\begin{equation*}
	\big|r(t,\gamma)\big|\leq c\left(\frac{\gamma^{\frac{3}{4}}}{(1-\sqrt{\gamma})|\ln\gamma|}\right)\e^{\frac{t}{2}\sqrt{-\ln\gamma}}\ \ \ \ \ \ \forall\,(-t)\geq t_0.
\end{equation*}
\end{prop}
\begin{proof}
We have, as $t\rightarrow-\infty$,
\begin{equation*}
	-\frac{1}{2\gamma}\int_{\Omega}R(\lambda,\lambda)\,\d\lambda=\frac{\partial}{\partial\gamma}\left[\frac{t}{\sqrt{2\pi}}\,\textnormal{Li}_{\frac{3}{2}}(\gamma)\right]+\frac{1}{2\pi\gamma}\big(\textnormal{Li}_{\frac{1}{2}}(\gamma)\big)^2+\mathcal{O}\left(\frac{s(\gamma)}{\sqrt{\gamma}\,|\ln\gamma|}\,\e^{\frac{t}{2}\sqrt{-\ln\gamma}}\right),
\end{equation*}
uniformly in $\gamma\in[0,1)$. Integrating this expansion in \eqref{b:16} from $0$ to $\gamma<1$ yields immediately the two leading terms in \eqref{b:24} and for the error term we estimate as follows
\begin{equation*}
	\left|\int_0^{\gamma}\frac{x^{-\frac{1}{4}}}{(1-\sqrt{x})\ln x}\e^{\frac{t}{2}\sqrt{-\ln x}}\,\d x\right|\leq-\frac{\e^{\frac{t}{2}\sqrt{-\ln\gamma}}}{1-\sqrt{\gamma}}\int_0^{\gamma}\frac{x^{-\frac{1}{4}}}{\ln x}\,\d x\leq c\left(\frac{\gamma^{\frac{3}{4}}}{(1-\sqrt{\gamma})|\ln\gamma|}\right)\e^{\frac{t}{2}\sqrt{-\ln\gamma}}.
\end{equation*}
This completes our proof.
\end{proof}. 
Combining our results we finally arrive at \eqref{e:15}.
\begin{cor} As $t\rightarrow-\infty$, for any $\gamma\in[0,1)$,
\begin{equation*}
	P(t;\gamma)=\exp\left[\frac{t}{2\sqrt{2\pi}}\,\textnormal{Li}_{\frac{3}{2}}(\bar{\gamma})+\frac{1}{2}\ln\left(\frac{2}{2-\gamma}\right)+\frac{1}{4\pi}\int_0^{\bar{\gamma}}\left(\big(\textnormal{Li}_{\frac{1}{2}}(x)\big)^2-\frac{\pi x}{1-x}\right)\frac{\d x}{x}\right]\big(1+o(1)\big)
\end{equation*}
\end{cor}
\begin{proof} From \eqref{e:9}, \eqref{z:1}, \eqref{b:3}, \eqref{b:10} and \eqref{b:24} (substituting $\gamma\mapsto\bar{\gamma}$ in the last equation),
\begin{align}
	P(t;\gamma)=&\,\sqrt{\det(1-\bar{\gamma}\chi_tT\chi_t\upharpoonright_{L^2(\mathbb{R})})}\,\left(\sqrt{\frac{1-\sqrt{\bar{\gamma}}}{2(2-\gamma)}}\,\e^{\frac{1}{2}\mu(t;\bar{\gamma})}+\sqrt{\frac{1+\sqrt{\bar{\gamma}}}{2(2-\gamma)}}\,\e^{-\frac{1}{2}\mu(t;\bar{\gamma})}\right)\nonumber\\
	=&\,\exp\left[\frac{t}{2\sqrt{2\pi}}\textnormal{Li}_{\frac{3}{2}}(\bar{\gamma})+\frac{1}{4\pi}\int_0^{\bar{\gamma}}\big(\textnormal{Li}_{\frac{1}{2}}(x)\big)^2\frac{\d x}{x}\right]\sqrt{2}\sqrt{\frac{1-\gamma}{2-\gamma}}\big(1+o(1)\big)\label{b:25}
\end{align}
as $t\rightarrow-\infty$. The claim follows now after writing
\begin{equation*}
	\ln(1-\gamma)=\frac{1}{2}\ln\big(1-\bar{\gamma}\big)=-\frac{1}{2}\int_0^{\bar{\gamma}}\frac{\d x}{1-x}.
\end{equation*}
\end{proof}
\subsection{Proof of Lemma \ref{toobad}} Since $\textnormal{Li}_{\frac{1}{2}}(x)=x+\mathcal{O}(x^2)$ as $x\rightarrow 0$ and, cf. \cite[$25.12.12$]{NIST},
\begin{equation*}
	\textnormal{Li}_{\frac{1}{2}}(x)=\sqrt{\frac{\pi}{1-x}}\,\big(1+\mathcal{O}(x-1)\big)+\mathcal{O}(1),\ \ \ x\uparrow 1,
\end{equation*}
we see that 
\begin{equation*}
	\int_0^{\bar{\gamma}}\left(\big(\textnormal{Li}_{\frac{1}{2}}(x)\big)^2-\frac{\pi x}{1-x}\right)\frac{\d x}{x}
\end{equation*}
converges as $\bar{\gamma}\uparrow 1$, so $c_0(\gamma)$ is indeed continuous in $\gamma\in[0,1]$. On the other hand, from the power series representation of the polylogarithm,
\begin{equation*}
	\big(\textnormal{Li}_{\frac{1}{2}}(x)\big)^2-\frac{\pi x}{1-x}=\sum_{n,m=1}^{\infty}\frac{x^{n+m}}{\sqrt{nm}}-\pi\sum_{n=1}^{\infty}x^n=\sum_{n=1}^{\infty}a_nx^n,\ \ \ |x|<1,
\end{equation*}
with
\begin{equation}\label{b:27}
	a_n:=-\pi+\sum_{m=1}^{n-1}\frac{1}{\sqrt{m(n-m)}},\ \ n\in\mathbb{Z}_{\geq 1}\ \ \ \ \ \textnormal{and}\ \ \ \ \sum_{m=1}^{n-1}\frac{1}{\sqrt{m(n-m)}}=\pi-\frac{c}{\sqrt{n}}+\mathcal{O}\big(n^{-1}\big),\ n\rightarrow\infty,
\end{equation}
for some $c>0$. Thus, for any $0\leq t<1$,
\begin{equation*}
	\int_0^t\left(\big(\textnormal{Li}_{\frac{1}{2}}(x)\big)^2-\frac{\pi x}{1-x}\right)\frac{\d x}{x}=\sum_{n=1}^{\infty}\frac{1}{n}\left(-\pi+\sum_{m=1}^{n-1}\frac{1}{\sqrt{m(n-m)}}\right)t^n,
\end{equation*}
which verifies \eqref{Jext} for $0\leq\gamma<1$ through \eqref{e:155}. But using again \cite[$25.12.12$]{NIST} we also have that
\begin{equation*}
	\int_t^1\left(\big(\textnormal{Li}_{\frac{1}{2}}(x)\big)^2-\frac{\pi x}{1-x}\right)\frac{\d x}{x}=o(1),\ \ \ t\uparrow 1,
\end{equation*}
so by Abel's convergence theorem,
\begin{align*}
	\int_0^1\left(\big(\textnormal{Li}_{\frac{1}{2}}(x)\big)^2-\frac{\pi x}{1-x}\right)\frac{\d x}{x}=&\,\lim_{t\uparrow 1}\left[\int_0^t\left(\big(\textnormal{Li}_{\frac{1}{2}}(x)\big)^2-\frac{\pi x}{1-x}\right)\frac{\d x}{x}+\int_t^1\left(\big(\textnormal{Li}_{\frac{1}{2}}(x)\big)^2-\frac{\pi x}{1-x}\right)\frac{\d x}{x}\right]\\
	=&\,\lim_{t\uparrow 1}\left[\sum_{n=1}^{\infty}\frac{1}{n}a_nt^n+o(1)\right]=\sum_{n=1}^{\infty}\frac{1}{n}a_n,
\end{align*}
since $\frac{1}{n}a_n$ is summable, see \eqref{b:27}. The proof of Lemma \ref{toobad} is now complete.


\begin{appendix}

\section{Integral identities}\label{sec:int}
Given two continuous functions $\phi,\psi:\mathbb{R}\rightarrow\mathbb{R}$ which decay exponentially fast at $+\infty$, we define
\begin{equation}\label{a:0}
	K(x,y):=\int_0^{\infty}\phi(x+u)\psi(y+u)\,\d u,\ \ \ \ x,y\in\mathbb{R}
\end{equation}
and the associated integral operator $K$ on $L^2(\mathbb{R})$ with kernel $K(x,y)$. We denote by $f_y(x):=f(x+y)$ the horizontal shift of a function $f$ by $-y$.
\begin{lem}\label{appprop} Let $I\subset\mathbb{R}$ be an interval and
\begin{equation*}
	\Phi(x):=\int_I\phi_v(x)\,\d v.
\end{equation*}
Then for any $y,t\in\mathbb{R}$ and $k\in\mathbb{Z}_{\geq 1}$,
\begin{equation}\label{a:1}
	\int_I(K\chi_t\upharpoonright_{L^2(\mathbb{R})})^k(x+t,y)\,\d x=\chi_{[t,\infty)}(y)\int_t^{\infty}\Phi(u)\big((K^{\ast}\chi_t\upharpoonright_{L^2(\mathbb{R})})^{k-1}\psi_{u-t}\big)(y)\,\d u,
\end{equation}
where $K^{\ast}$ is the real adjoint of $K$.
\end{lem}
\begin{proof} We proceed by induction on $k\in\mathbb{Z}_{\geq 1}$. For $k=1$, the left hand side in \eqref{a:1} equals
\begin{equation*}
	\int_IK(x+t,y)\chi_{[t,\infty)}(y)\,\d x\stackrel{\eqref{a:0}}{=}\chi_{[t,\infty)}(y)\int_I\int_0^{\infty}\phi_x(t+u)\psi(y+u)\,\d u\,\d x
\end{equation*}
and hence by Fubini's theorem and the definition of $\Phi(x)$,
\begin{equation*}
	\int_IK(x+t,y)\chi_{[t,\infty)}(y)\,\d x=\chi_{[t,\infty)}(y)\int_0^{\infty}\Phi(t+u)\psi(y+u)\,\d u=\chi_{[t,\infty)}(y)\int_t^{\infty}\Phi(u)\psi_{u-t}(y)\,\d u,
\end{equation*}
which is the right hand side in \eqref{a:1}. Now assume \eqref{a:1} holds true for general $k$, then
\begin{align*}
	\int_I\,(K\chi_t\upharpoonright_{L^2(\mathbb{R})})^{k+1}&\,(x+t,y)\,\d x=\int_I\int_t^{\infty}(K\chi_t\upharpoonright_{L^2(\mathbb{R})})^k(x+t,v)(K\chi_t\upharpoonright_{L^2(\mathbb{R})})(v,y)\,\d v\,\d x\\
	&\,\,=\chi_{[t,\infty)}(y)\int_t^{\infty}\left[\int_I(K\chi_t\upharpoonright_{L^2(\mathbb{R})})^k(x+t,v)\,\d x\right]K(v,y)\,\d v\\
	&\,\,=\chi_{[t,\infty)}(y)\int_t^{\infty}\int_t^{\infty}\Phi(u)\big((K^{\ast}\chi_t\upharpoonright_{L^2(\mathbb{R})})^{k-1}\psi_{u-t}\big)(v)K(v,y)\,\d u\,\d v,
\end{align*}
where we used Fubini's theorem in the second equality and the induction hypothesis in the third. Continuing further with Fubini's theorem and the fact that $K^{\ast}(x,y)=K(y,x)$, we have then
\begin{equation*}
	\int_I(K\chi_t\upharpoonright_{L^2(\mathbb{R})})^{k+1}(x+t,y)\,\d x=\chi_{[t,\infty)}(y)\int_t^{\infty}\Phi(u)\big((K^{\ast}\chi_t\upharpoonright_{L^2(\mathbb{R})})^k\psi_{u-t}\big)(y)\,\d y,
\end{equation*}
which is the right hand side of \eqref{a:1} with $k-1\mapsto k$, as desired. This concludes our proof.
\end{proof}
Lemma \ref{appprop} implies the following integral identity.
\begin{cor} For any $t\in\mathbb{R}$ and $k\in\mathbb{Z}_{\geq 1}$,
\begin{equation}\label{a:2}
	\int_I(K\chi_t\upharpoonright_{L^2(\mathbb{R})})^k(x+t,t)\,\d x=\int_t^{\infty}\Phi(u)\big((K^{\ast}\chi_t\upharpoonright_{L^2(\mathbb{R})})^{k-1}\psi\big)(u)\,\d u.
\end{equation}
\end{cor}
\begin{proof} By \eqref{a:1} (with $y=t$) and $\chi_{[t,\infty)}(t)=1$,
\begin{equation*}
	\int_I(K\chi_t\upharpoonright_{L^2(\mathbb{R})})^k(x+t,t)\,\d x=\int_t^{\infty}\Phi(u)\big((K^{\ast}\chi_t\upharpoonright_{L^2(\mathbb{R})})^{k-1}\psi_{u-t}\big)(t)\,\d u.
\end{equation*}
However from \cite[Proposition B.1]{BB} we have
\begin{equation*}
	\big((K^{\ast}\chi_t\upharpoonright_{L^2(\mathbb{R})})^k\psi_{u-t}\big)(t)=\big((K^{\ast}\chi_t\upharpoonright_{L^2(\mathbb{R})})^k\psi_{0}\big)_{u-t}(t),
\end{equation*}
since in the kernel of $K^{\ast}$ the functions $\phi$ and $\psi$ are simply interchanged, compare \eqref{a:0}. But $\psi_0\equiv\psi$ and for any function $f$ we have $f_{u-t}(t)=f(u)$ by definition of the shift. Thus all together,
\begin{equation*}
	\int_I(K\chi_t\upharpoonright_{L^2(\mathbb{R})})^k(x+t,t)\,\d x=\int_t^{\infty}\Phi(x)\big((K^{\ast}\chi_t\upharpoonright_{L^2(\mathbb{R})})^{k-1}\psi\big)(u)\,\d u,
\end{equation*}
as claimed.
\end{proof}
\begin{lem} Assume $\phi\in L^1(\mathbb{R})$ in \eqref{a:0}. Then, for any $y,t\in\mathbb{R}$ and $k\in\mathbb{Z}_{\geq 1}$,
\begin{equation}\label{c:1}
	\int_{-\infty}^{\infty}\big((K\chi_t\upharpoonright_{L^2(\mathbb{R})})^k\phi_y\big)(x)\,\d x=\left[\int_{-\infty}^{\infty}\phi(x)\,\d x\right]\int_0^{\infty}(K\chi_t\upharpoonright_{L^2(\mathbb{R})})^k(y+t,u+t)\,\d u.
\end{equation}
\end{lem}
\begin{proof} We use once more induction on $k\in\mathbb{Z}_{\geq 1}$. For $k=1$, the left hand side in \eqref{c:1} equals
\begin{equation*}
	\int_{-\infty}^{\infty}\int_t^{\infty}K(x,s)\phi(s+y)\,\d s\,\d x\stackrel{\eqref{a:0}}{=}\int_{-\infty}^{\infty}\int_t^{\infty}\int_0^{\infty}\phi(x+u)\psi(s+u)\phi(s+y)\,\d u\,\d s\,\d x,
\end{equation*}
so by Fubini's theorem
\begin{align*}
	&\int_{-\infty}^{\infty}\big((K\chi_t\upharpoonright_{L^2(\mathbb{R})})\phi_y\big)(x)\,\d x=\left[\int_{-\infty}^{\infty}\phi(x)\,\d x\right]\int_t^{\infty}\int_0^{\infty}\psi(s+u)\phi(s+y)\,\d u\,\d s\\
	=&\,\left[\int_{-\infty}^{\infty}\phi(x)\,\d x\right]\int_0^{\infty}\int_t^{\infty}\psi(s+u)\phi(s+y+t)\,\d u\,\d s
	\stackrel{\eqref{a:0}}{=}\left[\int_{-\infty}^{\infty}\phi(x)\,\d x\right]\int_0^{\infty}(K\chi_t\upharpoonright_{L^2(\mathbb{R})})(y+t,u+t)\,\d u,
\end{align*}
which is the right hand side in \eqref{c:1} for $k=1$. Assuming now that \eqref{b:1} holds for general $k$, we compute
\begin{align*}
	\int_{-\infty}^{\infty}\big((K\chi_t&\,\upharpoonright_{L^2(\mathbb{R})})^{k+1}\phi_y\big)(x)\,\d x=\int_{-\infty}^{\infty}\int_{-\infty}^{\infty}(K\chi_t\upharpoonright_{L^2(\mathbb{R})})^k(x,u)(K\chi_t\upharpoonright_{L^2(\mathbb{R})}\phi_y)(u)\,\d u\,\d x.
\end{align*}
Inserting \eqref{a:0} for $K(u,v)$ and using Fubini's theorem, we find that
\begin{align*}
	\int_{-\infty}^{\infty}\big((K\chi_t&\,\upharpoonright_{L^2(\mathbb{R})})^{k+1}\phi_y\big)(x)\,\d x=\int_{-\infty}^{\infty}\int_t^{\infty}\int_0^{\infty}\big((K\chi_t\upharpoonright_{L^2(\mathbb{R})})^k\phi_s\big)(x)\psi(v+s)\phi(v+y)\,\d s\,\d v\,\d x
\end{align*}
by Fubini's theorem. Using Fubini's theorem again and the induction hypothesis the above simplifies to
\begin{eqnarray*}
		&&\!\!\!\!\!\!\!\!\!\!\!\!\!\!\left[\int_{-\infty}^{\infty}\phi(x)\,\d x\right]\int_t^{\infty}\int_0^{\infty}\int_0^{\infty}(K\chi_t\upharpoonright_{L^2(\mathbb{R})})^k(s+t,u+t)\psi(v+s)\phi(v+y)\,\d u\,\d s\,\d v\\
	&=&\left[\int_{-\infty}^{\infty}\phi(x)\,\d x\right]\int_0^{\infty}\int_t^{\infty}\int_0^{\infty}(K\chi_t\upharpoonright_{L^2(\mathbb{R})})^k(s,u+t)\psi(v+s)\phi(v+y+t)\,\d u\,\d s\,\d v,
\end{eqnarray*}
and from \eqref{a:0} we conclude that
\begin{align*}
	\int_{-\infty}^{\infty}\big((K\chi_t&\,\upharpoonright_{L^2(\mathbb{R})})^{k+1}\phi_y\big)(x)\,\d x\\
	=&\left[\int_{-\infty}^{\infty}\phi(x)\,\d x\right]\int_t^{\infty}\int_0^{\infty}(K\chi_t\upharpoonright_{L^2(\mathbb{R})})^k(s,u+t)K(y+t,s)\,\d u\,\d s\\
	=&\left[\int_{-\infty}^{\infty}\phi(x)\,\d x\right]\int_0^{\infty}\big(K\chi_t\upharpoonright_{L^2(\mathbb{R})}\big)^{k+1}(y+t,u+t)\,\d u,
\end{align*}
which is the right hand side of \eqref{c:1} with $k\mapsto k+1$, as needed. This completes our proof.
\end{proof}
The special case $y=0$ in \eqref{c:1} will be useful for us, we summarize it below.
\begin{cor} For any $t\in\mathbb{R}$ and $k\in\mathbb{Z}_{\geq 1}$,
\begin{equation}\label{c:2}
	\int_{-\infty}^{\infty}\big((K\chi_t\upharpoonright_{L^2(\mathbb{R})})^k\phi\big)(x)\,\d x=\left[\int_{-\infty}^{\infty}\phi(x)\,\d x\right]\int_0^{\infty}(K\chi_t\upharpoonright_{L^2(\mathbb{R})})^k(t,u+t)\,\d u,
\end{equation}
provided $\phi\in L^1(\mathbb{R})$ in \eqref{a:0}.
\end{cor}

\section{Streamlined nonlinear steepest descent analysis}\label{steepbetter}
The purpose of this section is to simplify and streamline \cite[Section $5$]{BB}. The analysis presented in loc. cit. is sufficient for the $t$-derivative method of \cite[Proposition $3.7$]{BB} but not ideal for our current needs, i.e. for Proposition \ref{gderiv}. Here are the necessary steps: From \cite[Proposition $3.3$]{BB},
\begin{equation*}
	\det(1-\gamma\chi_tT\chi_t\upharpoonright_{L^2(\mathbb{R})})=\det(1-G\upharpoonright_{L^2(\Omega)}),\ \ \ \ (t,\gamma)\in\mathbb{R}\times[0,1],
\end{equation*}
where the integrable operator $G$, see \eqref{b:12}, is naturally associated with the following RHP.

\begin{problem}[{\cite[RHP $3.4$]{BB}}]\label{mast} For $(t,\gamma)\in\mathbb{R}\times[0,1]$, determine ${\bf N}(z)={\bf N}(z;t,\gamma)\in\mathbb{C}^{2\times 2}$ such that
\begin{enumerate}
	\item[(1)] ${\bf N}(z)$ is analytic for $z\in\mathbb{C}\setminus\Omega$ where $\Omega=\mathbb{R}\cup(\mathbb{R}+\im\omega)$, oriented from left to right as shown in Figure \ref{figure4} below. Moreover ${\bf N}(z)$ extends continuously to $\{z\in\mathbb{C}:\,\Im z\geq \omega\}\cup\{z\in\mathbb{C}:\,0\leq\Im z\leq \omega\}\cup\{z\in\mathbb{C}:\,\Im z\leq 0\}$.
	\item[(2)] The limiting values ${\bf N}_{\pm}(z),z\in\Omega$ from either side of $\mathbb{C}\setminus\Omega$ satisfy
	\begin{equation*}
		{\bf N}_+(z)={\bf N}_-(z)\begin{bmatrix}1 & -\im\sqrt{\gamma}\,\e^{-\frac{1}{4}z^2-\im tz}\smallskip\\
		0 & 1\end{bmatrix},\ z\in\mathbb{R};\ \ \ \ {\bf N}_+(z)={\bf N}_-(z)\begin{bmatrix}1 & 0\\
		-\im\sqrt{\gamma}\,\e^{-\frac{1}{4}z^2+\im tz} & 1\end{bmatrix},\ z\in\mathbb{R}+\im\omega.
	\end{equation*}
	\item[(3)] As $z\rightarrow\infty$, we enforce the normalization
	\begin{equation*}
		{\bf N}(z)=\mathbb{I}+\mathcal{O}\big(z^{-1}\big).
	\end{equation*}
\end{enumerate}
\end{problem}
\begin{figure}[tbh]
\begin{tikzpicture}[xscale=0.7,yscale=0.7]
\draw [->] (-6,0) -- (6,0) node[below]{{\small $\Re z$}};
\draw [->] (0,-1) -- (0,4) node[left]{{\small $\Im z$}};
\draw [very thin, dashed, color=darkgray,-] (0,0) -- (3.5,3.5) node[right]{$\frac{\pi}{4}$};
\draw [very thin, dashed, color=darkgray,-] (0,0) -- (-3.5,3.5) node[left]{$\frac{3\pi}{4}$};
\draw [fill=cyan, dashed,opacity=0.7] (0,0) -- (2.828427124,2.828427124) arc (45:0:4cm) -- (0,0);
\draw [fill=cyan, dashed,opacity=0.7] (0,0) -- (-2.828427124,2.828427124) arc (135:180:4cm) -- (0,0);
\draw [thick, color=red, decoration={markings, mark=at position 0.25 with {\arrow{>}}}, decoration={markings, mark=at position 0.75 with {\arrow{>}}}, postaction={decorate}] (-5,0) -- (5,0) node[above]{{\small $\mathbb{R}$}};
\draw [thick, color=red, decoration={markings, mark=at position 0.2 with {\arrow{>}}},
decoration={markings, mark=at position 0.8 with {\arrow{>}}}, postaction={decorate}] (-4.5,1.5) --
(4.5,1.5) node[right]{{\small $\mathbb{R}+\im\omega$}};
%
%
\end{tikzpicture}
\caption{The oriented jump contour $\Omega=\mathbb{R}\sqcup(\mathbb{R}+\im\omega)$ in RHP \ref{mast}. Compared to \cite[RHP $3.4$]{BB} we have chosen $\Gamma=\mathbb{R}+\im\omega$ with $\omega>0$ for concreteness.}
\label{figure4}
\end{figure}
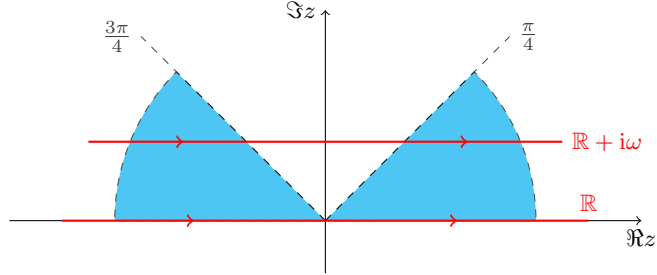
It was shown in \cite[Corollary $3.6$]{BB} that the above RHP \ref{mast} is uniquely solvable for every $(t,\gamma)\in\mathbb{R}\times[0,1]$ and its solution allows us to compute the resolvent $R=1+(1-G)^{-1}$ on $L^2(\Omega)$ in the form
\begin{equation}\label{a:3}
	R(\lambda,\mu)=\frac{{\bf F}^{\intercal}(\lambda){\bf G}(\mu)}{\lambda-\mu},\ \ \ \ \ \ \  {\bf F}(\lambda)={\bf N}_{\pm}(\lambda){\bf f}(\lambda),\ \ \ {\bf G}(\mu)=\big({\bf N}_{\pm}^{\intercal}(\mu)\big)^{-1}{\bf g}(\mu),\ \ \ \ \ \lambda,\mu\in\Omega.
\end{equation}
In order to solve RHP \ref{mast} asymptotically as $t\rightarrow-\infty$ with $\gamma\in[0,1)$ we first collapse the two jump contours in Figure \ref{figure4} and thus define
\begin{equation}\label{a:4}
	{\bf S}(z;t,\gamma):={\bf N}(z;t,\gamma)\begin{cases}\begin{bmatrix}1 & 0\\
	-\im\sqrt{\gamma}\,\e^{-\frac{1}{4}z^2+\im tz} & 1\end{bmatrix},&\Im z\in(0,\omega)\smallskip\\
	\mathbb{I},&\textnormal{else}
	\end{cases}.
\end{equation}
This leads us to the problem summarized below.
\begin{problem} For any $(t,\gamma)\in\mathbb{R}\times[0,1]$, the function ${\bf S}(z)={\bf S}(z;t,\gamma)\in\mathbb{C}^{2\times 2}$ defined in \eqref{a:4} satisfies
\begin{enumerate}
	\item[(1)] ${\bf S}(z)$ is analytic for $z\in\mathbb{C}\setminus\mathbb{R}$ and extends continuously to the closed upper and lower half-planes.
	\item[(2)] With ${\bf S}_{\pm}(z)=\lim_{\epsilon\downarrow 0}{\bf S}(z\pm\im\epsilon),z\in\mathbb{R}$, we have
	\begin{equation*}
		{\bf S}_+(z)={\bf S}_-(z)\begin{bmatrix}1-\gamma\e^{-\frac{1}{2}z^2} & -\im\sqrt{\gamma}\,\e^{-\frac{1}{4}z^2-\im tz}\smallskip\\
		-\im\sqrt{\gamma}\,\e^{-\frac{1}{4}z^2+\im tz} & 1\end{bmatrix},\ \ z\in\mathbb{R}.
	\end{equation*}
	\item[(3)] As $z\rightarrow\infty$,
	\begin{equation*}
		{\bf S}(z)=\mathbb{I}+\mathcal{O}\big(z^{-1}\big).
	\end{equation*}
\end{enumerate}
\end{problem}
Observe that \eqref{a:4} relates to the solution of RHP \ref{master} via the simple identity ${\bf S}(z;t,\gamma)={\bf Y}(z;\frac{t}{2},\gamma)$. Next, fix $t<0$, and define
\begin{equation}\label{a:5}
	{\bf T}(z;t,\gamma):={\bf S}(-zt;t,\gamma)\e^{g(z;t,\gamma)\sigma_3},\ \ z\in\mathbb{C}\setminus\mathbb{R},
\end{equation}
with the $g$-function from \cite[$(5.4)$]{BB}, i.e.
\begin{equation*}
	g(z)\equiv g(z;t,\gamma):=\frac{1}{2\pi\im}\int_{-\infty}^{\infty}\frac{h(s;t,\gamma)}{s-z}\,\d s,\ \ z\in\mathbb{C}\setminus\mathbb{R},
\end{equation*}
where $h(s;t,\gamma):=-\ln(1-\gamma\e^{-\frac{1}{2}t^2s^2})$ is H\"older continuous in $s\in\mathbb{R}$ for every $\gamma\in[0,1)$. Thus, by the standard Plemelj-Sokhotski formula, we arrive at the following problem:
\begin{problem} For any $(t,\gamma)\in(-\infty,0)\times[0,1)$, the function ${\bf T}(z)={\bf T}(z;t,\gamma)\in\mathbb{C}^{2\times 2}$ defined in \eqref{a:5} satisfies
\begin{enumerate}
	\item[(1)] ${\bf T}(z)$ is analytic for $z\in\mathbb{C}\setminus\mathbb{R}$ and extends continuously to the closed upper and lower half-planes.
	\item[(2)] The boundary values ${\bf T}_{\pm}(z)=\lim_{\epsilon\downarrow 0}{\bf T}(z\pm\im\epsilon)$ are related by the jump condition
	\begin{equation*}
		{\bf T}_+(z)={\bf T}_-(z)\begin{bmatrix}1 & -\im\sqrt{\gamma}f_1(z;t,\gamma)\,\e^{-2g_+(z;t,\gamma)}\smallskip\\
		-\im\sqrt{\gamma}f_2(z;t,\gamma)\,\e^{2g_-(z;t,\gamma)} & 1-\gamma\e^{-\frac{1}{2}t^2z^2}\end{bmatrix},\ \ z\in\mathbb{R},
	\end{equation*}
	with
	\begin{equation}\label{a:6}
		f_k(z)\equiv f_k(z;t,\gamma):=\frac{\e^{-t^2(\frac{1}{4}z^2+(-1)^k\im z)}}{1-\gamma\e^{-\frac{1}{2}t^2z^2}},\ \ z\in\mathbb{R},\ \ k=1,2.
	\end{equation}
	\item[(3)] As $z\rightarrow\infty$,
	\begin{equation*}
		{\bf T}(z)=\mathbb{I}+\mathcal{O}\big(z^{-1}\big).
	\end{equation*}
\end{enumerate}
\end{problem}
Note that $g_{\pm}(z)$ admit analytic continuation to the full upper, resp. lower half-planes, but $f_k(z;t,\gamma)$ does not, compare \cite[Proposition $5.1$]{BB} and \cite[Figure $11$]{BB}. However, if we define the region 
\begin{equation*}
	\pi_{t\gamma}:=\left\{z\in\mathbb{C}:\,|\Im z|<\frac{\sqrt{-2\ln\gamma}}{|t|}\right\}\cup\left\{z\in\mathbb{C}:\,\,\big|\textnormal{arg}(z-1)\big|\leq\frac{\pi}{4}\right\}\cup\left\{z\in\mathbb{C}:\,\frac{3\pi}{4}\leq\textnormal{arg}(z+1)\leq\frac{5\pi}{4}\right\},
\end{equation*}
then the denominators in \eqref{a:6} do not vanish for $z\in\pi_{t\gamma}$ and so $f_k(z)$ admits analytic continuation to $\pi_{t\gamma}$. Thus, using the matrix factorization
\begin{equation*}
	\begin{bmatrix}1 & -\im\sqrt{\gamma}f_1(z)\,\e^{-2g_+(z)}\smallskip\\
		-\im\sqrt{\gamma}f_2(z)\,\e^{2g_-(z)} & 1-\gamma\e^{-\frac{1}{2}t^2z^2}\end{bmatrix}=\begin{bmatrix}1 & 0\smallskip\\
		-\im\sqrt{\gamma}\,f_2(z)\e^{2g_-(z)} & 1\end{bmatrix}\begin{bmatrix}1 & -\im\sqrt{\gamma}\,f_1(z)\e^{-2g_+(z)}\smallskip\\
		0 & 1\end{bmatrix},\ z\in\mathbb{R},
\end{equation*}
the following transformation is well-defined, 
\begin{equation}\label{a:7}
	{\bf M}(z;t,\gamma):={\bf T}(z;t,\gamma)\begin{cases}\begin{bmatrix}1 & \im\sqrt{\gamma}\,f_1(z)\e^{-2g(z)}\\
	0 & 1\end{bmatrix},& z\in\Omega_{1t\gamma}\smallskip\\
	\begin{bmatrix}1 & 0\\
	-\im\sqrt{\gamma}\,f_2(z)\e^{2g(z)} & 1\end{bmatrix},&z\in\Omega_{2t\gamma}\end{cases},
\end{equation}
where the domains $\Omega_{kt\gamma}$ are shown in Figure \ref{figure5}. Subsequently we arrive at the problem below.
\begin{problem}\label{split} For every $(t,\gamma)\in(-\infty,0)\times[0,1)$, the function ${\bf M}(z)={\bf M}(z;t,\gamma)\in\mathbb{C}^{2\times 2}$ defined in \eqref{a:7} has the following properties
\begin{enumerate}
	\item[(1)] ${\bf M}(z)$ is analytic for $z\in\mathbb{C}\setminus\Sigma_{{\bf M}}$, see Figure \ref{figure5} for the oriented jump contour $\Sigma_{\bf M}$.
	\begin{figure}[tbh]
\begin{tikzpicture}[xscale=0.7,yscale=0.4]
\draw [->] (-5,0) -- (5,0) node[below]{{\small $\Re z$}};
\draw [->] (0,-4) -- (0,4) node[left]{{\small $\Im z$}};

\draw [thick, color=red, decoration={markings, mark=at position 0.2 with {\arrow{>}}}, decoration={markings, mark=at position 0.8 with {\arrow{>}}}, postaction={decorate}] (-4.5,2) -- (-2,2);
\draw [thick, color=red, decoration={markings, mark=at position 0.5 with {\arrow{>}}},  postaction={decorate}] (-2,2) -- (-1,1);
\draw [thick, color=red, decoration={markings, mark=at position 0.2 with {\arrow{>}}}, decoration={markings, mark=at position 0.8 with {\arrow{>}}}, postaction={decorate}] (-1,1) -- (1,1);
\draw [thick, color=red, decoration={markings, mark=at position 0.5 with {\arrow{>}}},  postaction={decorate}] (1,1) -- (2,2);
\draw [thick, color=red, decoration={markings, mark=at position 0.2 with {\arrow{>}}}, decoration={markings, mark=at position 0.8 with {\arrow{>}}}, postaction={decorate}] (2,2) -- (4.5,2);

\draw [thick, color=red, decoration={markings, mark=at position 0.2 with {\arrow{>}}}, decoration={markings, mark=at position 0.8 with {\arrow{>}}}, postaction={decorate}] (-4.5,-2) -- (-2,-2);
\draw [thick, color=red, decoration={markings, mark=at position 0.5 with {\arrow{>}}},  postaction={decorate}] (-2,-2) -- (-1,-1);
\draw [thick, color=red, decoration={markings, mark=at position 0.2 with {\arrow{>}}}, decoration={markings, mark=at position 0.8 with {\arrow{>}}}, postaction={decorate}] (-1,-1) -- (1,-1);
\draw [thick, color=red, decoration={markings, mark=at position 0.5 with {\arrow{>}}},  postaction={decorate}] (1,-1) -- (2,-2);
\draw [thick, color=red, decoration={markings, mark=at position 0.2 with {\arrow{>}}}, decoration={markings, mark=at position 0.8 with {\arrow{>}}}, postaction={decorate}] (2,-2) -- (4.5,-2);
\node [right] at (-4,1) {{\small $\Omega_{1t\gamma}$}};
\node [right] at (3,-1) {{\small $\Omega_{2t\gamma}$}};
\end{tikzpicture}
\caption{The oriented jump contour $\Sigma_{{\bf M}}$ in RHP \ref{split} in red. We fix $z_1=-2+\im,z_2=-1+\im\delta_{t\gamma},z_3=1+\im\delta_{t\gamma}$ and $z_4=2+\im$ as location of the four vertices in the upper half-plane. The ones in the lower half-plane are their complex conjugates and we choose $\delta_{t\gamma}:=\min\big\{\sqrt{-\ln\gamma}/|t|,\frac{1}{2}\big\}>0$. 
This way the contour $\Sigma_{\bf M}$ is fully contained in $\pi_{t\gamma}$ and thus $f_k(z)$ analytic for $z\in\Omega_{kt\gamma}$.}
\label{figure5}
\end{figure}
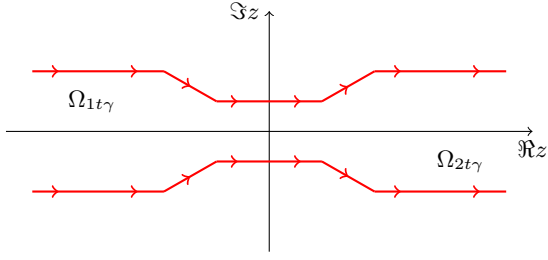
	\item[(2)] The boundary values ${\bf M}_{\pm}(z),z\in\Sigma_{\bf M}$ satisfy
	\begin{equation*}
		{\bf M}_+(z)={\bf M}_-(z)\begin{bmatrix}1 & -\im\sqrt{\gamma}\,f_1(z)\e^{-2g(z)}\\
		0 & 1\end{bmatrix},\ \ z\in\Sigma_{\bf M}\cap\{z\in\mathbb{C}:\,\Im z>0\},
	\end{equation*}
	and
	\begin{equation*}
		{\bf M}_+(z)={\bf M}_-(z)\begin{bmatrix}1 & 0\\
		-\im\sqrt{\gamma}\,f_2(z)\e^{2g(z)} & 1\end{bmatrix},\ \ z\in\Sigma_{\bf M}\cap\{z\in\mathbb{C}:\,\Im z<0\}.
	\end{equation*}
	\item[(3)] As $z\rightarrow\infty$,
	\begin{equation*}
		{\bf M}(z)=\mathbb{I}+\mathcal{O}\big(z^{-1}\big).
	\end{equation*}
\end{enumerate}
\end{problem}
The important properties of RHP \ref{split} are summarized in the following small norm estimates for its jump matrix ${\bf G}_{\bf M}(z;t,\gamma)$, see condition (2) in RHP \ref{split} above.
\begin{prop}\label{smallnorm} There exist constants $t_0,c>0$ such that
\begin{equation*}
	\|{\bf G}_{\bf M}(\cdot;t,\gamma)-\mathbb{I}\|_{L^{\infty}(\Sigma_{\bf M},|\d z|)}\leq c\sqrt{\gamma}\,\frac{\e^{\frac{1}{4}t^2\delta_{t\gamma}^2-t^2\delta_{t\gamma}}}{|1-\gamma\e^{\frac{1}{2}t^2\delta_{t\gamma}^2}|},\ \ \ \ \|{\bf G}_{\bf M}(\cdot;t,\gamma)-\mathbb{I}\|_{L^1(\Sigma_{\bf M},|\d z|)}\leq c\,\frac{\sqrt{\gamma}}{|t|}\,\frac{\e^{\frac{1}{4}t^2\delta_{t\gamma}^2-t^2\delta_{t\gamma}}}{|1-\gamma\e^{\frac{1}{2}t^2\delta_{t\gamma}^2}|},
\end{equation*}
and
\begin{equation*}
	\|{\bf G}_{\bf M}(\cdot;t,\gamma)-\mathbb{I}\|_{L^2(\Sigma_{\bf M},|\d z|)}\leq c\sqrt{\frac{\gamma}{|t|}}\,\frac{\e^{\frac{1}{4}t^2\delta_{t\gamma}^2-t^2\delta_{t\gamma}}}{|1-\gamma\e^{\frac{1}{2}t^2\delta_{t\gamma}^2}|},
\end{equation*}
hold true for all $(-t)\geq t_0$ and any $\gamma\in[0,1)$.
\end{prop}
\begin{proof} For $z$ on the components of $\Sigma_{\bf M}$ which extend to infinity we have
\begin{equation*}
	\|{\bf G}_{\bf M}(z;t,\gamma)-\mathbb{I}\|\leq c\sqrt{\gamma}\,\e^{-\frac{7}{4}t^2},\ \ c>0\ \ \textnormal{universal},
\end{equation*}
and thus sub leading corrections. The jumps on the two remaining horizontal segments are estimated as in \cite[Proposition $5.5$]{BB} yielding upper bounds as stated in the Proposition. Finally, for the four slanted segments we consider, say, $z=z(\lambda)=\lambda-2+\im(1-\lambda+\lambda\delta_{t\gamma}),\lambda\in[0,1]$ and obtain
\begin{equation*}
	\big|f_1(z(\lambda))\e^{-2g(z(\lambda))}\big|\leq c\,\e^{-\frac{1}{4}t^2(1-\delta_{t\gamma}^2)},\ \ c>0\ \ \textnormal{universal},
\end{equation*}
i.e. another sub leading contribution. This concludes our proof.
\end{proof}
On compact subsets of $[0,1)\ni\gamma$, and this is after all the situation we are considering in Section \ref{nonsteep}, Proposition \ref{smallnorm} yields existence of $t_0=t_0(\gamma)>0$ and a universal $c>0$ such that for all $(-t)\geq t_0$,
\begin{equation}\label{a:8}
	\|{\bf G}_{\bf M}(\cdot;t,\gamma)-\mathbb{I}\|_{L^{\infty}(\Sigma_{\bf M},|\d z|)}\leq c\,s(\gamma)\e^{t\sqrt{-\ln\gamma}},\ \ \ \|{\bf G}_{\bf M}(\cdot;t,\gamma)-\mathbb{I}\|_{L^1(\Sigma_{\bf M},|\d z|)}\leq c\,s(\gamma)|t|^{-1}\e^{t\sqrt{-\ln\gamma}},
\end{equation}
\begin{equation}\label{a:9}
	\|{\bf G}_{\bf M}(\cdot;t,\gamma)-\mathbb{I}\|_{L^2(\Sigma_{\bf M},|\d z|)}\leq c\,s(\gamma)|t|^{-\frac{1}{2}}\e^{t\sqrt{-\ln\gamma}};\ \ \ \ \ \ \ s(\gamma):=\frac{\gamma^{\frac{1}{4}}}{1-\sqrt{\gamma}}.
\end{equation}
We thus obtain
\begin{theo}\label{better} There exists a constant $c>0$ such that for every fixed $\gamma\in[0,1)$, there exist $t_0=t_0(\gamma)>0$ so that RHP \ref{split} is uniquely solvable in $L^2(\Sigma_{\bf M})$ for all $(-t)\geq t_0$. The solution can be computed iteratively from the integral equation
\begin{equation*}
	{\bf M}(z)=\mathbb{I}+\frac{1}{2\pi\im}\int_{\Sigma_{\bf M}}{\bf M}_-(\lambda)\big({\bf G}_{\bf M}(\lambda)-\mathbb{I}\big)\frac{\d\lambda}{\lambda-z},\ \ z\in\mathbb{C}\setminus\Sigma_{\bf M},
\end{equation*}
using the estimate
\begin{equation*}
	\|{\bf M}_-(\cdot;t,\gamma)-\mathbb{I}\|_{L^2(\Sigma_{\bf M})}\leq\,cs(\gamma)|t|^{-\frac{1}{2}}\e^{t\sqrt{-\ln\gamma}}\ln|t|\ \ \ \ \ \forall\,(-t)\geq t_0.
\end{equation*}
\end{theo}
\begin{proof} As in \cite[Section $5$]{BB}, the general theory of \cite{DZ} is not directly applicable in the analysis of RHP \ref{split} since our contour $\Sigma_{\bf M}$ varies with $t$. Still, using the arguments outlined in \cite[Appendix A]{BB} one directly proves from \eqref{a:8}, \eqref{a:9} that the Neumann series
\begin{equation*}
	\rho(z)=\mathbb{I}+\sum_{k=1}^{\infty}\rho_k(z);\ \ \ \rho_k(z)=\frac{1}{2\pi\im}\int_{\Sigma_{\bf M}}\rho_{k-1}(\lambda)\big({\bf G}_{\bf M}(\lambda)-\mathbb{I}\big)\frac{\d\lambda}{\lambda-z_-},\ \ \ z\in\Sigma_{\bf M},\ k\in\mathbb{Z}_{\geq 1}
\end{equation*}
with $\rho_0(z)=\mathbb{I}$ converges in $L^2(\Sigma_{\bf M},|\d z|)$ for sufficiently large $(-t)$ and any fixed $\gamma\in[0,1)$. Furthermore, adapting the arguments of \cite[page $497$]{BB} to our ${\bf G}_{\bf M}(z;t,\gamma)$ and $\Sigma_{\bf M}$ in RHP \ref{split}, we find from \eqref{a:8}, \eqref{a:9} that there exists $c>0$ such that for every fixed $\gamma\in[0,1)$ there exists $t_0=t_0(\gamma)>0$ so that
\begin{equation*}
	\|\rho_k\|_{L^2(\Sigma_{\bf M},|\d z|)}\leq \left(c\,s(\gamma)\e^{t\sqrt{-\ln\gamma}}\ln|t|\right)^k|t|^{-\frac{1}{2}},\ \ \ \forall\,(-t)\geq t_0,\ k\in\mathbb{Z}_{\geq 1}.
\end{equation*}
Hence, $\mathbb{I}+\sum_{k=1}^{\infty}\rho_k(z)$ converges in $L^2(\Sigma_{\bf M},|\d z|)$ for sufficiently large $(-t)\geq t_0$ and any fixed $\gamma\in[0,1)$. But its sum $\rho(z)$ satisfies the singular integral equation
\begin{equation*}
	\rho(z)=\mathbb{I}+\frac{1}{2\pi\im}\int_{\Sigma_{\bf M}}\rho(\lambda)\big({\bf G}_{\bf M}(\lambda)-\mathbb{I}\big)\frac{\d\lambda}{\lambda-z_-},\ \ z\in\Sigma_{\bf M},
\end{equation*}
by construction and so
\begin{equation*}
	{\bf M}(z):=\mathbb{I}+\frac{1}{2\pi\im}\int_{\Sigma_{\bf M}}\rho(\lambda)\big({\bf G}_{\bf M}(\lambda)-\mathbb{I}\big)\frac{\d\lambda}{\lambda-z},\ \ \ z\in\mathbb{C}\setminus\Sigma_{\bf M},
\end{equation*}
yields ${\bf M}_-(z)=\rho(z)$ for $z$ on any of the ten straight line segments which comprise $\Sigma_{\bf M}$. This completes our proof.
\end{proof}
We conclude this section with the following Corollary to Theorem \ref{better}.
\begin{cor} For any fixed $\gamma\in[0,1)$, as $t\rightarrow-\infty$,
\begin{equation}\label{a:10}
	{\bf M}(z)=\mathbb{I}+\mathcal{O}\left(\frac{s(\gamma)\,t^{-1}\e^{t\sqrt{-\ln\gamma}}}{\textnormal{dist}(\Sigma_{\bf M},z)}\right),\ \ \ \ \ \ {\bf M}'(z)=\mathcal{O}\left(\frac{s(\gamma)\,t^{-1}\e^{t\sqrt{-\ln\gamma}}}{\textnormal{dist}^2(\Sigma_{\bf M},z)}\right).
\end{equation}

\end{cor}
\end{appendix}

\end{document}